\documentclass[a4paper,11pt,twoside,reqno]{amsart}
\usepackage{preamble}
\usepackage{amsmath,amssymb,amsfonts}
\usepackage{algorithmic}
\usepackage{graphicx}
\usepackage{textcomp}
\usepackage{xcolor}
\usepackage{caption}
\usepackage{subcaption}
\usepackage{xurl}

\usepackage{listings}
\def\BibTeX{{\rm B\kern-.05em{\sc i\kern-.025em b}\kern-.08em
    T\kern-.1667em\lower.7ex\hbox{E}\kern-.125emX}}

\begin{document}

\title[]{Understanding Disclosure Risk in Differential Privacy with Applications to Noise Calibration and Auditing\\{\small \textit{E\lowercase{xtended version}}}}

\author{Patricia Guerra-Balboa}
\email{patricia.balboa@kit.edu}
\address{Karlsruhe Institute of Technology, KASTEL SRL}
\author{Annika Sauer}
\email{annika.sauer@student.kit.edu}
\author{Héber H. Arcolezi}
\address{Inria Centre at the University Grenoble Alpes and ÉTS Montréal}
\email{heber.hwang-arcolezi@etsmtl.ca}
\author{Thorsten Strufe}
\email{thorsten.strufe@kit.edu}


\begin{abstract}
Differential Privacy (DP) is widely adopted in data management systems to enable data sharing with formal disclosure guarantees. A central systems challenge is understanding how DP noise translates into effective protection against inference attacks, since this directly determines achievable utility.
Most existing analyses focus only on membership inference---capturing only a threat---or rely on reconstruction robustness (ReRo). However, under realistic assumptions, we show that ReRo can yield misleading risk estimates and violate claimed bounds, limiting their usefulness for principled DP calibration and auditing.

This paper introduces reconstruction advantage, a unified risk metric that consistently captures risk across membership inference, attribute inference, and data reconstruction. We derive tight bounds that relate DP noise to adversarial advantage and characterize optimal adversarial strategies for arbitrary DP mechanisms and attacker knowledge. These results enable risk-driven noise calibration and provide a foundation for systematic DP auditing.
We show that reconstruction advantage improves the accuracy and scope of DP auditing and enables more effective utility-privacy trade-offs in DP-enabled data management systems.

\end{abstract}

\maketitle
\section{Introduction}

Differential Privacy (DP)~\cite{dwork2006calibrating} and its distributed variant, local DP (LDP), have emerged as the de facto standard to mitigate privacy risk---that is, the extent to which a learning process allows sensitive information about participants to be inferred. DP aims to make participation as safe as not participating~\cite{Dwork2006Differential}, and its privacy-utility trade-off is governed by the privacy budget $\varepsilon$ (smaller values provide stronger guarantees) and by $\delta$, which captures the probability mass of outcomes in which the guarantee may fail, weighted by the severity of their deviation from $\varepsilon$~\cite{Meiser2018Cryptology}. Despite this solid theoretical foundation, a central practical question remains: How do these formal parameters, especially $\varepsilon$, translate into concrete protection against real-world attacks?~\cite{Nanayakkara2023What} This question is critical for calibrating $\varepsilon$: if set too high, sensitive information may be exposed; if too low, utility is unnecessarily compromised. Furthermore, understanding this relationship is essential for DP auditing, which aims to empirically estimate privacy~\cite{Jagielski2020Auditing}, test the tightness of DP mechanisms~\cite{Nasr2021Adversary}, and detect bugs~\cite{Tramer2022Debugging}.

Motivated by its applications in noise calibration and auditing, there is growing interest in the data management community in risk assessment for DP mechanisms~\cite{Bernau2021Quantifying,Arcolezi2023On,Cormode2025Synthetic,carey2023measuring}. Significant progress has been made in connecting DP to the risk of \textit{membership inference attacks} (MIAs)~\cite{Bernau2021Quantifying,Yeom2017Privacy, Erlingsson2019That, Humphries2023Investigating}, even enabling direct noise calibration for desired MIA risk levels~\cite{kulynych2024attack} without explicitly choosing $\varepsilon$. However, MIAs capture only one aspect of privacy risk and may be less relevant in deployments such as census data releases. In particular, \textit{attribute inference attacks} (AIAs)~\cite{Yeom2017Privacy}, which can expose sensitive information even when membership is public~\cite{Balle2022Reconstructing}, remain less understood. Recently, \textit{data reconstruction attacks} (DRAs)~\cite{Balle2022Reconstructing} were proposed as a unifying framework subsuming both MIAs and AIAs, while also accounting for partial or imperfect reconstruction, e.g., revealing a car’s license plate may suffice to compromise privacy even if the background is inaccurate.

\citeauthor{Balle2022Reconstructing}~\cite{Balle2022Reconstructing} introduced the first metric for DRAs, \textit{reconstruction robustness} (ReRo), providing a pioneering unified view of DP attack resilience. 
ReRo was foundational, but has limitations as a comprehensive adversarial metric. 
First, ReRo and existing bounds~\cite{Hayes2023Bounding,Balle2022Reconstructing} assume attackers have no target-specific auxiliary knowledge, ignoring partial information such as demographic attributes or social media data—information that real-world attacks often exploit~\cite{Sweeney2000Simple, Montjoye2013Unique, Narayanan2008Robust}. 
We empirically confirm this limitation: when target-specific auxiliary information is available, the empirical ReRo exceeds the existing ReRo bounds (see~\Cref{fig:dp-sgd_mnist}).
Second, ReRo is a success probability, which penalizes mechanisms for providing global statistical knowledge---the end goal of data release---and incorrectly accounts for success from background knowledge or statistical imputation as participation risk~\cite{Bun2021Statistical,Kifer2022Bayesian}, leading to unnecessary utility loss when used for noise calibration (\Cref{fig:laplace_accu}).

We address such limitations by introducing \textit{reconstruction advantage} (RAD), which extends advantage metrics to the unifying DRA framework. RAD overcomes ReRo's limitations, naturally incorporating auxiliary knowledge and avoiding risk overestimation. We establish tight bounds linking DP noise to RAD, enabling noise injection calibrated to a participant's true risk of information disclosure. Specifically, we provide: (i) a worst-case bound independent of the attacker’s auxiliary knowledge (\Cref{th:dp_implies_aux-urero}), and (ii) an auxiliary-dependent, universally tight bound (\Cref{th:optimal_bound}). To assess tightness, we construct and prove the optimal attack strategy for any reconstruction goal, auxiliary knowledge, and mechanism---which also serves as a practical tool for DP auditing.

\Cref{th:optimal_bound} is universally tight and cannot be further improved. However, it requires full knowledge of the mechanism $\mathcal{M}$, limiting its applicability in auditing external software. While \Cref{th:dp_implies_aux-urero} can serve as a fallback in such scenarios, it may strongly overestimate risk when no auxiliary information is available.  To address this, we provide closed-form, black-box upper bounds for RAD without auxiliary knowledge (i.e., when the entire target record is considered secret, as in~\cite{Arcolezi2024Revealing,Balle2022Reconstructing,Hayes2023Bounding}) and for the case of perfect reconstruction, which is particularly relevant for categorical data where sensitive attributes (e.g., diseases, political opinions, or religious beliefs) cannot be partially reconstructed~\cite{Fredrikson2015Model, Fredrikson2014Privacy}. All our bounds substantially reduce the required noise compared to existing ReRo bounds, and we validate these improvements experimentally.

These results provide the theoretical foundation for practical DP auditing. Modern DP systems deployed in industry~\cite{Erlingsson2014RAPPOR,Lu2024Eureka}, government~\cite{Abowd2018CensusBureau}, and data-processing pipelines~\cite{McSherry2009PINQ} still lack general-purpose tools for quantifying real-world privacy leakage. Existing auditing tools either focus on a narrow attack class (often MIAs)~\cite{Jagielski2020Auditing,Tramer2022Debugging,Nasr2021Adversary,Mahloujifar2024Auditing,Arcolezi2024Revealing} or rely on learning-based strategies requiring extensive tuning without mechanism-independent guarantees~\cite{Lu2024Eureka}. RAD fills this gap, offering a principled, mechanism-agnostic characterization of reconstruction risk. Building on our novel bounds, we introduce a RAD-based auditing framework that generalizes beyond prior tools~\cite{Ding2018Detecting,Arcolezi2024Revealing}, capturing all reconstruction risks and providing more accurate, actionable privacy assessments. While our auditing framework is general in scope, in this paper we instantiate it for LDP and address key limitations of the state-of-the-art tool, \textsc{\textsc{LDP Auditor}}~\cite{Arcolezi2024Revealing}. Unlike \textsc{\textsc{LDP Auditor}}, which relies on perfect reconstruction without target-specific auxiliary knowledge---and thus misses important threats such as AIAs---our method is both more general and produces \emph{tighter empirical estimates} of the privacy budget for all the tested LDP mechanisms as demonstrated in our empirical study (see~\Cref{fig:audit_porto}).

Our contributions are summarized as follows:
\begin{itemize}
    \item We empirically show that ReRo and its existing bounds fail to account for imputation-based success and target-specific auxiliary knowledge, limiting applicability.
    \item We introduce \textit{Reconstruction Advantage (RAD)} as a consistent, unifying risk metric that naturally incorporates auxiliary knowledge.
    \item We establish tight worst-case and auxiliary-dependent bounds for RAD, along with black-box bounds for attackers lacking auxiliary knowledge.
    \item We construct the optimal attack strategy for any reconstruction goal, mechanism, and prior distribution, proving its optimality and demonstrating empirical utility for auditing.
    \item We propose a RAD-based DP auditing framework that provides broader threat analyses and more accurate privacy-budget estimates than existing LDP auditing techniques.
\end{itemize}

This is an extended version of the paper accepted in the Proceedings of the VLDB Endowment
(PVLDB), 2026. The code used for our experiments is accessible in \url{https://github.com/PatriciaBalboaKIT/Understanding-Risk-in-DP}.

%------------------------------------------------------------
\section{Background}\label{sec:background}
In this section, we introduce the relevant concepts for this work and present the notation used throughout the manuscript.

\subsection{Differential Privacy}\label{sec:background_dp}
We assume each record $z\in\Z$ to be drawn independently from an underlying prior distribution $\Z\sim \pi$. Let $\mathcal{D}(\Theta)$ denote the space of probability distributions over the output space $\Theta$. We consider a mechanism $\M \colon \mathcal{Z}^n \to \mathcal{D}(\Theta)$ which, given an input database $D \in \Z^n$, produces a global output (e.g., an aggregate statistic or a trained model) $\theta \in \Theta$ with probability/density function $p_{\M}(\theta \mid D)$.
In this context, DP is formalized as follows:

\begin{definition}[\cite{dwork2006calibrating}]\label{def:dp}
    A mechanism $\M\colon\Z^n \rightarrow \mathcal{D}(\Theta)$ is $(\varepsilon, \delta)$-differentially private if for all $S\subseteq \Theta$  and for every pair of datasets $D_0,D_1 \in \Z^n$ such that $d_{H}(D_0,D_1)\leq 1$: \[
    \Pr(\M(D_0)\in S) \leq \e^{\varepsilon}  \Pr(\M(D_1)\in S) + \delta \] 
    where $d_{H}(D_0,D_1)$ denotes the Hamming distance~\cite{mackay2003information}.
\end{definition}

If $\delta=0$ we speak of \textit{pure} DP ($\varepsilon$-DP). If $n=1$, i.e., $\M$ takes as input a single data record $z\in\Z$, we obtain \textit{Local Differential Privacy} (LDP).
LDP is a rigorous and increasingly relevant privacy model in which data is randomized on the client side before being transmitted to a data collector~\cite{Dwork2014Algorithmic}.
Consequentially, it is especially suitable for privacy-sensitive applications such as telemetry and location-based services where no trusted data curator is considered~\cite{Erlingsson2014RAPPOR,Lu2024Eureka}.

The privacy budget $\varepsilon$ determines how closely the probabilities of observing the same output on databases $D_0$ and $D_1$ must align, hence bounding their statistical ``indistinguishability''. A smaller $\varepsilon$ provides stronger privacy guarantees but typically comes at the cost of utility~\cite{Dwork2014Algorithmic}. The parameter $\delta$ allows certain violations of $\varepsilon$-DP while characterizing how likely such failures are to occur and the degree of such failures. Consequently, we aim to parameterize the attack performance based on the privacy parameters.

Many real-world deployments apply multiple DP mechanisms sequentially~\cite{bu2008privacy,cunningham2021real}. By DP’s adaptative composition property, the total privacy loss is determined by the parameters of the individual mechanisms~\cite{kairouz15composition}. 
Formally, given $[T]=\{1,\dots,T\}$, for each $i \in [T]$, let $\overline{\Theta}_{i-1} = \prod_{j=1}^{i-1} \Theta_j$ denote the space of previous outputs, and define $\mathcal{M}_i \colon \mathcal{Z}^n \times \overline{\Theta}_{i-1} \to \Theta_i$. The \emph{$T$-fold composed} mechanism is $\mathcal{M}(D) = (\mathcal{M}_1(D), \mathcal{M}_2(D, Y_1), \ldots, \mathcal{M}_T(D, Y_{T-1}))$, where $Y_i = (\mathcal{M}_1(D), \ldots, \mathcal{M}_i(D, Y_{i-1}))$ denotes the first $i$ outputs. \citeauthor{dwork2006calibrating}~\cite{dwork2006calibrating} established the first general bound on the privacy loss under $T$-fold adaptive composition: Composing $(\varepsilon,\delta)$-DP mechanisms yields to $(T\varepsilon, T\delta)$-DP. Subsequent refinements led to the tighter composition bounds as presented in~\cite{kairouz15composition}.

\subsection{Differential Privacy and Attack Resilience}\label{sec:background_attacks}
Following previous work we consider for any target record $z$ an \textit{informed adversary}~\cite{Balle2022Reconstructing} with access to:
 the fixed dataset $D_{-} = D \setminus \{z\}$,  
 the distribution of data records $\pi$,  
 the output $\theta$ of the model trained on $D_{z}=D_{-}\cup\{z\}$,  
 the mechanism $\mathcal{M}$, and  
 optional target-specific auxiliary knowledge $a(z)$ about target record $z$.  We adopt this adversary model because, under the assumption that records are independently drawn from $\pi$, bounding the performance of such an attacker also bounds the performance of any attacker with less information~\cite{Balle2022Reconstructing}.

Our analysis focuses on DRAs, where the adversary's goal is to correctly reconstruct completely or partially the target record $z$, potentially given auxiliary knowledge $a(z)\in aux$ about the target. DRAs cover AIAs and MIAs as particular cases~\cite{Balle2022Reconstructing}:
In an MIA, the attacker knows the entire target record $a(z)=z$ and seeks only to infer its participation in the dataset. In an AIA, records are structured as $z=(x,y)$, where $a(z)=x$ is considered public and the attacker aims to perfectly reconstruct the sensitive attribute $y$. More generally, in a DRA setting, it is natural to assume access to target-specific auxiliary knowledge. For example, when reconstructing a license plate number from a target’s car image, the attacker may already know the color of the car.  
Hence, DRAs cover the broad range of commonly discussed privacy risks, including MIAs and AIAs as a particular instance~\cite{Balle2022Reconstructing}. Formally, a DRA, denoted by $A\colon \Theta\times aux\to\Z$ uses the output of a DP mechanism $\theta\sim\M(D)$ and the target auxiliary information to produce a candidate $\Tilde{z}=A(\theta,a(z))$. Note that, in case of composing several mechanisms, we consider the final output after the whole process. 

The attack is considered successful if the output is similar enough (according to a success threshold $\eta$) to the real input $z$: $\ell(\Tilde{z},z)\leq \eta$. The error function $\ell$ depends on the context, for instance, in a classic AIA, given $z=(x,y)$ we define $\phi(z)=y$ and
$
 \ell(\Tilde{z},z)=
     0$ if $\phi(\Tilde{z})=\phi(z)$ and equals one otherwise.
In a MIA, $\ell$ is the characteristic function such that $\ell(\Tilde{z},z)=0$ when $\Tilde{z}=z$ and equals one otherwise. However, it may be sensitive enough to partially reconstruct the target. For instance, in the image domain, even if not all pixels are correct, we may gather sensitive information such as the action performed in the image. Consequently, $\ell$ may be chosen as an image-specific metric, such as the Learned Perceptual Image Patch Similarity (LPIPS)~\cite{Balle2022Reconstructing}.   
Given the error function $\ell$ and the threshold $\eta$, we define the \textit{success set} of a target $z$ as
\[S_{\eta}(z)=\{z'\in\Z\colon \ell(z,z')\leq \eta\}.\]

After defining a DRA, the question of how to evaluate its performance arises. For the particular cases of AIA and MIAs, the current literature~\cite{Guerra2024Analysis,Yeom2017Privacy} agrees on the following metric:

\begin{definition}[Adapted from~\cite{Yeom2017Privacy}]\label{def:att_adv}
Given $\pi$ the distribution of data records and  $\M,\phi(z),a(z),A$ as defined above, the \textit{attribute advantage}, $\mathrm{Adv}_{AIA}$, is defined as
\[
\Pr_{\substack{z_0\sim\pi\\\theta\sim\M(D_{z_0})}}\left[A(\theta,a(z_0))=\phi(z_0)\right]-\Pr_{\substack{z_0,z_1\sim\pi\\\theta\sim\M(D_{z_1})}}\left[A(\theta,a(z_0))=\phi(z_0)\right].
\]
\end{definition}

The attribute advantage measures the adversary’s gain in correctly inferring a sensitive attribute 
$\phi(z)$ when a record is in the input dataset $z_0\in D$, compared to when it is drawn from the underlying distribution $\pi$. The second term in~\Cref{def:att_adv} corrects for cases where the attribute could be inferred even without the record being in the database (e.g., through imputation~\cite{Jayaraman2022Are}). In the context of MIAs, the advantage reduces to the true positive rate (TPR) minus the false positive rate (FPR), effectively discounting trivial attacks, such as always predicting ``member'', which achieve high success probability (the probability to correctly identify a member is one) without revealing any meaningful private information.\footnote{Strictly speaking, \Cref{def:att_adv} reduces to the true positive rate (TPR) minus the probability of identifying any arbitrary record as a member; by normalizing instead by the resampling probability, we obtain the TPR–FPR (see~\cite[Prop. 8.1]{Guerra2024Analysis} for details).}

The current proposed performance metric for general DRAs~\cite{Balle2022Reconstructing} does not define an advantage but instead only accounts for the success probability of an attack that has as input solely the output of the DP mechanism and the known dataset $D_{-}$, ignoring any possible target-specific auxiliary knowledge:
\begin{definition}[ReRo~\cite{Balle2022Reconstructing}]\label{def:rero}
    Let $\pi$ be a prior over $\mathcal{Z}$ and $\ell: \mathcal{Z} \times \mathcal{Z} \rightarrow \mathbb{R}_{\geq 0}$ a error function. Mechanism $\mathcal{M}: \mathcal{Z}^n \rightarrow \D(\Theta)$ is $(\eta, \gamma)$-reconstruction robust with respect to $\pi, \ell$ if for any dataset $D_- \in \mathcal{Z}^{n-1}$ and any reconstruction adversary $A: \Theta \rightarrow \mathcal{Z}$,  \[
    \Pr_{Z \sim \pi, \theta \sim \mathcal{M}(D_{Z}) } [\ell(Z, A(\theta)) \leq \eta] \leq \gamma.
    \]
\end{definition}

The first bound for ReRo under $\varepsilon$-DP was given by~\cite{Balle2022Reconstructing}: 
\begin{equation}\label{eq:balle_bound}
    \gamma\leq \kappa^{+}_{\pi,\ell}(\eta)\e^{\varepsilon},
\end{equation}
where $\kappa^{+}_{\pi,\ell}(\eta)=\sup_{z_0}\Pr_{Z\sim\pi}[\ell(z_0,Z)\leq \eta]$. Intuitively, $\kappa^{+}_{\pi,\ell}(\eta)$ represents the success probability of an oblivious attack that always selects the most likely reconstruction under the prior $\pi$.

Recent work~\cite{Hayes2023Bounding} refined this bound using the $f$-DP~\cite{dong2019Gaussiana}, a characterization of DP that captures the exact statistical indistinguishability between neighbors through the functional $f$. Formally,
\begin{definition}[\cite{Kifer2022Bayesian}]\label{def:f-DP}
Let $f\colon[0,1]\to[0,1]$ be a continuous, convex, non-increasing function such that $f(x)\leq 1-x$. 
A mechanism $\mathcal{M}$ satisfies $f$-DP if for all $D_0,D_1 \in \Z^n$ such that $d_{H}(D_0,D_1)\leq 1$ and all post-processing algorithms $A\colon\mathrm{Range}(\M)\to\D(\{0,1\})$,  
\[
\Pr(A(\mathcal{M}(D_0)) = 1) \leq 1 - f\big(\Pr(A(\mathcal{M}(D_1)) = 1)\big).
\]
\end{definition}

Here, $f$ is known as a \emph{trade-off function}~\cite{dong2019Gaussiana}, named for its interpretation in the context of hypothesis testing. 
Specifically, consider $A$ as a test of
\[
\begin{cases}
    H_0:& \text{ the input is } D_0\\
    H_1:&  \text{ the input is } D_1,
\end{cases}
\]
applied to the output of $\mathcal{M}$. Then $\Pr(A(\mathcal{M}(D_0))=1)$ is the significance level and $\Pr(A(\mathcal{M}(D_1))=1)$ is the power of the test.   Under this interpretation, for a given significance level, $f$ bounds the maximum achievable power. When \(f\) is the trade-off function between two normal distributions with different means, namely
$
f(\alpha) = \Phi\!\left(\Phi^{-1}(1-\alpha) - \mu\right),
$
where \(\Phi\) denotes the standard normal CDF, the resulting notion is known as \emph{Gaussian DP} (\(\mu\)-GDP).

The $f$-DP framework facilitates the computation of quantities such as the \emph{total variation} distance:

\begin{definition}\label{def:TV}
   A mechanism $\M$ has total variation at most $\mathrm{TV}(\M)$ if, for all neighboring datasets $D_0, D_1$,
   \[
   \sup_{S\subseteq \Theta}|\Pr(\M(D_0)\in S)-\Pr(\M(D_1)\in S)|\leq \mathrm{TV}(\M).
   \]
\end{definition}

For any $\M$ satisfying $(\varepsilon,\delta)$-DP, its TV is bounded~\cite{kairouz15composition} as
\begin{equation}\label{eq:f-to-tv}
 \mathrm{TV}(\M)\leq \max_{\alpha\in[0,1]} \big(1 - f(\alpha) - \alpha\big) \leq \frac{e^{\varepsilon}-1+2\delta}{e^{\varepsilon}+1}.    
\end{equation}
Both $f$-DP and TV are preserved under composition. Specifically, the $T$-fold  composition of an $f$-DP mechanism satisfies  $f^{\otimes T}$-DP, where $f \otimes f$ denotes the trade-off function $T(P \times P, Q \times Q)$ for $f=T(P,Q)$. For instance, if a mechanism is $\mu$-GDP, then its $T$-fold composition is $(\mu \sqrt{T})$-GDP~\cite{dong2019Gaussiana}. Moreover, if $\mathrm{TV}(\mathcal{M}_i) = \Delta$, then the $T$-fold  composition satisfies $\mathrm{TV}(\mathcal{M}) \le 1 - (1-\Delta)^T$~\cite{ghazi2024total}. This bound can be sharpened to $\max_{\alpha}\bigl(1 - f^{\otimes T}(\alpha) - \alpha\bigr)$ when $f$ is known.

\citeauthor{Hayes2023Bounding}~\cite{Hayes2023Bounding} present the first bound for any $f$-DP mechanism:
\begin{equation}\label{eq:hayes_bound}
    \gamma\leq 1-f(\kappa^{+}_{\pi,\ell}(\eta)).
\end{equation}
which they showed empirically nearly tight for DP-SGD, the most known DP algorithm for private learning~\cite{Abadi2016Deep}. 

\subsection{Measure Theory Results}
In this section, we present the disintegration theorem, a fundamental result in measure theory that plays a key role in the proofs of this paper.

In continuous probability spaces, events of the form $X = x$ have probability zero, so conditional probabilities defined via ratios are not well-defined. The disintegration theorem provides a rigorous substitute: Any joint probability measure $\mu$ on $X \times Y$ can be decomposed as
\[
\mu(dy\,dx) = \mu_x(dy)\,\mu_X(dx),
\]
where $\mu_X$ is the marginal of $X$ and $\mu_x$ is a probability measure on $Y$ representing the conditional law of $Y$ given $X = x$. This decomposition allows conditional distributions to be defined point-wise (almost everywhere), despite conditioning on null events.

This intuition extends from Cartesian products to general measurable maps $a \colon \Z \to \X$, where disintegration allows one to define conditional measures $\mu_x$ supported on the fibers $a^{-1}(x)$, providing a rigorous notion of conditioning on $a(z)=x$ even when $\mu(a^{-1}(x))=0$. Formally:

\begin{theorem}[Disintegration Theorem~\cite{baccelli24random}]~\label{th:disintegration} Let $(\Z,\mathcal{B}(\Z))$ and $(\X,\mathcal{B}(\X))$ be standard Borel spaces and  $(\Z,\mathcal{B}(\Z),\mu)$ be a probability space. Let
$a: \Z \to \X$ be a measurable map.
Denote by $\nu = \mu \circ a^{-1}$ the push-forward measure of $\mu$ through $a$.
Then there exists a $\nu$-almost everywhere uniquely determined family of probability measures
$\{\mu_x\}_{x \in \X}$ on $\Z$ such that:
\begin{enumerate}
    \item For $\nu$-a.e.\ $x \in \X$, $\mu_x$ is supported on the fiber
    $a^{-1}(x)$, i.e.
    \[
    \mu_x\!\left(\Z \setminus a^{-1}(x)\right) = 0 .
    \]

    \item For every measurable set $B \subseteq\Z$,
    \[
    \mu(B)
    =
    \int_{\X} \mu_x(B)\, d\nu(x).
    \]

    \item For every integrable function $f \in L^1(\Z,\mu)$,
    \[
    \int_\Z f(z)\,d\mu(z)
    =
    \int_\X\left(
        \int_{a^{-1}(x)} f(z)\,d\mu_x(z)
    \right)d\nu(x).
    \]
\end{enumerate}
\end{theorem}

Note that when $\mu$ is a product measure $\mu_Z \otimes \mu_X$ and $a$ is the projection onto $Z$, the conditional measures $\mu_z$ can be taken equal to $\mu_X$ for $\nu$-almost every $z$, and the above reduces to the classical Fubini--Tonelli theorem.

\begin{Remark}\label{re:discrete}
The disintegration theorem applies straightforwardly when $\mu$ is the counting
measure on a discrete space. Let $\Z$ and $\X$ be  discrete finite sets and let
$a : \Z \to \X$ be any measurable function.

For each $x \in \X$, define the fiber
\[
z_x := a^{-1}(x) = \{ z \in \Z : a(z) = x \}.
\]
Let $\mu$ be the counting measure on $\Z$, i.e., $\mu(B)=\#B$ for $B \subseteq \Z$.
The push-forward measure is then
\[
\nu(x) = \mu(a^{-1}(x)) = \# z_x.
\]

For each $x$ with $\nu(x)>0$, define the conditional measure $\mu_x$ on $\Z$ by
\[
\mu_x(B) = \frac{\#(B \cap z_x)}{\# z_x},
\]
i.e., the uniform distribution on the fiber $z_x$.
Then for any $B \subseteq \Z$,
\[
\mu(B) = \sum_{x \in \X} \mu_x(B)\, \nu(x).
\]

Consequently, for any function $f : \Z \to [0,\infty]$,
\[
\sum_{z\in \Z} f(z)
= \sum_{x \in \X} \left( \int f(z)\, \mu_x(dz) \right) \nu(x)
= \sum_{x \in \X} \left( \sum_{z \in z_x} f(z) \right).
\]
\end{Remark}

\Cref{th:disintegration} is crucial to understand and prove novel properties on attack resilience of DP that we present in~\Cref{sec:main}.

\section{Review of the Related Work}\label{sec:related_work}
In this section, we review relevant previous work on measuring the effective attack resilience of DP mechanisms for calibration and auditing. Particularly, we discuss novel insights and gaps that motivate our work.

\textbf{Attack-Based DP Noise Calibration.} Several recent studies~\cite{Bernau2021Quantifying,kulynych2024attack,chatzikokolakis2023bayes} demonstrate that calibrating DP noise based on resilience to specific attacks can significantly improve utility. Such approaches, however, primarily target MIAs, which may lead to unnecessary utility degradation without offering meaningful privacy benefits when membership is public or considered non-sensitive~\cite{Balle2022Reconstructing}. 

Beyond MIAs, privacy concerns often involve AIA, where the adversary aims to infer sensitive attributes of individuals from released data~\cite{Pyrgelis2017What, Jayaraman2022Analyzing}. 
A common metric for evaluating such attacks is the attribute advantage~\cite{Yeom2017Privacy}. Existing works that provide theoretical bounds for AIAs either analyze specific attack strategies~\cite{Yeom2017Privacy} or adopt more general DRA frameworks~\cite{Balle2022Reconstructing, Guerra2024Analysis}. Within the latter, the notion of ReRo has emerged as the metric for measuring the risk of DRAs, under which attribute inference can be modeled as a special case~\cite{Balle2022Reconstructing}. Moreover, ~\Cref{eq:balle_bound}~\cite{Balle2022Reconstructing} and~\Cref{eq:hayes_bound}~\cite{Hayes2023Bounding} provide ReRo-based DP noise calibration methods. 

\textbf{A note on limitations of ReRo.}
%NEW (be nicer with ReRo)
\citeauthor{Balle2022Reconstructing}'s pioneering work~\cite{Balle2022Reconstructing} introduced ReRo and linked it to DP, providing a framework to assess the risks of DRAs and enabling risk analysis beyond MIAs. ReRo is suitable when the adversary’s reconstruction capability is entirely based on the participation of the record, yet extending it to broader settings introduces significant limitations.

A general-purpose risk metric  would be expected to cover all relevant attack scenarios.
However, ReRo does not formally account for the impact of target-specific auxiliary knowledge, hence excluding MIAs, AIAs and targeted DRAs as introduced in~\Cref{sec:background_attacks}. 
Formally, the attack  considered in \cite{Balle2022Reconstructing} (see~\Cref{def:rero} for details), only has access to the mechanism output $\M(D)$, i.e., $A\colon \Theta \to \D(\Z)$, implying that $\Pr(A(\M(D), a(z))\in S) = \Pr(A(\M(D), a(z'))\in S)$ for any pair of possible targets  $z, z'$ and output set  $S$.
Under this assumption, the attacker \(A\) cannot adapt its strategy to a specific target \(z\). This choice is reasonable for attacks that attempt to reconstruct a record without relying on auxiliary knowledge. However, it  fundamentally prevents assessing the risk of MIA and AIA, as they use full or partial knowledge of some target records.
This is a relevant limitation since most real-world privacy attacks historically exploit publicly available information about the target~\cite{Sweeney2000Simple, Montjoye2013Unique, Narayanan2008Robust}. 
Moreover, we show in~\Cref{sec:main}  several attacks that leverage target-specific auxiliary knowledge, and its success highly depends on it.

All formal bounds connecting ReRo and DP were proven under this restrictive exclusion.  
The requirement that the attack depends only on $\M(D)$---ignoring target-specific information---is critical to establishing both~\Cref{eq:balle_bound,eq:hayes_bound}. This is not merely a theoretical limitation: we show in~\Cref{sec:experiments} that these bounds do not hold for attacks that exploit target-specific knowledge against well-known mechanisms such as DP-SGD.

A direct extension of ReRo to targeted attacks $A(\theta, a(x))$ may also lead to problematic assessments: not only do the original bounds no longer hold, but the metric also collapses to a substantial overestimation of risk due to imputation and background knowledge. For instance, the trivial MIA that identifies every target as a member regardless of the mechanism output, $A(\theta, z) = z$, has success probability $1$, which ReRo would interpret as a catastrophic privacy risk, even though no actual leakage occurs. This is not a negligible edge case; it has caused misleading overestimation of risk in black-box attacks on classification models~\cite{Jayaraman2022Are}, where much of the reported success arose from data imputation rather than exploiting the mechanism’s output. Such overestimation obscures the true leakage and can lead to unnecessary utility loss when ReRo is used to calibrate noise in DP. 

Even under the original assumption that the attacker has no target-specific knowledge, ReRo still overestimates risk, as we discussed in our preliminary work~\cite{Guerra2024Analysis}. The mechanism output $\M(D)$ inherently reveals distributional information and population-level statistics, which are the primary goals of any learning process. This information can be used to perform imputation and infer attributes of individual records---even those not in $D$---with high accuracy, particularly when strong correlations exist (e.g., smoking correlating with cancer). In this case, the apparent attack success is driven by statistical inference rather than actual privacy violations, a phenomenon often referred to as a \emph{privacy fallacy}~\cite{Dwork2006Differential,Kifer2022Bayesian}. Indeed, several works establish that it is impossible to simultaneously provide utility and eliminate absolute information gain~\cite{Dwork2006Differential,Kifer2022Bayesian}.

We conclude that, while foundational, ReRo may be misleading as a general-purpose attack resilience metric, as it overlooks key statistical phenomena that distort privacy risk assessment, such as data imputation and targeted attacks. Both cases are very common and have an impact in practice (see~\Cref{sec:experiments}), motivating the need for a novel framework to more accurately assess the risk of DP mechanisms with respect to attacks.

\textbf{DP Auditing.} 
DP auditing~\cite{Annamalai2025Hitchhiker} seeks to demonstrate tight estimates of the privacy budget, discover implementation flaws, and estimate empirical privacy.
% Black-box methods for discovery of DP violations, 
However, auditing in practice remains a significant challenge. For instance, implementation bugs or design flaws can severely degrade privacy guarantees in ways that are not immediately obvious. To address this, black-box discovery methods such as DP-Sniper~\cite{Bichsel2021DPSniper} and Eureka~\cite{Lu2024Eureka} have been developed to detect DP violations by training classifiers to distinguish between mechanism outputs from ``worst-case'' adjacent inputs. 
This methodology implicitly assumes that the mechanism’s output distribution lies in a low-dimensional, learnable representation.
While effective at uncovering certain classes of violations, this assumption breaks down for frequency-oracle mechanisms over high-dimensional categorical domains, where outputs are discrete randomized encoding~\cite{Arcolezi2023On} with inherently combinatorial structure. Consequently, the learned classifiers fail to scale, becoming prohibitively slow or ineffective as the domain dimension grows.
 
% When trying to assess empirical privacy, most existing auditing approaches focus on 
Beyond identifying bugs, existing empirical privacy auditing approaches primarily focus on MIAs~\cite{Bernau2021Quantifying, Annamalai2024Nearly, Jagielski2020Auditing, Steinke2023Privacy}, which limits their ability to detect broader forms of privacy leakage. Some auditing techniques extend beyond MIAs to consider AIAs, but these are restricted to specific contexts—such as Label DP~\cite{Malek2021Antipodes} or synthetic data generation~\cite{Houssiau2022Tapas}.
In the LDP setting, the state-of-the-art framework \textsc{\textsc{LDP Auditor}}~\cite{Arcolezi2024Revealing} relies specifically on perfect reconstruction without target-specific auxiliary knowledge for auditing. 
 
Summarizing, despite its practical importance, no existing auditing framework incorporates auxiliary information or supports a DRA-based analysis that goes beyond MIAs and enables systematic evaluation across diverse DP mechanisms. Our preliminary work~\cite{Guerra2024Analysis} made partial progress by analyzing adversaries that rely solely on the mechanism output; however, it did not account for the impact of target-specific auxiliary information, which is often decisive in real-world privacy breaches, such as the classical census re-identification~\cite{Sweeney2000Simple}.
This gap motivates the development of a general auditing methodology designed to capture realistic adversaries and to quantify broader classes of privacy risks. 

\section{Reconstruction Advantage}\label{sec:main}
In this section, we introduce reconstruction advantage (RAD) as a novel, unifying metric for adversarial risk assessment. 
We first establish a worst-case bound on RAD that holds for any mechanism, data distribution, and auxiliary knowledge, ensuring robustness when the attacker’s prior knowledge is unknown. We then refine this result by deriving a tighter bound under known auxiliary knowledge and prove its tightness by constructing the corresponding optimal attack that achieves it. Together, these results provide a noise calibration method to optimize utility for a given risk. We empirically validate the practical tightness of our bounds in~\Cref{sec:experiments}.

In order to address ReRo’s lack of accounting for the impact of target-specific auxiliary knowledge, we explicitly incorporate this concept into RAD. Formally, each record $z \in \Z$ may be associated with target-specific auxiliary information $a(z) \in \mathit{aux}$. The auxiliary information can take different forms. For instance, in the classical AIA setting, where records are pairs $z=(x,y)$, one may define $a(z)=x$ and attempt to infer $y$. Alternatively, in the image reconstruction setting, the target may be the full record $z$, while $a(z)$ could correspond to a label such as ``image of a person'' or ``image of an animal''. The only structural assumption we impose is that the type of auxiliary information is consistent across all records: if $a(z)$ corresponds to a set of pixels, then for any other record $z'$, $a(z')$ must also be a set of pixels (and not, for example, a semantic label).
Having established this formalization, we introduce our metric.

\begin{definition}[$\eta$-RAD]\label{def:new_u-rero}
    Let $\pi$ be a prior over $\Z$, $\ell\colon \Z\times\Z \to\R_{\geq 0}$ an error function, and $a(z)\in aux$ the target-specific auxiliary information for each $z\in\Z$.  
     Given a mechanism $\mathcal{M}: \Z^n \rightarrow \D(\Theta)$, any dataset $D_- \in \Z^{n-1}$ and any adversary $A: \Theta \times aux \rightarrow \D(\Z)$ we define the $\eta$-reconstruction advantage, $\eta$-RAD, as

    \[
     \eta\text{-}\mathrm{RAD}=  \Pr_{\substack{Z_1 \sim \pi\\ \theta \sim \mathcal{M}(D_{Z_1})}}[\ell(Z_1, A(\theta, a(Z_1))) \leq \eta]
    -
            \Pr_{\substack{Z_0,Z_1 \sim \pi\\ \theta \sim \mathcal{M}(D_{Z_0})}}
                [\ell(Z_1, A(\theta, a(Z_1)) \leq \eta].
    \]
    
\end{definition}
RAD explicitly accounts for target-specific auxiliary knowledge, providing a generalization of the membership and attribute advantages to arbitrary reconstruction attacks. Importantly,
RAD outputs values between $-1$ and $(1-\kappa_{\pi})\leq 1$ where $\kappa_{\pi}=\Pr_{Z,Z'\sim\pi}[Z=Z']$, i.e., the probability of resampling from the distribution $\pi$, analogously to membership and attribute advantage (see~\Cref{sec:background}). Intuitively, $\kappa_{\pi}$ reflects the fact that if dataset members are drawn from a finite universe, when we randomly sample a record from the universe to simulate non-members, there is a probability, $\kappa_{\pi}$, that it coincides with the record of the actual participant. 

Intuitively, RAD measures the increase in the attacker’s success probability that arises solely from the target’s participation in the private learning process. In this way, RAD avoids the overestimation of risk that is inherent in ReRo. If $\mathrm{RAD} \leq 0$, participation carries no risk, since the attacker’s probability of correctly reconstructing the record is no greater than if the individual had not participated. Larger values of RAD indicate higher participation risk. In the extreme case where $\mathrm{RAD} = 1 - \kappa_{\pi}$, participation entails absolute risk: the attacker always succeeds in reconstructing the participant’s record, while no sensitive information can be reconstructed from non-participants. RAD can be normalized by $(1-\kappa_{\pi})$ to obtain an upper bound of $1$.

Previous bounds for ReRo assume that DRAs perform equally for every target. This assumption holds when the adversary has no target-specific auxiliary  knowledge ($aux=\{\varnothing\}$), but breaks once \emph{aux} is available: for instance, knowing that a target's surname is ``Smith'' might give less information than knowing that it is ``Sainthorpe-Burton'', as the latter is less frequent and hence carries more information. 
Such differences are not captured by ReRo, nor reflected in the proofs of the corresponding bounds~\cite{Balle2022Reconstructing, Guerra2024Analysis}. 
Hence, we provide the first theoretical bound that explicitly accounts for \emph{aux} and covers any possible attack from MIAs to the most general DRAs:

\begin{theorem}[$(\varepsilon, \delta)$-DP implies $\eta$-RAD]\label{th:dp_implies_aux-urero}
    Let $\pi,\ell,\eta \geq 0$ as in Def. \ref{def:new_u-rero}, and $\kappa_{\pi}=\Pr_{Z,Z'\sim\pi}[Z=Z']$. If a mechanism $\mathcal{M}\colon\Z^{n}\to\D(\Theta)$ satisfies $(\varepsilon,\delta)$-DP, then for any attack $A\colon \Theta\times aux\to\D(\Z)$, and database $D_{-}$ we have 
    \[
    \eta\text{-}\mathrm{RAD}\leq \mathrm{TV}(\mathcal{M})(1-\kappa_\pi) \leq  \frac{e^{\varepsilon}-1+2\delta}{e^{\varepsilon}+1}(1-\kappa_\pi).
    \]
\end{theorem}
\begin{proof} 
We use $\int f(z)\diff\mu(z)$ as unified notation that represents either a sum (if $\mu$ is the counting measure) or an integral (if 
$\mu$ is the Lebesgue measure), aggregating both the discrete and continuous case in one single notation. 

First, note that for every $z\in\Z$ and target-specific knowledge $a(z)$, any attack admits the representation $ A(D,a(z))\equiv\A_{z}(\M(D))$, verifying
\[
p_{\A_{z}}(s\mid D)\equiv p_{A}\big(s\mid a(z),D\big)=\int_{\Theta} p_{\M}(\theta\mid D)\, p_{A}(s\mid \theta,a(z))\diff \mu(\theta).
\]

Note that the attack outputs values in $\Z$. Therefore,
\begin{align}
    \MoveEqLeft[5]\mathrm{TV}(\A_x(D),\A_x(D'))
    \coloneqq \sup_{S\subseteq\Z}\;\big|\Pr(\A_x(D)\in S)-\Pr(\A_x(D')\in S)\big| \notag\\
    &= \tfrac{1}{2}\int_{\Z}\Big|\,p_{A}(s\mid \mathcal M(D),a(z))-p_{A}(s\mid \mathcal M(D'),a(z))\Big| \diff \mu(s)\label{eq:mixing}\\
    &= \tfrac{1}{2}\int_{\Z}\Big|\int_{\Theta} p_{A}\big(s\mid \theta,a(z)\big)\,
       \big(p_{\mathcal M}(\theta\mid D)-p_{\mathcal M}(\theta\mid D')\big)\diff\mu(\theta)\,\Big|\diff\mu(s)\notag \\
    &\leq
       \tfrac{1}{2}\int_{\Z}\int_{\Theta} p_{A}\big(s\mid \theta,a(z)\big)\,
       \big|p_{\mathcal M}(\theta\mid D)-p_{\mathcal M}(\theta\mid D')\big|\diff\mu(\theta)\,\Big|\diff\mu(s) \label{eq:minkowski}\\
    &= \tfrac{1}{2}\int_{\Theta}\big|p_{\mathcal M}(\theta\mid D)-p_{\mathcal M}(\theta\mid D')\big|\diff\mu(\theta)
       \int_{\Z}p_{A}\big(s\mid \theta,a(z)\big)\diff \mu(s)\notag \\
    &= \tfrac{1}{2}\int_{\theta}\big|p_{\mathcal M}(\theta\mid D)-p_{\mathcal M}(\theta\mid D')\big|\diff\mu(\theta)\notag \\
    &= \mathrm{TV}(\mathcal M(D),\mathcal M(D')),
    \label{eq:tv_postprocessing}
\end{align}

where \Cref{eq:mixing} follows from~\cite[Proposition~4.2, p.~48]{levin2017Markov} and~\Cref{eq:minkowski} from Minkowski's inequality. 

Moreover, given any success set $S_{\eta}(z)=\{z'\in\Z\colon \ell(z,z')\leq\eta\}$, and using the notation $ A(D,a(z))\equiv\A_{z}(\M(D))$, we have
\[
\Pr_{\substack{Z_1 \sim \pi\\ \theta \sim \mathcal{M}(D_{Z_0})}}
        [\ell(Z_1, A(\theta, a(Z_1))) \leq \eta]=\Pr_{Z_1 \sim \pi}
        [ \A_{Z_1}(D_{Z_0})\in S_{\eta}(Z_1)].
\]
Hence, applying~\Cref{eq:tv_postprocessing} and~\Cref{def:TV} 
to RAD~\Cref{def:new_u-rero} we obtain:
\begin{gather*}
   \eta\text{-RAD}= \Pr_{Z_1 \sim \pi}
        [ \A_{Z_1}(D_{Z_1})\in S_{\eta}(Z_1)]
    -
            \Pr_{Z_0,Z_1 \sim \pi}
                 [ \A_{Z_1}(D_{Z_0})\in S_{\eta}(Z_1)]\\
    =
    \E_{Z_0\sim\pi}\left[\Pr_{Z_1 \sim \pi}
        [ \A_{Z_1}(D_{Z_1})\in S_{\eta}(Z_1)]-\Pr_{Z_1 \sim \pi}
                 [ \A_{Z_1}(D_{Z_0})\in S_{\eta}(Z_1)]\right]\\
    =\E_{Z_0,Z_1\sim\pi}\left[\boldsymbol{1}_{\{Z_0\neq Z_1\}}\big(\,
    \Pr[\A_{Z_1}(D_{Z_1})\in S_{\eta}(Z_1)]-\Pr[\A_{Z_1}(D_{Z_0})\in S_{\eta}(Z_1)]\,\big)
    \right]\\
    \overset{\text{Eq.}\ref{eq:tv_postprocessing}}{\leq} \mathrm{TV}(\M)\E_{Z_0,Z_1\sim\pi}\left[\boldsymbol{1}_{\{Z_0\neq Z_1\}}\right].
\end{gather*}
The result follows from the fact that $\E_{Z_0,Z_1\sim\pi}\left[\boldsymbol{1}_{\{Z_0\neq Z_1\}}\right]=1-\sum_{z}\pi^2_{z}$ for discrete variables and $1$ for continuous ones. Finally,~\Cref{eq:f-to-tv} completes the proof.
 \end{proof}

Note that in the discrete case, $\kappa_{\pi}=\sum_{z}\pi_{z}^2$, which is maximized when $\pi$ is uniform over two possible records ( e.g., $\pi=U\{Z_0,Z_1\}$). In the continuous case, the resampling probability is, by definition, zero. Consequently, the  result simplifies to $\eta$-RAD $\leq \mathrm{TV}(\M)$, unaffected by the prior distribution.

\Cref{th:dp_implies_aux-urero} is the first bound for RAD under the strongest threat model, where the attacker may leverage auxiliary knowledge. 
Particularly, this results states that if the mechanism is fully known, we can determine the attack mitigation it provides  by its total variation. When only the DP parameters $\varepsilon$ and $\delta$ are available, the bound quantifies how each parameter contributes to the attacker’s advantage. Consequently, this theorem serves as a key tool for DP noise calibration, improving on ReRo by encompassing a broader spectrum of potential attackers.

Moreover, \Cref{th:dp_implies_aux-urero} allows upper bounding RAD under composition. As we discussed in~\Cref{sec:background}, given $\mathrm{TV}(\M_i)=\Delta$, the $T$-adaptive composition satisfies $\mathrm{TV}(M)\leq(1-(1-\Delta)^T)$. Hence, $\eta$-RAD $\leq(1-(1-\Delta)^T)(1-\kappa_\pi)$.

\Cref{th:dp_implies_aux-urero} does not depend on the attacker’s auxiliary knowledge. Therefore, the same bound holds whether the attacker has no auxiliary information ($aux=\{\varnothing\}$) or complete knowledge of the record ($a(z)=z$), since the result is derived in a worst-case manner. 
However, when the attacker's goal is to reconstruct an entire record (as in DRA) or infer parts of it (as in AIA), it is unreasonable to assume that the attacker already knows the full record ($a(z)=z$)---as assumed for MIA. 
Therefore, we next provide a tighter bound that explicitly incorporates the target-specific auxiliary knowledge.

\begin{theorem}\label{th:optimal_bound}
Given $\mathcal{M}\colon\Z^{n}\to\D(\Theta)$ and $a\colon\Z\to aux$ measurable, then for any attack $A\colon \Theta\times aux\to\D(\Z)$, we have
\[
\eta\text{-}\mathrm{RAD}\leq \int_{\Theta}\int_{ aux}\max_{z_{\theta}\in\Z}\left(\int_{S^x_{\eta}(z_{\theta})} w(\theta,z)\,\pi_z\,\diff \mu_{x}(z)\right)\diff\nu(x)\diff\mu(\theta),
\]
where $\mu$ the counting (or Lebesgue measure) and $\pi_z$ mass (or density function)
in the discrete (or continuous case). Additionally,  $w(z,\theta)=p_{\M}(\theta\mid z)-p_{\M}(\theta)$, $S_{\eta}^x(z_{\theta})=\{z\colon a(z)=x\wedge \ell(z_{\theta},z)\leq\eta\}$,  $\nu(x)=\mu\circ a^{-1} (z)$ and $\mu_x$ the disintegration theorem measure. 

The discrete case simplifies to
\[
\eta\text{-}\mathrm{RAD}\leq \sum_{\theta\in\Theta}\sum_{x\in aux}\max_{z_{\theta}\in\Z}\sum_{\substack{{\ell(z,z_\theta)\leq\eta}\\a(z)=x}} w(\theta,z)\pi_z
\]
by direct application  of Remark~\ref{re:discrete}.
\end{theorem}
\begin{proof}
We denote by $\mu$ the counting measure in the discrete case and the Lebesgue measure in the continuous case. Moreover, following the notation introduced in~\Cref{th:dp_implies_aux-urero}, we consider $A(D,a(z))\equiv\A_{z}(\M(D))$.
First, using probability properties, we rewrite RAD definition as
{\small
      \begin{align}
   \MoveEqLeft[1]\eta\text{-RAD} 
   = \Pr_{Z_1 \sim \pi}
        [ \A_{Z_1}(D_{Z_1})\in S_{\eta}(Z_1)]
    - \Pr_{Z_0,Z_1 \sim \pi}
        [ \A_{Z_1}(D_{Z_0})\in S_{\eta}(Z_1)] \notag\\
   &= \E_{Z_0\sim\pi}\Big[\Pr_{Z_1 \sim \pi}
        [ \A_{Z_1}(D_{Z_1})\in S_{\eta}(Z_1)]
        - \Pr_{Z_1 \sim \pi}
          [ \A_{Z_1}(D_{Z_0})\in S_{\eta}(Z_1)]\Big] \notag\\
   &= \E_{Z_0,Z_1\sim\pi}\Big[\boldsymbol{1}_{\{Z_0\neq Z_1\}}\big(\,
        \Pr[\A_{Z_1}(D_{Z_1})\in S_{\eta}(Z_1)]
        - \Pr[\A_{Z_1}(D_{Z_0})\in S_{\eta}(Z_1)]\,\big)\Big] \notag\\
   &= \E_{Z_0,Z_1\sim\pi}\Big[\boldsymbol{1}_{\{Z_0\neq Z_1\}}
     \int_{\Theta} p_{A}\big(S_{\eta}(Z_1)\mid \theta,a(Z_1)\big)\,
       \Big(p_{\mathcal M}(\theta\mid D_{Z_1})
         - p_{\mathcal M}(\theta\mid D_{Z_0})\Big)\diff\mu(\theta)\Big] \notag\\
   &= \E_{Z_1\sim\pi}\Big[
     \int_{\Theta} p_{A}\big(S_{\eta}(Z_1)\mid \theta,a(Z_1)\big)\,
       \E_{Z_0\sim\pi}\Big[\boldsymbol{1}_{\{Z_0\neq Z_1\}}
         \big(p_{\mathcal M}(\theta\mid D_{Z_1})
          - p_{\mathcal M}(\theta\mid D_{Z_0})\big)\Big]\diff\mu(\theta)\Big] \notag\\
   &= \E_{Z_1\sim\pi}\Big[
     \int_{\Theta} p_{A}\big(S_{\eta}(Z_1)\mid \theta,a(Z_1)\big)\,
       \Big(p_{\mathcal M}(\theta\mid D_{Z_1})
         \E_{Z_0\sim\pi}[\boldsymbol{1}_{\{Z_0\neq Z_1\}}]
       - \E_{Z_0\sim\pi}[\boldsymbol{1}_{\{Z_0\neq Z_1\}}
         p_{\mathcal M}(\theta\mid D_{Z_0})]\Big)\diff\mu(\theta)\Big] \notag\\
   &= \E_{Z_1\sim\pi}\Big[
     \int_{\Theta} p_{A}\big(S_{\eta}(Z_1)\mid \theta,a(Z_1)\big)\,
       \underbrace{\big(p_{\mathcal M}(\theta\mid D_{Z_1})
         -p_{\mathcal M}(\theta)\big)}_{w(Z_1,\theta)}
     \diff\mu(\theta)\Big] \label{eq:truq} \\
    &= \int_{\Z}  \int_{\Theta} p_{A}\big(S_{\eta}(z_1)\mid \theta,a(z_1)\big)\,w(z_1,\theta)\,\pi_z
     \diff\mu(\theta) \diff\mu(z),\notag
\end{align}   
}

where~\Cref{eq:truq} follows trivially for the continuous case, since $\E_{Z_0\sim\pi}\left[\boldsymbol{1}_{\{Z_0\neq Z_1\}}\right]=1$, and for the discrete one, since
    \begin{gather*}
        p_{\mathcal M}(\theta\mid D_{Z_1})\E_{Z_0\sim\pi}\left[\boldsymbol{1}_{\{Z_0\neq Z_1\}}\right]-\E_{Z_0\sim\pi}\left[\boldsymbol{1}_{\{Z_0\neq Z_1\}}p_{\mathcal M}(\theta\mid D_{Z_0})\right]\\
        =p_{\mathcal M}(\theta\mid D_{Z_1})(1-\pi_1)-\E_{Z_0\sim\pi}\left[p_{\mathcal M}(\theta\mid D_{Z_0})\right]+p_{\mathcal M}(\theta\mid D_{Z_1})\pi_1\\
        =p_{\mathcal M}(\theta\mid D_{Z_1})-p_{\M}(\theta).
    \end{gather*}
    
Moreover, for all records, $z_1,z_2$, with the same auxiliary knowledge, i.e., $a(z_1)=a(z_2)=x$, and for any fixed output $\theta$, we have that
\[
\Pr_{A}(S_{\eta}(z_1)\mid\theta, a(z_1))=\Pr_{A}(S_{\eta}(z_1)\mid \theta,a(z_2))=\Pr_{A}(S_{\eta}(z_1)\mid \theta, x).
\]

Hence, given $a^{-1}(x)=\{z\colon a(z)=x\}$ for all $x\in aux$, and $\nu(x)=\mu\circ a^{-1} (x)$, applying disintegration theorem~\cite{baccelli24random} there exists a unique measure $\mu_x$ such that 
\begin{gather}
   \eta\text{-RAD}=\int_{\Z}  \int_{\Theta} \Pr_{A}(S_{\eta}(z)\mid a(z),\theta)\, w(z,\theta)\, \pi_z \diff\mu(z)\diff\mu(\theta)\notag\\
    =  \int_{\Theta} \int_{ aux}\int_{a^{-1}(x)} \Pr_{A}(S_{\eta}(z)\mid x,\theta)\, w(z,\theta)\, \pi_z \diff\mu_x(z)\diff \nu(x)\diff\mu(\theta)\notag
    \\
    =\int_{\Theta} \int_{ aux}\int_{a^{-1}(x)} \int_{\Z} \boldsymbol{1}_{\{\ell(z,\tilde{z})\leq \eta\} }\, p_{A}(\tilde{z}\mid x,\theta)\, w(z,\theta)\, \pi_z \diff\mu(\tilde{z}) \diff\mu_x(z)\diff \nu(x)\diff\mu(\theta)\notag
    \\
    = \int_{\Theta} \int_{ aux} \int_{\Z} p_{A}(\tilde{z}\mid x,\theta)\left(\int_{a^{-1}(x)}\boldsymbol{1}_{\{\ell(z,\tilde{z})\leq\eta \} }\, w(z,\theta)\, \pi_z \diff\mu_{x}( z)\right)\diff\mu(\tilde{z})\diff\nu(x)\diff\mu(\theta)\notag\\
    \leq
    \int_{\Theta} \int_{aux} \int_{\Z} p_{A}(\tilde{z}\mid x,\theta)\left(\max_{z_{\theta}\in \Z}\int_{a^{-1}(x)}\boldsymbol{1}_{\{\ell(z,z_{\theta})\leq\eta \} }\, w(z,\theta)\, \pi_z \diff\mu_x(z)\right)\diff \mu(\tilde{z})\diff\nu(x)\diff\mu(\theta)\notag\\
         =\int_{\Theta} \int_{aux} \max_{z_{\theta}\in \Z}\int_{a^{-1}(x)}\boldsymbol{1}_{\{\ell(z,z_{\theta})\leq\eta \} } w(z,\theta)\, \pi_z \diff\mu_x(z)\left(\int_{\Z} p_{A}(\tilde{z}\mid x,\theta)\diff \mu(\tilde{z})\right)\diff\nu(x)\diff\mu(\theta)\notag\\
     =\int_{\Theta} \int_{ aux} \max_{z_{\theta}\in \Z}\int_{a^{-1}(x)\cap \{z\colon \ell(z,z_{\theta})\leq \eta\}} w(z,\theta)\, \pi_z \diff\mu_x(z)\diff\nu(x)\diff\mu(\theta)\notag\\
     =\int_{\Theta} \int_{ aux} \max_{z_{\theta}\in \Z}\int_{S_{\eta}^z(z_{\theta})} w(z,\theta) \pi_z \diff\mu_x(z)\diff\nu(x)\diff\mu(\theta)\notag
\end{gather}
where $S_{\eta}^x(z_{\theta})=\{z\colon a(z)=x\wedge \ell(z_{\theta},z)\leq\eta\}$.
\end{proof}

\Cref{th:optimal_bound} bounds RAD when the specific mechanism, $\M$, and auxiliary knowledge, $aux$, are known. 
At the same time, it becomes more precise than our worst-case bound~\Cref{th:dp_implies_aux-urero}, as we illustrate in~\Cref{fig:optimal_examples}. 
Moreover, although this bound is inherently complex due to its generality, it admits simpler characterizations for commonly studied threat models---such as MIAs and DRAs where no target-specific auxiliary knowledge is assumed.

In particular, in an MIA, where the attacker has full-knowledge about the record and just seek to infer participation (i.e. $a(z)=z$ for all records), then $aux=\Z$, the push-forward metric simplifies to $\nu(z)=\mu\circ a^{-1}(z)=\mu(z)$, and 
\begin{gather*}
    \mu_{z}(\Z\backslash a^{-1}(z))=\mu_x(\Z\backslash \{z\})=0 \Rightarrow \mu_{z'}(z)=\delta_{z'}(z),
\end{gather*}
where $\delta_{z'}(z)=1$ if $z'=z$ and zero otherwise, therefore satisfying that, for any measurable set $B\subseteq \Z(\equiv aux)$,
\[
\int_{aux} \delta_{z}(B) \diff \nu(x)=\int_{aux}\boldsymbol{1}_{\{z\in B\}}\diff\mu(z)=\mu(B).
\]

Then, according to the disintegration theorem (see~\Cref{sec:background}), $\nu(z)=\mu(z)$, and $\mu_{z}(z)=\delta_{z}(z)$. Moreover, since $a$ is the identity function, 
\[
S_{\eta}^{z}(z_{\theta})=\{z\colon a(z)=x\wedge\ell(z,z_{\theta})\leq\eta\} =\begin{cases}
    \{z\} &\text{ if } \ell(z,z_{\theta})\leq\eta,\\
    \varnothing &\text{ otherwise. }
\end{cases} 
\]

Hence, applying our~\Cref{th:optimal_bound},
\begin{gather}\label{eq:full_auz_continous}
\eta\text{-RAD}\leq\int_{\Theta} \int_{ aux} \max_{z_{\theta}\in \Z}\int_{S_{\eta}^z(z_{\theta})} w(z,\theta) \pi_z \diff\mu_x(z)\diff\nu(x)\diff\mu(\theta)\\ =\int_{\Theta} 
\int_{aux}\max_{\substack{z_{\theta}\in\Z\\\ell(z_{\theta},z)\leq \eta}}\int_{\{z\}}\boldsymbol{1}_{\{\ell(z_{\theta},z)\leq\eta \}} w(\theta,z)\pi_z \diff \delta_{z}(z)\diff\mu(z)
\diff\mu(\theta)\\
=\int_{\Theta}\int_{aux} \max_{\substack{z_{\theta}\in\Z\\\ell(z_{\theta},z)\leq\eta}} w(\theta,z)\pi_z \diff\mu(z)\diff\mu(\theta),\\
=\int_{\Theta}\int_{\{z\colon w(\theta,z)>0\}} w(\theta,z)\pi_z \diff\mu(z)\diff\mu(\theta),
\end{gather}
since $z\in\Z$, $\argmax_{z_{\theta}}=S_{\eta}(z)$ if $w(\theta,z)>0$ and $\argmax_{z_{\theta}}=\Z\backslash S_{\eta}(z)$ otherwise, avoiding negative values.
For discrete variables previous formula simplifies to
\begin{gather}\label{eq:full_auz_discrete}
\eta\text{-RAD}\leq \sum_{\theta\in\Theta}\sum_{\substack{z\in\Z\\ w(\theta,z)>0}} w(\theta,z)\pi_z.
\end{gather}

We can consider the other extreme, when $aux=\{\varnothing\}$. Here, $\varnothing$ is treated as the single element in $aux$. To avoid confusion with properties of the empty set, we instead denote $aux=\{a\}$, meaning that the auxiliary function is constant for every user, i.e., $a^{-1}(a)=\Z$. In this case, there is no target-specific auxiliary knowledge.

Consequently, $\nu(a)=\mu(a^{-1}(a))=\mu(\Z)=1$. Hence, $\nu$ is the Dirac measure $\delta_{a}$.
The first condition defining $\mu_{a}$ according to~\Cref{th:disintegration} is 
\[
\mu_{a}(\Z\backslash a^{-1}(a))=\mu_{a}(\Z\backslash\Z)=\mu_{a}(\varnothing)=0,
\]
which is satisfied by any measure by definition. Hence, we look to the second defining condition, for any measurable set $B\subseteq \Z$
\begin{gather}
    \mu(B)=\int_{\{a\}}\mu_{a}(B)\diff\delta_{a}(a)=\mu_{a}(B),
\end{gather}
hence, $\mu_{a}=\mu$, obtaining:
\begin{gather*}
    \eta\text{-RAD}\leq \int_{\Theta}\int_{ aux}\max_{z_{\theta}\in\Z}\left(\int_{S^a_{\eta}(z_{\theta})} w(\theta,z)\pi_z\diff \mu_{z}(z)\right)\diff\nu(x)\diff\mu(\theta)\\
    =\int_{\Theta}\int_{\{a\}} \max_{z_{\theta}\in\Z}\left(\int_{S^a_{\eta}(z_{\theta})} w(\theta,z)\pi_z\diff \mu_{a}(z)\right)\diff\delta_{a}(z)\diff\mu(\theta)\\
    =\int_{\Theta} \max_{z_{\theta}\in\Z}\left(\int_{S^{a}_{\eta}(z_{\theta})} w(\theta,z)\pi_z\diff \mu(z)\right)\diff\mu(\theta).
\end{gather*}

Note that, $S^{a}_{\eta}(z_{\theta})=\{z\in a^{-1}(a)\colon \ell(z_{\theta},z)\leq\eta\}=\{z\in \Z\colon \ell(z_{\theta},z)\leq\eta\}=S_{\eta}(z_{\theta})$. Particularly, it simplifies for the discrete case to:
\begin{gather}\label{eq:no_aux}
    \eta\text{-RAD}\leq \sum_{\theta\in\Theta} \max_{z'\in\Z} \sum_{\ell(z',z)\leq\eta}w(\theta,z)\pi_z.
\end{gather}

Moreover, if $\eta=0$ (perfect reconstruction), such as any AIA setting and the original ReRo setting~\cite{Hayes2023Bounding}), \Cref{th:optimal_bound} formula simplifies to:
\begin{equation}
    0\text{-}\mathrm{RAD}\leq \sum_{\theta\in\Theta} \sum_{x\in aux}\max_{\substack{a(z)=x\\w(z,\theta)>0}} w(z,\theta)\,\pi_z.
\end{equation}
Importantly, $0$-RAD is consistently zero for continuous random variables by definition.

Finally, given $|\Z|=m$ and $aux=\{\varnothing\}$, previous equation admits the simplification
\begin{gather}\label{eq:no_auz_uni}
    0\text{-RAD}\leq \sum_{i=1}^{m} \left(\Pr_{\M}(\Theta_{i}\mid z_i)-\Pr_{\M}(\Theta_i)\right)\pi_{i},
\end{gather}
where 
$\Theta_1=\{\theta\in\Theta\colon z_1\in \argmax_{j}w(\theta,z_j)\,\pi_j\}$ and  for every $i\geq1$, $\Theta_{i+1}$ is recursively defined as  \[
\Theta_{i+1}=\{\theta\in\Theta\colon z_{i+1}\in\argmax_{j}w(\theta,z_j)\,\pi_j\}\backslash\bigcup_{k=1}^{i}\Theta_{k}.\]

We illustrate the benefits of~\Cref{th:optimal_bound} for relevant DP mechanisms through the following examples and visualizations in~\Cref{fig:optimal_examples,fig:census_attack}. To compute the RAD bounds for each mechanism, we directly apply the formula from~\Cref{th:optimal_bound} to the corresponding mechanism distribution. The full computational details are provided in~\Cref{ap:rad}.

\begin{figure}[t]
    \centering
        \includegraphics[width=0.55\linewidth]{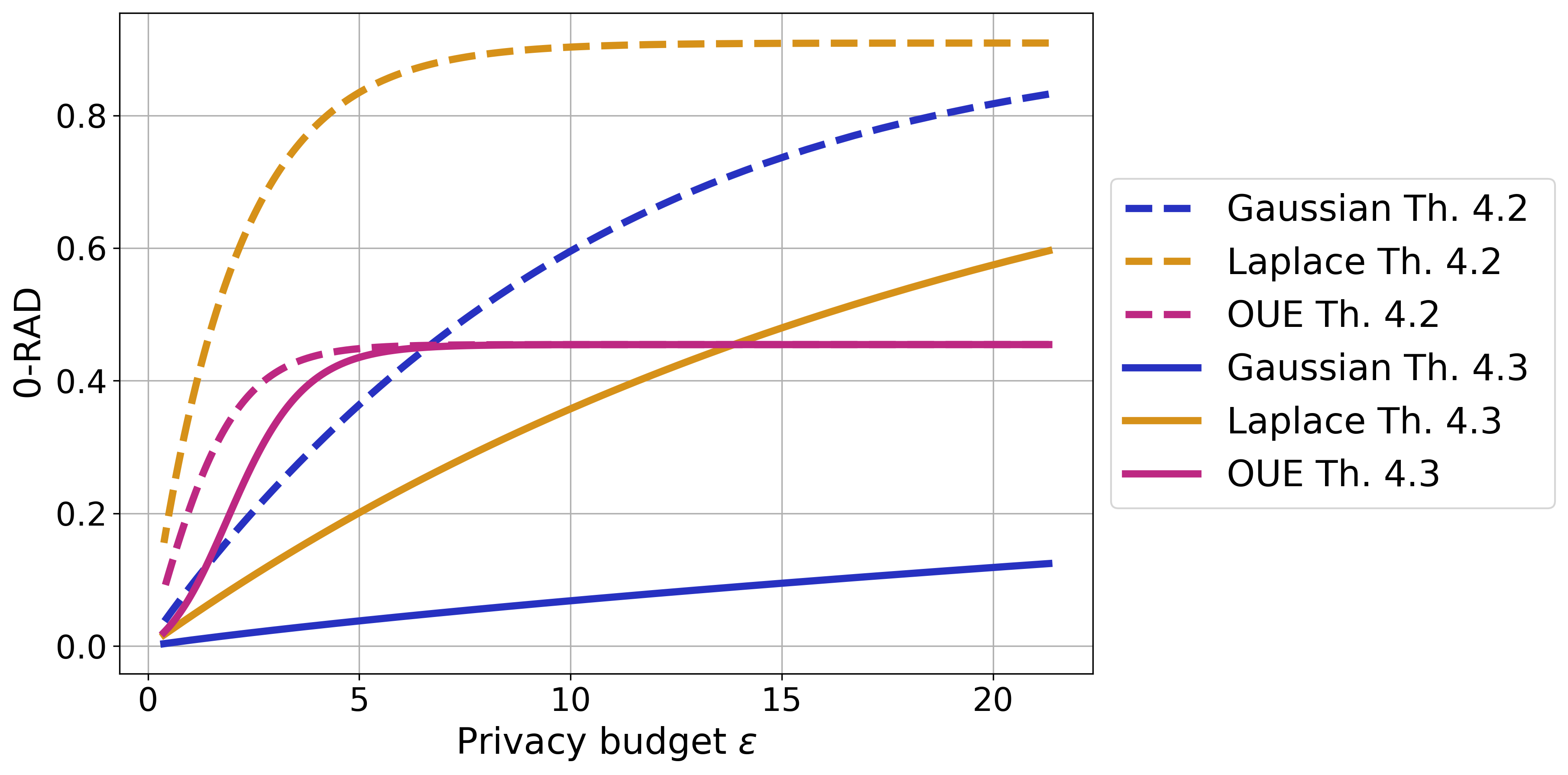}
        \label{fig:compare_discrete}
    \caption[\Cref{th:optimal_bound} Bound on DP Mechanisms]{\Cref{th:optimal_bound} bound for different DP mechanisms with $|\Z| = 11$, $aux=\{\varnothing\}$ and a uniform prior. Importantly, for the same $\varepsilon$, each mechanism offers different levels of attack mitigation, highlighting the need for RAD analysis as a complementary tool to traditional privacy parameters.
    Moreover, in all cases, we observe that the bound in \Cref{th:optimal_bound} improves upon \Cref{th:dp_implies_aux-urero}.}
    \label{fig:optimal_examples}
   
\end{figure}
\begin{example}\label{ex:grr_short}
The generalized randomized response mechanism (GRR)~\cite{Kairouz2016Discrete} is an LDP mechanism that outputs the true record $z_1$ with probability $p=\e^{\varepsilon}/(\e^{\varepsilon}+m-1)$ and any other record $z_0\neq z_1$ with probability $q=(\e^{\varepsilon}+m-1)^{-1}$. Since, $p\geq q$ for all $\varepsilon\geq 0$, 
    \begin{gather}
        w(\theta,z)=\begin{cases}
            (p-q)(1-\pi_{\theta}) &\text{ if } z=\theta\\
            (q-p)\pi_{\theta} &\text{ otherwise,}
        \end{cases}
    \end{gather}
 and $w(z,\theta)>0$ iff $z=\theta$. Hence, applying~\Cref{th:optimal_bound} for $a(z)=z$:
 \begin{gather*}
    \eta\text{-RAD}\leq \sum_{\theta}(p-q)(1-\pi_{\theta})\pi_{\theta}=\frac{\e^{\varepsilon}-1}{\e^{\varepsilon}+m-1} (1-\kappa_{\pi})=\mathrm{TV}(1-\kappa_{\pi}).
 \end{gather*}

Hence, the advantage of an attacker only depends on the chosen $\varepsilon$, the total universe size $|\Z|=m$ and the initial distribution. Particularly, given a maximum risk threshold RAD $\leq\gamma$, we can choose $\varepsilon$ following:
\[
\varepsilon=\ln\frac{1+\gamma\frac{m-1}{1-\kappa_{\pi}}}{1-\frac{\gamma}{1-\kappa_{\pi}}}.
\]

For instance, to guarantee RAD below $0.1$ on binary queries (with uniform prior) the user must set $\varepsilon=\ln(1.5)\approx0.405 $, while for the same RAD in a query with $m=100$ possibilities must select $\varepsilon=\ln(12.2087)\approx2.503$.

Now if we consider a reconstructions attack without target-specific auxiliary knowledge, i.e., $aux=\{\varnothing\}$, we obtain
  \begin{gather*}
    \eta\text{-RAD}=(p-q)(1-\sum_{\theta}\pi_{\theta}\inf_{\ell(z_{\theta},\theta)\leq\eta}\Pr_{Z\sim\pi}[\ell(Z,z_{\theta})\leq\eta]).
 \end{gather*}

Hence, the advantage of such attacker is always less than one with full-knowledge. However, it is not much worse, since for instance considering $\eta=0$ and a uniform distribution we get exactly the same formula.
\end{example}
\begin{example}\label{ex:oue_short} The optimal unary encoding (OUE) mechanism~\cite{Wang2017Locally} maps each input $z\in\Z$ to an $m$-dimensional one-hot  binary vector and perturbs each bit independently. For each position $i\in[m]$, the obfuscated vector $\theta$ is sampled such that $\Pr[\theta_i=1]=1/2$ if $i=z$, and $\Pr[\theta_i=1]=q=\frac{1}{\e^{\varepsilon}+1}$ otherwise. Denoting $p=1-q$, according to~\Cref{th:optimal_bound}, we obtain that, for $a(z)=z$:
\[
\eta\text{-}\mathrm{RAD}\leq \frac{1}{2}\frac{\e^{\varepsilon}-1}{\e^{\varepsilon}+1}(1-\kappa_{\pi})=\mathrm{TV}(\mathrm{OUE})(1-\kappa_{\pi}).
\]

First, note that for the same attack and prior distribution, OUE provides a different level of protection than GRR. In particular, while increasing 
$\varepsilon$ in GRR always increases the attacker’s advantage---approaching 1 as 
$\varepsilon\to\infty$---in the case of OUE, the attacker’s advantage is upper bounded by $0.5$, regardless of how large 
$\varepsilon$ becomes. This illustrates that 
$\varepsilon$
alone does not capture the full picture: mechanisms with the same 
$\varepsilon$ can yield markedly different levels of attack mitigation.

If we consider $aux=\{\varnothing\}$, then the bound becomes:
\[
0\text{-}\mathrm{RAD}\leq \frac{p-q}{2p}\left(\sum_{i=1}^m p^{m-i}\pi_i(1-\pi_i) -q\sum_{i=1}^m p^{m-i}\pi_i\sum_{z=1}^{i-1}\pi_z\right)
\]
which in particular for $\pi=U[m]$:
    \[
    0\text{-}\mathrm{RAD}\leq\frac{(2p-1)\left(1-p^{\,m-1}\right)}{2m(1-p)}=\frac{\e^{\varepsilon}-1}{2m}\left(1-\left(\frac{\e^{\varepsilon}}{1+\e^{\varepsilon}}\right)^{\left(m-1\right)}\right).
\]

Note that when $\varepsilon\to\infty$ previous bound converges to $\frac{m-1}{2m}$, hence even if we keep reducing the noise (increasing $\varepsilon$), the attacker's advantage is limited. We plot this bound in~\Cref{fig:optimal_examples}.
\end{example}
\begin{example}\label{ex:ss_short}
In the subset selection mechanism (SS)~\cite{Min2018Optimal} users report a subset $\theta \subseteq \Z=\{z_1,\dots,z_m\}$ containing their true value $z$ with probability $p = \frac{\omega \e^\varepsilon}{\omega \e^\varepsilon + m - \omega}$, where $\omega = |\theta| = \max\left(1, \left\lfloor \frac{m}{\e^\varepsilon + 1} \right\rfloor \right)$. The subset is completed by sampling uniformly from $\Z \setminus \{z\}$. According to~\Cref{th:optimal_bound} we obtain that for $\pi=U[m]$
\[
\text{0-RAD}\leq \frac{pm-\omega}{m\omega}.
\]
Once again, we obtain a direct formula to calibrate the mechanism parameters (in this case, 
$p$) to achieve a desired RAD. Furthermore, the protection offered by SS against reconstruction attacks lies between that of GRR, which provides weaker protection, and OUE, which provides stronger protection, as illustrated in~\Cref{fig:LDP_attacks}.
\end{example}
\begin{example}\label{ex:laplace_short} The Laplace mechanism adds Laplace noise with scale $b=\Delta q/\varepsilon$ to the query value $q(D)\in\R$~\cite{Dwork2014Algorithmic}.
If $\Z=\{z_1,\dots,z_m\}$ is uniformly distributed and $\Delta q=1$ applying \Cref{th:optimal_bound} we obtain
\[
0\text{-RAD}\leq \frac{m-1}{m}\left(1-\e^{-\frac{\varepsilon}{2(m-1)}}\right).
\] 
First, we observe that the Laplace mechanism provides stronger protection against reconstruction attacks than OUE for small values of $\varepsilon$. For example, as shown in~\Cref{fig:optimal_examples}, for all $\varepsilon \in [0,14]$ the Laplace mechanism achieves lower RAD than OUE on a data domain with $|\Z|=11$.

Moreover, we derive a direct calibration method for the Laplace mechanism. As illustrated in~\Cref{fig:laplace_accu}, calibrating 
$\varepsilon$ according to the maximum admissible risk using our approach yields significantly higher accuracy compared to the state-of-the-art method based on ReRo.
\end{example}

\begin{example}\label{ex:gaussian_short}
The Gaussian mechanism adds Gaussian noise $\mathcal{N}(0,\sigma)$ to the query value $q(D)\in\R$~\cite{Balle2022Reconstructing}.
Given $\Phi$ the CDF of the standard normal distribution, if $\Z=\{z_1,\dots,z_m\}$ is uniformly distributed and $\Delta q=1$, applying \Cref{th:optimal_bound} we obtain 
\begin{gather*}
    0\text{-RAD}
    \leq \frac{m-1}{m} \left(2\Phi\!\Big(\frac{1}{2\sigma(m-1)}\Big)\;-1\right).
\end{gather*}

We plot this bound in~\Cref{fig:optimal_examples} alongside the corresponding OUE and Laplace bounds under the same attack model. The comparison shows that the Gaussian mechanism provides substantially stronger protection against reconstruction attacks without auxiliary knowledge---for a universe of size $11$ and a uniform prior---than both OUE and Laplace. 
\end{example}

\begin{figure}
    \centering
\includegraphics[width=0.35\linewidth]{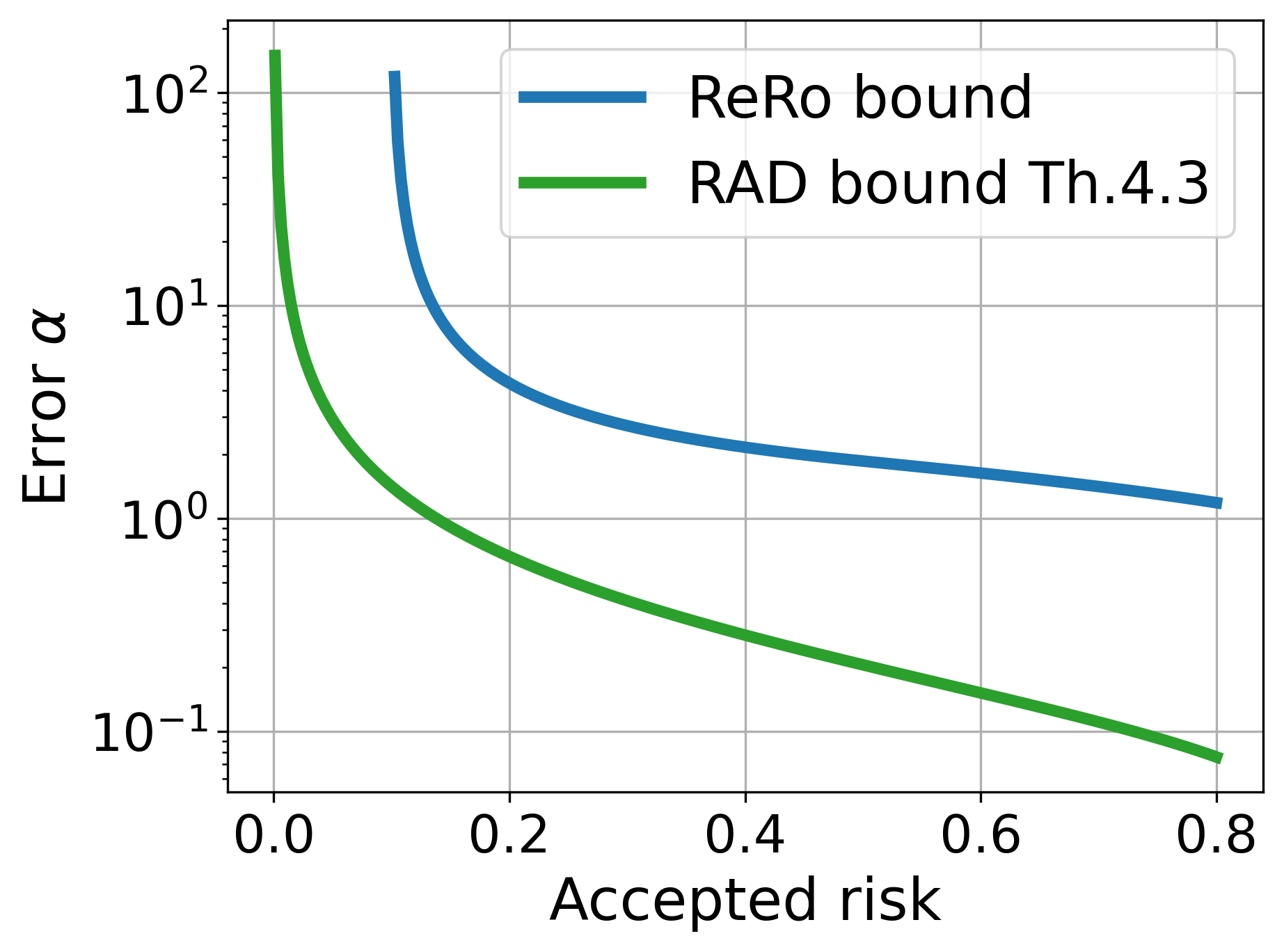}
    \caption[RAD Utility Improvement]{Upper bound on the Laplace mechanism query error (utility) at $95\%$ confidence when the noise is calibrated using ReRo vs.\ RAD. We see that for the same risk estimation, calibrating with using RAD improves utility. }
    \label{fig:laplace_accu}
    
\end{figure}

These examples highlight both the practical applicability of~\Cref{th:optimal_bound} for estimating reconstruction risk in real-world settings and the importance of conducting a dedicated RAD analysis of DP mechanisms.
First, as illustrated in~\Cref{fig:optimal_examples}, mechanisms with the same privacy parameters can provide substantially different levels of protection in terms of risk mitigation. This observation underscores the need for resilience analysis beyond the $\varepsilon$-based criterion.
Moreover, these examples demonstrate that~\Cref{th:optimal_bound} yields simple and explicit characterizations of the RAD for common DP mechanisms, directly relating their noise parameters to the resulting level of risk mitigation.

Moreover, in \Cref{fig:optimal_examples} we see the improvement when we target specific auxiliary knowledge instead of using our worst-case bound (\Cref{th:dp_implies_aux-urero}). 
Hence,~\Cref{th:optimal_bound} offers an improved noise calibration method to ensure protection against real attacks, when the auxiliary knowledge is well defined. 
For instance, when the entire record is considered private ($aux=\{\varnothing\}$); alternatively, when a specific attribute $y$ is deemed sensitive, we consider all the remainder record public (we denote it as
$a(z)=x\backslash y$).

Importantly, we illustrate in~\Cref{fig:laplace_accu} the utility gain of noise calibration using our RAD bounds compared to using  the best existing ReRo bound~\cite{Hayes2023Bounding}, showing the benefit of our bounds for system design. Specifically, we consider $aux=\{\varnothing\}$---allowing comparison with~\cite{Hayes2023Bounding}. We plot the upper bound on the Laplace mechanism's query error that can be guaranteed with $95\%$ confidence, for $|\Z|=10$ and $\Delta=1$, showing a substantial improvement in utility enabled by our RAD-based calibration.

Crucially, \Cref{th:optimal_bound} is universally tight: for any mechanism and auxiliary knowledge, there exists an attack achieving the bound, so it cannot be further improved. We illustrate this by explicitly constructing such an attack in Algorithm~\ref{alg:OptimalAIA1}, proving the existence of an optimal adversary for any auxiliary model.

The attack strategy is conceptually simple yet highly effective.
Let's start with the case of $\eta=0$, i.e., perfect reconstruction.
In the fully informed setting---where the adversary knows the entire target record (as in an informed MIA)---the optimal strategy is to declare the target a member whenever the mechanism’s output provides any positive evidence of participation, that is, whenever
\[
w(\theta,z) = p_{\mathcal M}(\theta \mid z) - p_{\mathcal M}(\theta) > 0.
\]
Intuitively, if the observed output is more likely under the target record than under the prior distribution, the adversary should infer membership.
If there is more than one candidate $z$ sharing the same auxiliary knowledge, $a(z)=x$, (e.g., several users may shared a common attribute) the attacker can not optimize for all at the same time, therefore select $\tilde{z}$ such that it maximizes the posterior weight $w(\theta,z)\,\pi_z$, as long as it provides positive evidence.
In the extreme case, when $aux=\{\varnothing\}$ (no auxiliary information), the attacker cannot narrow the candidate set and so the optimal reconstruction selects
$z^{*}\in \argmax_{z\in\Z}w(z,\theta)\,\pi_z$,
i.e., any record that maximizes the posterior probability given $\theta$. 
When the attacker does not require exact reconstruction but is satisfied with producing a candidate within a controlled error $\eta$ of the true record, the optimal strategy retains a similar structure. However, rather than comparing records based on the posterior probability of a single output, the analysis evaluates the posterior mass of their associated success sets. The attacker then selects the record $\tilde{z}$ whose success set $S_{\eta}(\tilde{z})$ attains the largest posterior probability given the observed output.

\begin{algorithm}[t]
\small
\caption{Optimal Attack}
\label{alg:OptimalAIA1}
\SetKwInOut{Input}{Input}\SetKwInOut{Output}{Output}
\Input{$\theta$, $\eta$ and $a(z)=x$}
\Output{$\tilde{z}$}
Compute $a^{-1}(x)=\{z\colon a(z)=x\}$\\
\For{$z'\in \Z$}{
      $\displaystyle \mathcal{W}^z_{\eta}(z')=\sum_{z\in a^{-1}(x)\colon \ell(z,z')\leq\eta}w(\theta,z)\pi_z$\;
}
Select $\tilde{z}\in  \argmax_{z'} \mathcal{W}^z_{\eta}(z')$\, (at random)
\end{algorithm}

This result is particularly relevant, as it implies that, for a given risk tolerance, the utility of a mechanism cannot exceed what our method achieves; in other words, our approach yields optimal noise calibration.

\begin{corollary}
[Attack Optimality]\label{co:optimal_attack}
    Given the conditions as in \Cref{th:optimal_bound},  Algorithm~\ref{alg:OptimalAIA1} achieves the highest attainable $\eta$-RAD. 
\end{corollary}

\begin{proof}
Following Algorithm~\ref{alg:OptimalAIA1}, given $\theta, x$, the attack always select (at random) an output from the set:
\begin{equation}\label{eq:defi}
    S_{\theta}^{x}
=
\argmax_{\tilde{z}\in\mathcal Z}
\int_{a^{-1}(x)}
\boldsymbol{1}_{\{\ell(\tilde{z},z_1)\le \eta\}}\,
w(z_1,\theta)\,\pi_{z_1}\,
\diff\mu_x(z_1).
\end{equation}

Hence, the attack $A$ verifies
\[
\Pr_A\!\big(A(\theta,x)\in S_{\theta}^{x}\big)=1,
\qquad
p_A(\tilde{z}\mid\theta,x)
=
\frac{\boldsymbol{1}_{\{\tilde{z}\in S_{\theta}^{x}\}}}{\mu(S_{\theta}^{x})}, 
\]
and for all $\tilde{z}\in S_{\theta}^{x}$,
\[
\int_{a^{-1}(x)}
\boldsymbol{1}_{\{\ell(\tilde{z},z_1)\le \eta\}}\,
w(z_1,\theta)\,\pi_{z_1}\,
\diff\mu_x(z_1)=\max_{s\in\Z} \int_{S^x_{\eta}(s)}\,
w(z_1,\theta)\,\pi_{z_1}\,
\diff\mu_x(z_1)\equiv I_{x,\theta}.
\]

Computing RAD according to the reformulation in~\Cref{eq:truq}, we obtain
{\small
\begin{align}
\MoveEqLeft[2]\eta\text{-}\mathrm{RAD}(A)
=
\int_{\Z}\int_{\Theta}
\Pr_{A}\big(S_{\eta}(z_1)\mid \theta,a(z_1)\big)\,
w(z_1,\theta)\,\pi_{z_1}\,
\diff\mu(\theta)\diff\mu(z_1)
\notag\\
=&
\int_{\Theta}\int_{aux}\int_{a^{-1}(x)}
\Pr_{A}\big(S_{\eta}(z_1)\mid \theta,x\big)\,
w(z_1,\theta)\,\pi_{z_1}\,
\diff\mu_x(z_1)\diff\nu(x)\diff\mu(\theta)
\notag\\
=&
\int_{\Theta}\int_{aux}\int_{a^{-1}(x)}\left(
\int_{S_{\eta}(z_1)}
\frac{\boldsymbol{1}_{\{\tilde{z}\in S_{\theta}^{x}\}}}{\mu(S_{\theta}^{x})}\,
w(z_1,\theta)\,\pi_{z_1}\,
\diff\mu(\tilde{z})\right)\diff\mu_x(z_1)\diff\nu(x)\diff\mu(\theta)
\notag\\
=&
\int_{\Theta}\int_{aux}\int_{a^{-1}(x)}
\left(\int_{\Z}
\boldsymbol{1}_{\{\ell(\tilde{z},z_1)\leq\eta\}}\frac{\boldsymbol{1}_{\{\tilde{z}\in S_{\theta}^{x}\}}}{\mu(S_{\theta}^{x})}\,
w(z_1,\theta)\,\pi_{z_1}\,
\diff\mu(\tilde{z})\right)\diff\mu_x(z_1)\diff\nu(x)\diff\mu(\theta)
\notag\\
=&
\int_{\Theta}\int_{aux}
\int_{\Z}
\frac{\boldsymbol{1}_{\{\tilde{z}\in S_{\theta}^{x}\}}}{\mu(S_{\theta}^{x})}\,\left(\int_{a^{-1}(x)} \boldsymbol{1}_{\{\ell(\tilde{z},z_1)\leq\eta\}}
w(z_1,\theta)\,\pi_{z_1}\,\diff\mu_x(z_1)\right)
\diff\mu(\tilde{z})\diff\nu(x)\diff\mu(\theta)
\notag\\
=&
\int_{\Theta}\int_{aux}
\int_{\Z}
\frac{\boldsymbol{1}_{\{\tilde{z}\in S_{\theta}^{x}\}}}{\mu(S_{\theta}^{x})}\, I_{x,\theta}
\diff\mu(\tilde{z})\diff\nu(x)\diff\mu(\theta)
\notag\\
=&
\int_{\Theta}\int_{aux} I_{x,\theta} \diff\nu(x)\diff\mu(\theta)
\left(\int_{\Z}
\frac{\boldsymbol{1}_{\{\tilde{z}\in S_{\theta}^{x}\}}}{\mu(S_{\theta}^{x})}\, 
\diff\mu(\tilde{z})\right)
\notag\\
=&
\int_{\Theta}\int_{aux} I_{x,\theta} \diff\nu(x)\diff\mu(\theta)
\notag\\
=&
\int_{\Theta}\int_{aux}
\max_{s\in\mathcal Z}
\int_{S^x_{\eta}(s)}
w(z_1,\theta)\,\pi_{z_1}
\diff\mu_x(z_1)
\diff\nu(x)\diff\mu(\theta), \notag
\end{align}
}
which according to~\Cref{th:optimal_bound}, coincides with the maximum attainable bound. 
\end{proof}

\Cref{co:optimal_attack} directly establishes that \Cref{th:optimal_bound} is universally tight, i.e., for every mechanism, prior auxiliary knowledge and error threshold, \Cref{th:optimal_bound} exactly determines the maximum achievable RAD.  Moreover,~\Cref{th:dp_implies_aux-urero} is tight, since there exists at least one mechanism (GRR,~\Cref{ex:grr_short}) for which \Cref{th:dp_implies_aux-urero} is achieved. We further validate that this is not an isolated case by empirically demonstrating tightness on additional mechanisms, such as DP-SGD (see \Cref{fig:dp-sgd_mia_mnist}).

Beyond the theoretical contribution, our results provide a practical tool: a general attack algorithm that practitioners can directly use to evaluate the privacy risks of their systems or the tightness of their bounds. As a concrete demonstration, we apply this attack in the context of LDP auditing (see~\Cref{sec:auditing}) and to assess empirical risk and tightness of our bounds in (see~\Cref{sec:experiments}). We also provide the application of our optimal attack in the specific case of DP-SGD:
\begin{example}[Optimal Attack on DP-SGD]

Our analysis of DP-SGD is motivated by its central role in private learning: distributionally robust attacks were first introduced in this context~\cite{Yeom2017Privacy}, and DP-SGD remains the most widely used algorithm in practice~\cite{Abadi2016Deep}.  In particular, we study the reconstruction setting considered by \citeauthor{Hayes2023Bounding}~\cite{Hayes2023Bounding}, where the adversary attempts to reconstruct the target record $z^*$ from a candidate set $\{z_1,\dots,z_m\}$ with uniform prior using access to the privatized gradients $\{\bar{g}_1,\dots,\bar{g}_T\}$ released during training, i.e., white-box setting. Note that in each DP-SGD iteration $\bar{g}_t$ is obtained as
\[
\bar{g}_t=\sum_x \mathrm{clip}_C(\nabla_{\theta}\ell(\theta_t,z))+\mathcal{N}(0,c^2\sigma^2 I),
\]
 where $\sigma$ is the noise scale, $\mathrm{clip}_C(\Vec{v})=\Vec{v}\min(1,\frac{C}{\|\Vec{v}\|_2})$ and $\theta_t$ the released weights in the previous iteration. 

Given $\theta=(\theta_1,\dots,\theta_T)$, our optimal attack is determined by
$
\argmax_{z: a(z)=x}\, w(\theta,z),
$ and its sign, i.e., whether $w(\theta,z)>0$ or not, for each candidate $z$ and auxiliary knowledge $z$. Concretely, since the public dataset $D_-$ is known, we can isolate the noisy contribution of the target’s gradient at iteration $t$:
\[
g_t = \bar{g}_t - \sum_{z \in D_-} \mathrm{clip}_C(\nabla_\theta \ell(\theta_t, z)),
\]
and simplify $w$ maximization to
\begin{gather}
    \argmax_{z: a(z)=x}\, w(\theta,z)=\argmax_{z: a(z)=x}\sum_{t} W(g_t, \mathrm{clip}_C(\nabla_\theta \ell(\theta_t, z)))
\end{gather}
where $W(u,v)=\langle u,v\rangle-\frac{1}{m}\sum_{z}\langle u,\mathrm{clip}_C(\nabla_\theta \ell(\theta_t, z))\rangle$, since $W$  preserves the sign and $\argmax$ of $w$. We present the pseudo-code of the optimal attack in Algorithm~\ref{alg:OptimalAIA2}.

Indeed, given $\theta,z$, under DP-SGD the privatized gradient at step $t$ is
\[
g_t \sim \mathcal N\!\left(\mu_z,\, C^2\sigma^2 I\right),
\qquad
\mu_z = \mathrm{clip}_C(\nabla_\theta \ell(\theta_t, z)),
\]
where $C$ is the clipping parameter and $I$ the identity function of dimension $d$, corresponding to the dimension of the gradients. Hence the likelihood is
\[
p_{\mathcal M}(g_t\mid z) \;=\; \underbrace{\frac{1}{(2\pi C^2\sigma^2)^{d/2}}}_{A}
\exp\!\Big(\underbrace{-\frac{1}{2C^2\sigma^2}}_{B}\|g_t-\mu_z\|^2\Big),
\]
where both $A,B$ are independent from $z$. Consequently,
\begin{align}
    w(g,z)=\prod_{t}p_{\mathcal M}(g_t\mid z)-\prod_{t} p_{\mathcal M}(g_t)&>0\Leftrightarrow \\
    A^{T}\left(\prod_{t}\e^{B\langle g_t,\mu_z\rangle}-\prod_{t}\frac{1}{m}\sum_{i}\e^{B\langle g_t,\mu_{z_i}\rangle}\right)&>0\Leftrightarrow\\
    \e^{B\sum_{t}\langle g_t,\mu_z\rangle}&>\prod_{t}\frac{1}{m}\sum_{i}\e^{B\langle g_t,\mu_{z_i}\rangle}\Leftrightarrow\\
    B\sum_{t}\langle g_t,\mu_z\rangle&>\sum_{t}\ln\left(\frac{1}{m}\sum_{i}\e^{B\langle g_t,\mu_{z_i}\rangle}\right)\Leftrightarrow
    \\
    B\sum_{t}\langle g_t,\mu_z\rangle&>  \sum_{t}\frac{1}{m}\sum_{z}\ln(\e^{B\langle g_t,\mu_{z_i}\rangle})\Leftrightarrow\label{eq:ln_chachi}\\
    B\sum_{t}\langle g_t,\mu_z\rangle&> \sum_{t} \frac{B}{m}\sum_{i}\langle g_t,\mu_{z_i}\rangle\Leftrightarrow\\
    \sum_{t}\langle g_t,\mu_z\rangle-\sum_{t}\frac{1}{m}\sum_{z_i}\langle g_t,\mu_{z_i}\rangle&>0\Leftrightarrow\\
    \sum_{t} W(g_t,z)&>0.
\end{align}
Where~\Cref{eq:ln_chachi} follows from the application of Jensen's inequality to the logarithm.
Moreover, $\argmax_{z} w(g,z)=\argmax_{z}\ln(p_{\M}(g,z))=\argmax_{z} \sum_{t}\ln p_{\M}(g_t,z)$, where 
\[
\ln p_{\mathcal M}(g_t \mid z_i) \propto 
- \tfrac{1}{2C^2\sigma^2}\,\|g_t - \mathrm{clip}_C(\nabla_\theta \ell(\theta_t, z_i))\|^2.
\]

Expanding the squared norm leads to
\[
\|g_t\|^2 + \|\mathrm{clip}_C(\nabla_\theta \ell(\theta_t, z_i))\|^2 
- 2 \langle g_t, \mathrm{clip}_C(\nabla_\theta \ell(\theta_t, z_i)) \rangle.
\]

The term $\|g_t\|^2$ is independent of $z_i$, and the term
$\|\mathrm{clip}_C(\nabla_\theta \ell(\theta_t, z_i))\|^2$ is bounded by $C^2$ (often nearly constant across candidates). Therefore, maximizing the log-likelihood is equivalent to maximizing
\[
\langle g_t, \mathrm{clip}_C(\nabla_\theta \ell(\theta_t, z_i)) \rangle.
\]

Consequently, our optimal attack can be simplified by using $W(g_t,z)$ instead of $w(g_t,z)$.

When $aux=\{\varnothing\}$, our optimal attack coincides with the attack presented in~\cite{Hayes2023Bounding}. Whereas they identified such an attack as the empirically best, we formally establish that this choice is indeed optimal. Moreover, we extend the optimal attack for any attacker that has target-specific auxiliary information. In particular, our optimal attack for attackers with $aux\neq\{\varnothing\}$ is empirically tested in~\Cref{sec:experiments}, showing that previous bounds for ReRo indeed do not hold for attackers with target-specific auxiliary knowledge.

\begin{algorithm}
\small
\caption{Optimal Attack for DP-SGD }
\label{alg:OptimalAIA2}
\SetKwInOut{Input}{Input}\SetKwInOut{Output}{Output}
\Input{ $\theta=(\theta_1,\dots,\theta_T)$, $a(z)=x$ and $g=(g_1,\dots,g_T)$ }
\Output{ $\Tilde{z}$}
\For{$z\colon a(z)=x$}{
  compute $\sum_{t} W(\bar{g}_t, \mathrm{clip}_C(\nabla_\theta \ell(\theta_t, z)))$\;
}
 Select $z^{*}=\argmax_{z\colon a(z)=x} \sum_{t} W(\bar{g}_t, \mathrm{clip}_C(\nabla_\theta \ell(\theta_t, z)))\pi_z$\;
\If{$W(\bar{g}_t,z^{*}) > 0$}{
  $\Tilde{z}=z^*$\;
}
\Else{
  $\Tilde{z}\longleftarrow U[\Z\backslash\{z\colon a(z)=x\}]$\;
}
\end{algorithm}
\end{example}

The bounds presented in this section offer concrete guidance for algorithm design. They can be directly leveraged for noise calibration, achieving rigorous privacy guarantees while maximizing utility. 
In particular, they induce a simple protocol for practitioners. First, one must specify which information is deemed private (e.g., the full record, a subset of attributes, or membership), which determines the choice of the auxiliary information $aux$ and $a\colon \Z \to aux$. 
Second, if prior knowledge about the distribution of $\Z$ is available, it should be encoded in a distribution 
$\pi$. If this is not the case, however, one must resort to the worst-case prior; otherwise, the attacker’s risk may be underestimated. This worst-case prior typically corresponds to $\pi_z = \pi_y = 1/2$  for the two records that are easiest to distinguish (see~\Cref{ex:grr_short,ex:oue_short} and \Cref{fig:census_attack}). Nevertheless, even when the worst-case prior cannot be explicitly identified, the total variation bound given in~\Cref{th:dp_implies_aux-urero} provides a safe upper bound for any choice of prior and $aux$.

Third, the resulting RAD of the mechanism can be computed using \Cref{th:optimal_bound}---an auxiliary-dependent bound proven to be universally tight, or upper-bounded by a worst-case guarantee when the nature of $aux$ is unknown (\Cref{th:dp_implies_aux-urero}). Finally, by inverting the corresponding bound, one can directly derive the noise-injection parameters that meet a prescribed risk level. Since our bounds are tight, this procedure yields mechanisms that are utility-optimal for any given risk acceptance.

Note that while the closed form of \Cref{th:optimal_bound} is easy to derive for discrete data, this may not hold for continuous data, where the bound involves Lebesgue integrals. In such case, the bound can be evaluated numerically using a nested Monte Carlo procedure as we show in~\Cref{ap:montecarlo}. 
Since numerical approximations introduce error, as a safer alternative, one may always use our closed-form upper bound in \Cref{th:dp_implies_aux-urero}. However, this bound can be overly conservative when $aux=\{\varnothing\}$, motivating the tighter closed-form upper-bounds derived in the next section, which avoid numerical procedures even for continuous data.

\section{\texorpdfstring{$\eta$}{eta}-RAD Upper Bounds under \texorpdfstring{$aux=\{\varnothing\}$}{no target-specific auxiliary knowledge}}\label{sec:no_aux}
Our bound in \Cref{th:optimal_bound} is universally tight, but two limitations remain. First, it requires full knowledge of the mechanism, making it suitable for noise calibration; however, in DP auditing, we often have only query access (e.g., auditing external software) without insight into the internal protocol~\cite{gorla2025estimating}. Second, the bound lacks a closed form hence may rely on numerical approximation, particularly for continuous data domains.  
Consequently, in this section we provide black-box bounds for the case $aux = \{\varnothing\}$, both because this is the standard assumption in prior DP auditing~\cite{Arcolezi2024Revealing, Mahloujifar2024Auditing} and data reconstruction studies~\cite{Balle2022Reconstructing,Hayes2023Bounding}, and because it makes practical sense: for other auxiliary-information models, one can always rely on the closed-form bound provided by~\Cref{th:dp_implies_aux-urero}.

%%%% AUX = \emptyset
First, we present a general bound that applies to any reconstruction setting as long as no target-specific auxiliary knowledge is available. For this purpose, we introduce $\kappa_{\pi,\ell}^{-}(\eta)$ as the infimum counterpart of $\kappa_{\pi,\ell}^{+}(\eta)$, formally defined as
\begin{equation}
    \kappa_{\pi,\ell}^{-}(\eta) = \inf_{z_0 \in \Z} \Pr_{Z \sim \pi}\big[\ell(Z,z_0)\leq \eta\big],
\end{equation}
representing the success probability of an oblivious attacker attempting to reconstruct the most difficult target only using $\pi$.

\begin{theorem}\label{th:f-DP}
    If a mechanism $\M\colon\Z^{n}\to\D(\Theta)$ satisfies $f$-DP, then for any attack with $aux=\{\varnothing\}$, $A\colon \Theta\to\D(\Z)$, it satisfies 
    \[
    \eta\text{-}\mathrm{RAD}\leq \max_{\alpha\in[\kappa^{-}_{\pi,\ell}(\eta),\kappa^{+}_{\pi,\ell}(\eta)]} 1-f(\alpha)-\alpha.
    \]
    
    If $\Z$ is discrete, then it also holds
     \[
    \eta\text{-}\mathrm{RAD}\leq (1-\kappa_{\pi})\max\limits_{\alpha\in[0,\frac{\kappa^{+}_{\pi,\ell}(\eta)}{1-\kappa_{\pi}}]} 1-f(\alpha)-\alpha .     
     \]
\end{theorem}
\begin{proof}
\citeauthor{Kifer2022Bayesian}~\cite[p.23]{Kifer2022Bayesian} showed that 
for any $S\subseteq\Theta$, for any $f$-DP mechanism, and $z_0,z_1\in\Z$, 
\begin{gather}\label{eq:kifer_f}
    \Pr_{\M}(S\mid D_{z_1})\leq 1-f(\Pr_{\M}(S\mid D_{z_0}).
\end{gather}

Moreover, since $f$ is convex (see~\Cref{sec:background}), applying Jensen's inequality:
\begin{gather}\label{eq:f_jensen}
    f(\mathbb{E}_{Z}[Z])\leq \mathbb{E}_{Z}[f(Z)]\Leftrightarrow -\mathbb{E}_{Z}[f(Z)]\leq -f(\mathbb{E}_{Z}[Z]) .
\end{gather}

Combining both~\Cref{eq:kifer_f} and~\Cref{eq:f_jensen} we obtain
\begin{align}
   \MoveEqLeft[3]\eta\text{-RAD} 
   = \Pr_{Z_1 \sim \pi}
        [ \A_{Z_1}(D_{Z_1})\in S_{\eta}(Z_1)]
    - \Pr_{Z_0,Z_1 \sim \pi}
        [ \A_{Z_1}(D_{Z_0})\in S_{\eta}(Z_1)] \notag\\
   =& \E_{Z_0\sim\pi}\Big[\Pr_{Z_1 \sim \pi}
        [ \A_{Z_1}(D_{Z_1})\in S_{\eta}(Z_1)]
       - \Pr_{Z_1 \sim \pi}
        [ \A_{Z_1}(D_{Z_0})\in S_{\eta}(Z_1)]\Big] \notag\\
   =& \E_{Z_0,Z_1\sim\pi}\Big[
        \Pr[\A_{Z_1}(D_{Z_1})\in S_{\eta}(Z_1)]
       - \Pr[\A_{Z_1}(D_{Z_0})\in S_{\eta}(Z_1)]\Big] \notag\\
   \leq& \E_{Z_0,Z_1\sim\pi}\Big[
        1 - f\!\big(\Pr[\A_{Z_1}(D_{Z_0})\in S_{\eta}(Z_1)]\big)
        - \Pr[\A_{Z_1}(D_{Z_0})\in S_{\eta}(Z_1)]\Big] \notag\\
    =& 1-\E_{Z_1,Z_0}[f\left(\Pr[\A_{Z_1}(D_{Z_0})\in S_{\eta}(Z_1)]\right)]-\E_{Z_1,Z_0}[\Pr[\A_{Z_1}(D_{Z_0})\in S_{\eta}(Z_1)]]\notag\\
    \leq& 1-f\left(\E_{Z_1,Z_0}[\Pr[\A_{Z_1}(D_{Z_0})\in S_{\eta}(Z_1)]]\right)-\E_{Z_1,Z_0}[\Pr[\A_{Z_1}(D_{Z_0})\in S_{\eta}(Z_1)]],\notag
\end{align}
where last inequality follows from~\Cref{eq:f_jensen}.
Therefore, it suffices to determine the interval containing $\E_{Z_1,Z_0}[\Pr[\A_{Z_1}(D_{Z_0})\in S_{\eta}(Z_1)]$. 
 \begin{gather*}
       \mathbb{E}_{Z_1,Z_0\sim\pi}\left[\Pr_{Z \sim \pi}[\A(D_{Z_0})\in S_{\eta}(Z)] \right]\\
       =\int_{\Z}\int_{\Z} \Pr[\A(D_{z_0})\in S_{\eta}(z_1)]\pi_{z_0}\pi_{z_1}\diff z_0\diff z_1\\
       = 
       \int_{\Z}\int_{\Z} \int_{\Z}p_{\A}[z\mid D_{z_0}]\boldsymbol{1}_{\{\ell(z,z_1)\leq \eta\}} \pi_{z_0}\pi_{z_1} \diff z_0 \diff z_1 \diff z \\
       =\int_{\Z} \int_{\Z}p_{\A}[z\mid D_{z_0}]\left(\int_{\Z}\boldsymbol{1}_{\{\ell(z,z_1)\leq \eta\}}\pi_{z_1} \diff z_1 \right) \pi_{z_0} \diff z_0 \diff z \\
       \leq \kappa^{+}_{\pi,\ell}(\eta) \int_{\Z} \int_{\Z}p_{\A}[z\mid D_{z_0}]\pi_{z_0} \diff z_0 \diff z = \kappa^{+}_{\pi,\ell}(\eta).
   \end{gather*}
   and analogous for $\kappa^{-}_{\pi,\ell}(\eta)$ since any attack output $z\in \Z$ and hence it follows by definition.
   Note that last inequality assumes no auxiliary knowledge is available, therefore $p_{\A}[z\mid D_{z_0},a(z_1)]=p_{\A}[z\mid D_{z_0}]$, hence it factors out of the integral with respect to $z_1$.

Now, we prove that, for discrete variables, the bound can be further improved. We follow the same notation as in~\Cref{th:dp_implies_aux-urero}, i.e.,
\begin{equation}\label{eq:exp_con}
\E_{Z_0,Z_1\sim\pi}\left[\boldsymbol{1}_{\{Z_0\neq Z_1\}}\right]=1-\Pr_{Z,Z'\sim\pi}[Z=Z']=\begin{cases}
    1 &\text{ if } \pi \text{ continuous,}\\
    1-\kappa_{\pi}  &\text{ if } \pi \text{ discrete.}\\
\end{cases}    
\end{equation}
and $\sum_{z_1}\sum_{z_0\neq z_1}\frac{\pi_0\pi_1}{(1-\kappa_{\pi})}=1$.
Now, combining~\Cref{eq:kifer_f} and~\Cref{eq:f_jensen} we obtain:
\begin{align}
   \MoveEqLeft[3]\eta\text{-RAD} 
   = \Pr_{Z_1 \sim \pi}
        [ \A_{Z_1}(D_{Z_1})\in S_{\eta}(Z_1)]
    - \Pr_{Z_0,Z_1 \sim \pi}
        [ \A_{Z_1}(D_{Z_0})\in S_{\eta}(Z_1)] \notag\\
   =& \E_{Z_0\sim\pi}\Big[\Pr_{Z_1 \sim \pi}
        [ \A_{Z_1}(D_{Z_1})\in S_{\eta}(Z_1)]
       - \Pr_{Z_1 \sim \pi}
        [ \A_{Z_1}(D_{Z_0})\in S_{\eta}(Z_1)]\Big] \notag\\
   =& \E_{Z_0,Z_1\sim\pi}\Big[\boldsymbol{1}_{\{Z_0\neq Z_1\}}\big(
        \Pr[\A_{Z_1}(D_{Z_1})\in S_{\eta}(Z_1)]
       - \Pr[\A_{Z_1}(D_{Z_0})\in S_{\eta}(Z_1)]\big)\Big] \notag\\
   \leq& \E_{Z_0,Z_1\sim\pi}\Big[\boldsymbol{1}_{\{Z_0\neq Z_1\}}\big(
        1 - f\!\big(\Pr[\A_{Z_1}(D_{Z_0})\in S_{\eta}(Z_1)]\big)
        - \Pr[\A_{Z_1}(D_{Z_0})\in S_{\eta}(Z_1)]\big)\Big] \notag\\
    =& 
     \E_{Z_1,Z_0\sim\pi}[\boldsymbol{1}_{\{Z_0\neq Z_1\}}]-  \E_{Z_1,Z_0\sim\pi}\Big[\boldsymbol{1}_{\{Z_0\neq Z_1\}}f\!\big(\Pr[\A_{Z_1}(D_{Z_0})\in S_{\eta}(Z_1)]\big)\Big]\notag\\
     &\quad   -  \E_{Z_1,Z_0\sim\pi}\Big[\boldsymbol{1}_{\{Z_0\neq Z_1\}}\Pr[\A_{Z_1}(D_{Z_0})\in S_{\eta}(Z_1)]\big) \Big]\notag\\
     &=(1-\kappa_{\pi})\Big(
     1- \E_{Z_1,Z_0}[\boldsymbol{1}_{\{Z_0\neq Z_1\}}\frac{f\left(\Pr[\A_{Z_1}(D_{Z_0})\in S_{\eta}(Z_1)]\right)}{(1-\kappa_{\pi})}]\notag\\
     &\quad\quad\quad\quad\quad-\E_{Z_1,Z_0}[\boldsymbol{1}_{\{Z_0\neq Z_1\}}\frac{\Pr[\A_{Z_1}(D_{Z_0})\in S_{\eta}(Z_1)]}{(1-\kappa_{\pi})}]
    \Big)\notag\\
    &\leq (1-\kappa_{\pi})\Big(1-
     f(\sum_{z_1}\sum_{z_0\neq z_1}\Pr[\A_{z_1}(D_{z_0})\in S_{\eta}(z_1)] \tfrac{\pi_0\pi_1}{(1-\kappa_{\pi})})\notag\\
     &\quad\quad\quad\quad\quad-\sum_{z_1}\sum_{z_0\neq z_1}\Pr[\A_{z_1}(D_{z_0})\in S_{\eta}(z_1)] \tfrac{\pi_0\pi_1}{(1-\kappa_{\pi})}
   \Big).\notag
\end{align}
Therefore, the proof follows from the following upper-bound:
 \begin{gather*}
     \sum_{z_1}\sum_{z_0\neq z_1}\Pr[\A_{z_1}(D_{z_0})\in S_{\eta}(z_1)] \tfrac{\pi_0\pi_1}{(1-\kappa_{\pi})}\\
\leq \frac{1}{(1-\kappa_{\pi})}\E_{Z_0,Z_1}[\Pr[\A_{Z_1}(D_{Z_0})\in S_{\eta}(Z_1)]=\frac{\kappa^{+}}{(1-\kappa_{\pi})}. 
\end{gather*}
Concluding both bounds.
\end{proof}

If~\Cref{th:optimal_bound} cannot be computed in closed form, this result provides an upper bound for RAD when $aux=\{\varnothing\}$. It avoids numerical approximation errors and yields a tighter estimate than the conservative upper bound given in~\Cref{th:dp_implies_aux-urero}.

In the following example we see its practical application to Gaussian DP:
\begin{example}\label{ex:maz_alpha}
We consider uniform prior and $\eta=0$, hence $\kappa^{+}=\frac{1}{m}$. Applying~\Cref{th:f-DP} we obtain
\[
0\text{-}\mathrm{RAD}\leq \max_{\alpha\in\left[0,\frac{1}{m-1}\right]}
1-f(\alpha)-\alpha
\;=\;\max_{\alpha\in\left[0,\frac{1}{m-1}\right]}
1-\Phi\!\left(\Phi^{-1}(1-\alpha)-\mu\right)-\alpha \equiv \max_{\alpha\in\left[0,\frac{1}{m-1}\right]}
g(\alpha),
\;
\]

where $\Phi$ and $\varphi$ denote respectively the CDF and PDF of the standard normal distribution.

Using the chain rule and the identity
\[
\frac{d}{d\alpha}\Phi^{-1}(1-\alpha)
=
-\frac{1}{\varphi(\Phi^{-1}(1-\alpha))},
\]
we obtain
\[
\begin{aligned}
g'(\alpha)
&=
-\varphi\!\left(\Phi^{-1}(1-\alpha)-\mu\right)
\cdot
\frac{d}{d\alpha}\!\left[\Phi^{-1}(1-\alpha)-\mu\right]
-1 \\
&=
\frac{\varphi(\Phi^{-1}(1-\alpha)-\mu)}
{\varphi(\Phi^{-1}(1-\alpha))}-1.
\end{aligned}
\]
Moreover, the derivative can be rewritten in closed form.
Recall that the standard normal density is
\[
\varphi(z)=\frac{1}{\sqrt{2\pi}}\,e^{-z^{2}/2}.
\]
Therefore,
\begin{gather}
g'(\alpha)
=
\frac{\varphi(\Phi^{-1}(1-\alpha)-\mu)}
     {\varphi(\Phi^{-1}(1-\alpha))}-1
\label{eq:fprime_ratio}
\\
=
\frac{
\exp\!\left(-\frac{1}{2}(\Phi^{-1}(1-\alpha)-\mu)^2\right)
}{
\exp\!\left(-\frac{1}{2}(\Phi^{-1}(1-\alpha))^2\right)
}-1
\label{eq:fprime_exp1}
\\
=
\exp\!\left(
\mu\,\Phi^{-1}(1-\alpha)-\frac{\mu^2}{2}
\right)-1.
\label{eq:fprime_exp}
\end{gather}

An interior maximizer satisfies $g'(\alpha)=0$, i.e.,
\[
\mu\,\Phi^{-1}(1-\alpha)-\frac{\mu^2}{2}=0.
\]
Because $\mu>0$, the unique solution is
\[
\Phi^{-1}(1-\alpha)=\frac{\mu}{2}\Leftrightarrow \alpha=1-\Phi(\frac{\mu}{2}).
\]
Moreover, since
\[
g'(\alpha)
=
\exp\!\left(
\mu\,\Phi^{-1}(1-\alpha)-\frac{\mu^2}{2}
\right)-1,
\]
we have $g'(\alpha)>0$ for $\alpha<1-\Phi(\mu/2)$ and
$g'(\alpha)<0$ for $\alpha>1-\Phi(\mu/2)$.
Hence, $g$ increases up to $\alpha^\star$ and decreases thereafter,
and the maximizer is unique.

It follows that the unconstrained maximizer is
\[
\alpha^\star_{\mathrm{free}}
=
1-\Phi\!\left(\frac{\mu}{2}\right).
\]
Imposing the constraint $\alpha\leq\frac{1}{m-1}$ yields
\[
\alpha^\star
=
\min\!\left\{
\frac{1}{m-1},
\;
1-\Phi\!\left(\frac{\mu}{2}\right)
\right\}.
\]
Consequently, 
\[
0\text{-}\mathrm{RAD}\leq \frac{m-1}{m}\left(1-\Phi\left(\Phi^{-1}(1-\alpha^{\star})-\mu\right)-\alpha^{\star}\right).
\]
\end{example}
We plot this bound for DP-SGD in~\Cref{fig:dp-sgd_mnist,fig:dp-sgd_fashion}. While it is not perfectly tight, it provides a reliable approximation, avoiding the numerical computations required by the exact bound of~\Cref{th:optimal_bound}.

Moreover, as a consequence of the previous result, we can obtain a bound of the RAD of any $(\varepsilon,\delta)$-DP mechanism: 
\begin{proposition}
\label{prop:dp_implies_urero}
    If a mechanism $\mathcal{M}\colon\Z^{n}\to\D(\Theta)$ satisfies $(\varepsilon,\delta)$-DP, then  for any attack $A\colon \Theta\to\D(\Z)$, it satisfies
    \[
    \eta\text{-}\mathrm{RAD} \leq \min\{ \kappa_{\pi, \eta}^+ (\e^\varepsilon -1)+\delta,\tfrac{(1-\kappa^{-}_{\pi,\eta})(\e^{\varepsilon}-1)+\delta}{\e^{\varepsilon}},\tfrac{e^{\varepsilon}-1+2\delta}{e^{\varepsilon}+1}(1-\kappa_{\pi})\}.
    \]     
\end{proposition}
\begin{proof}
    Follows from combining previous theorem with~\cite{dong2019Gaussiana} result that any $(\varepsilon,\delta)$-DP mechanism is $f$-DP with, $f(\alpha)=\max\{1-\delta-\e^{\varepsilon}\alpha, \frac{1-\delta-\alpha}{\e^{\varepsilon}}\}$, and analyze the different cases until we arrive to the bound. Formally, every $(\varepsilon,\delta)$-DP mechanism verifies the that $f$-DP, with $f$
\begin{gather}
    f(\alpha)=\max\!\bigg\{\underbrace{1-\delta-\e^{\varepsilon}\alpha}_{f_1(\alpha)}, \underbrace{\frac{1-\delta-\alpha}{\e^{\varepsilon}}}_{f_2(\alpha)}\bigg\}.
\end{gather}
On the other side, applying~\Cref{th:f-DP} we have
\begin{gather}
\eta\text{-}\mathrm{RAD}\leq\max_{\alpha\in[\kappa^{-},\kappa^{+}]}\left( 1-f(\alpha)-\alpha\right).
\end{gather}
Combining both equations we obtain,
\begin{gather*}
    \eta\text{-RAD}\leq \max_{\alpha\in[\kappa^{-},\kappa^{+}]}\left( 1-f(\alpha)-\alpha\right)\\
    =\max_{\alpha\in[\kappa^{-},\kappa^{+}]} 1-\max\{f_1(\alpha),f_{2}(\alpha)\}-\alpha\\
    =\max_{\alpha\in[\kappa^{-},\kappa^{+}]}(1-\max\{f_1(\alpha)+\alpha,f_{2}(\alpha)+\alpha\})\\
    =\max_{\alpha\in[\kappa^{-},\kappa^{+}]}(\min\{1-f_1(\alpha)-\alpha,1-f_{2}(\alpha)-\alpha\})\\
    \leq \min\!\bigg\{\max_{\alpha\in[\kappa^{-},\kappa^{+}]}1-f_1(\alpha)-\alpha,\max_{\alpha\in[\kappa^{-},\kappa^{+}]}1-f_2(\alpha)-\alpha\bigg\}
\end{gather*}
Therefore, we analyze both maximums.

First, for $f_1$ we have: 
\begin{gather}
    1-f_1(\alpha)-\alpha=\delta+\e^{\varepsilon}\alpha-\alpha\\
    =\alpha(\e^{\varepsilon}-1)+\delta\leq \kappa^{+}(\e^{\varepsilon}-1)+\delta
\end{gather}

Second, for $f_2$ we obtain:
\begin{gather}
    1-f_2(\alpha)-\alpha =1-\frac{1-\delta-\alpha}{\e^{\varepsilon}}-\alpha\\
    =1-\frac{1-\delta}{\e^{\varepsilon}}+\alpha(\e^{-\varepsilon}-1)\leq 1-\kappa^{-}(1-\e^{-\varepsilon})-\frac{1-\delta}{\e^{\varepsilon}}=\frac{(1-\kappa^{-})(\e^{\varepsilon}-1)+\delta}{\e^{\varepsilon}}.
\end{gather}
Combined with the general bound~\Cref{th:dp_implies_aux-urero} it follows the result.
\end{proof}
This bound enables to better interpret DP parameters in terms of reconstruction attacks without auxiliary knowledge. 

Moreover, this bound remains informative even when the mechanism in use is completely unknown. For instance, consider auditing external software from a company that claims to provide $(\varepsilon,\delta)$-DP but does not disclose the mechanism used. In such a black-box setting, where we are allowed to query the model but never know the underlying mechanism, our~\Cref{prop:dp_implies_urero} still applies.

However, as discussed in~\Cref{sec:main}, the actual privacy protection of a DP mechanism depends on its specific design and cannot be characterized solely by its privacy parameters. Consequently, this upper-bound should be used as a last-option estimate when the mechanism is unknown, rather than as a substitute for proper noise calibration in mechanism design.

Next, we focus on improving this black-box bound for perfect reconstruction, i.e.,  $\eta=0$, in categorical data. This case is particularly relevant since many sensitive attributes, such as diseases, political opinions, or religious beliefs, are categorical and do not trivially support partial reconstruction, e.g.~\cite{Fredrikson2015Model, Fredrikson2014Privacy}. For such settings, we derive more precise bounds. To do so, we first introduce the following auxiliary lemma:
 \begin{lemma}\label{lemma:Gamma}
     Given $|\Z|=m$ and $\mathcal{M}\colon \Z^n\to\D(\Theta)$ an $(\varepsilon,\delta)$-DP mechanism, for any attack $A\colon\Theta\to\D(\Z)$ and $\gamma_x=\Pr_{\M}(\Theta_z\mid z)-\Pr_{\M}(\Theta_z)$, with $\Theta_z$ as in \Cref{eq:no_auz_uni}, then
\begin{gather}
        \Gamma\coloneqq \sum_{z\in\Z} \gamma_x 
        \leq \frac{(m-1)(\e^{\varepsilon}-1+\delta m)}{\e^{\varepsilon}+m-1}.
    \end{gather}
 \end{lemma}
 \begin{proof}
By definition $\Theta_z\cap\Theta_{z'}=\varnothing$. Besides, for all $\theta$ it exits at least one $z_{\theta}\in\argmax_{z}p_{\M}(\theta\mid z)\pi_z$, and $\bigcup_z \Theta_z=\Theta$. Hence, $\{\Theta_{z}\}_{z\in\Z}$ determines a partition in~$\Theta$. Therefore, by the law of total probability, for each $Z_0$ we have
\begin{gather}\label{eq:total_prop}
   \sum_{z\in\Z}\Pr_{\M}(\Theta_z\mid Z_0)=\sum_{z} \int_{\Theta_z} p_{\M}(\theta\mid Z_0)\diff\mu(\theta) = \int_{\Theta} p_{\M}(\theta\mid Z_0)\diff\mu(\theta) =1.
\end{gather}

On the other hand, since $\M$ is $(\varepsilon,\delta)$-DP, for every $Z_1,Z_0\in \Z$,
\begin{gather}\label{eq:DP_def}
     \Pr_{\M}(\Theta_1\mid Z_0)\geq \e^{-\varepsilon}(\Pr_{\M}(\Theta_1\mid Z_1)-\delta).
\end{gather} 
Substituting~\Cref{eq:DP_def} in~\Cref{eq:total_prop} we obtain, for all $i,j\in[m]$,
\begin{gather}
     \Pr_{\M}(\Theta_i\mid z_i) +\e^{-\varepsilon} \sum_{i\neq j}\Pr_{\M}(\Theta_j\mid z_j)\leq 1+\delta\e^{-\varepsilon}(m-1)
\end{gather}

Summing the above inequality over all $i\in[m]$,
\begin{gather}
    \sum_{i=1}^{m}\Pr_{\M}(\Theta_i\mid z_i)+(m-1)\e^{-\varepsilon}\sum_{i=1}^{m}\Pr_{\M}(\Theta_i\mid z_i)\leq m(1+\delta\e^{-\varepsilon}(m-1))\Leftrightarrow\\
    \sum_{i=1}^{m}\Pr_{\M}(\Theta_i\mid z_i)\leq \frac{m(1+\delta\e^{-\varepsilon}(m-1))}{1+(m-1)\e^{-\varepsilon}}=\frac{m\e^{\varepsilon}+\delta m(m-1)}{\e^{\varepsilon}+(m-1)}.
\end{gather}

Hence, 
    \begin{align}
        \Gamma &= \sum_{z\in\Z} \gamma_x\\
        &=\sum_{z\in\Z}\left(\Pr_{\M}(\Theta_z\mid z)-\Pr_{\M}(\Theta_z)\right) \\
        &=\sum_{z\in\Z}\Pr_{\M}(\Theta_z\mid z)-1\\
        &\leq \frac{m\e^{\varepsilon}+\delta m(m-1)}{\e^{\varepsilon}+m-1}-1 \\
        &=\frac{(m-1)(\e^{\varepsilon}-1+\delta m)}{\e^{\varepsilon}+m-1}. \qedhere
    \end{align}
 \end{proof}
Applying this lemma we obtain the following RAD bound:
 \begin{theorem}[$0$-RAD under $(\varepsilon,\delta)$-DP]\label{th:perfect_reco_bb}
Given $|\Z|=m$ with prior $\pi_1(1-\pi_1)\geq\dots\geq \dots\geq \pi_m(1-\pi_m)$ and $\mathcal{M}\colon \Z^n\to\D(\Theta)$ an $(\varepsilon,\delta)$-DP mechanism, for any attack $A\colon\Theta\to\D(\Z)$ 
\[
0\text{-}\mathrm{RAD}\leq  \frac{e^{\varepsilon}-1+2\delta}{e^{\varepsilon}+1} K_{\pi} + R \max_{i> K}\pi_{i}
\]
where $K\in[m]$ is the largest index satisfying 
  $R=(m-1)  \frac{e^{\varepsilon} - 1+m\delta}{e^{\varepsilon} + m - 1}-(K-\sum_{i=1}^K\pi_i)\frac{e^{\varepsilon}-1+2\delta}{e^{\varepsilon}+1}\geq 0
  $ and $K_{\pi}=\sum_{i}^{K}(1-\pi_i)\pi_i$.
\end{theorem}
\begin{proof}
Since $|\Z|=m$ and $aux=\{\varnothing\}$, \Cref{th:optimal_bound} gets reduced to \Cref{eq:no_auz_uni}, hence
\begin{gather}
    0\text{-RAD}\leq \sum_{i=1}^{m}\left(\Pr_{\M}(\Theta_i\mid z_i)-\Pr_{\M}(\Theta_i)\right)\pi_i\equiv\sum_{i=1}^{m} \gamma_i\,\pi_i.
\end{gather}

For one side, we obtain that for all $i\in[m]$,
\begin{gather}
    \gamma_i=\Pr_{\M}(\Theta_i\mid z_i)-\Pr_{\M}(\Theta_i)=\int_{\Theta_i} p_{\M}(\theta\mid z_i)-\sum_{j\in[m]} p_{\M}(\theta\mid z_j)\pi_j \diff\mu(\theta)\\
    =\int_{\Theta_i} \sum_{j\in[m]}\left( p_{\M}(\theta\mid z_i)- p_{\M}(\theta\mid z_j)\right)\pi_j \diff\mu(\theta)\\
    =\sum_{j\neq i} \left(\Pr_{\M}(\Theta_i\mid z_i)-\Pr_{\M}(\Theta_i\mid z_j)\right)\pi_j\\
    \leq \mathrm{TV}(\M)\sum_{j\neq i}\pi_j\leq\frac{e^{\varepsilon}-1+2\delta}{e^{\varepsilon}+1}(1-\pi_i).
\end{gather}

If we simply apply this bound we recover~\Cref{th:dp_implies_aux-urero} result:
\[
0\text{-RAD}\leq\sum_{i=1}^{m} \gamma_i\,\pi_i\leq  \sum_{i=1}^{m}\frac{e^{\varepsilon}-1+2\delta}{e^{\varepsilon}+1}(1-\pi_i)\pi_i=\frac{e^{\varepsilon}-1+2\delta}{e^{\varepsilon}+1}(1-\kappa_{\pi}).
\]

However, due to~\Cref{lemma:Gamma}, we know that this bound is loose, since in this case,
\begin{gather}
    \Gamma=\sum_{i=1}^{m} \gamma_i=\frac{e^{\varepsilon}-1+2\delta}{e^{\varepsilon}+1} (m-1) \geq \frac{e^{\varepsilon} - 1+m\delta}{e^{\varepsilon} + m - 1} (m-1)=\Gamma_{\max},
\end{gather}
contradicting~\Cref{lemma:Gamma}; therefore, it is impossible to achieve the local inequality $\gamma_i\leq \mathrm{TV}(\M)(1-\pi_i)$ simultaneously for all $i\in[m]$. In most cases, we can apply the local bound to a reduced set of indexes $k$, and the remainders must adjust so that the total sum $\sum_i \gamma_i=\Gamma$. Formally, at most, we can sum $k$ summands such that,
\begin{gather}
    \sum_{r=1}^{k} \gamma_{i_r}\leq \frac{e^{\varepsilon} - 1+m\delta}{e^{\varepsilon} + m - 1} (m-1)\Leftrightarrow\\
    \sum_{r=1}^{k}(1-\pi_{i_r})\leq (m-1)\frac{(\e^{\varepsilon}-1+m\delta)((\e^{\varepsilon}+1)}{(\e^{\varepsilon}-1+2\delta)((\e^{\varepsilon}-1+2\delta))}
\end{gather}

Hence, 
without loss of generality we order the indices so that
\[
\pi_1(1-\pi_1)\ge \pi_2(1-\pi_2)\ge\cdots\ge\pi_m(1-\pi_m).
\] obtaining,
\begin{gather}
    0\text{-RAD}\leq \frac{e^{\varepsilon}-1+2\delta}{e^{\varepsilon}+1} \sum_{i=1}^{k_{\pi}} \pi_i(1-\pi_i)+R\max_{r>k_{\pi}}\pi_{r}
\end{gather}
with $k_{\pi}$ the maximum index verifying:
\[
\sum_{i=1}^{k_{\pi}}(1-\pi_i)\leq (m-1)\frac{(\e^{\varepsilon}-1+m\delta)((\e^{\varepsilon}+1)}{(\e^{\varepsilon}-1+2\delta)((\e^{\varepsilon}-1+2\delta))},
\]
and $R$ the reminder, i.e, 
$K$ the biggest index such that 
  \[
  R=(m-1)  \frac{e^{\varepsilon} - 1+m\delta}{e^{\varepsilon} + m - 1}-(K-\sum_{i=1}^K\pi_i)\frac{e^{\varepsilon}-1+2\delta}{e^{\varepsilon}+1}\geq 0.\qedhere
  \]
\end{proof}
Note that in the extreme case where $\pi_1 = \pi_2 = \frac{1}{2}$ and $\pi_i = 0$ for all $i \neq 1,2$, we recover exactly the same result as in~\Cref{th:dp_implies_aux-urero}.

Importantly,~\Cref{th:perfect_reco_bb} is less applicable than~\Cref{prop:dp_implies_urero}, since it only applies for perfect reconstruction, $\eta=0$, in categorical data. However, under these assumptions it offers a more accurate bound as we see in~\Cref{fig:black_box bounds}. 

\begin{figure}
    \centering
    \includegraphics[width=0.6\linewidth]{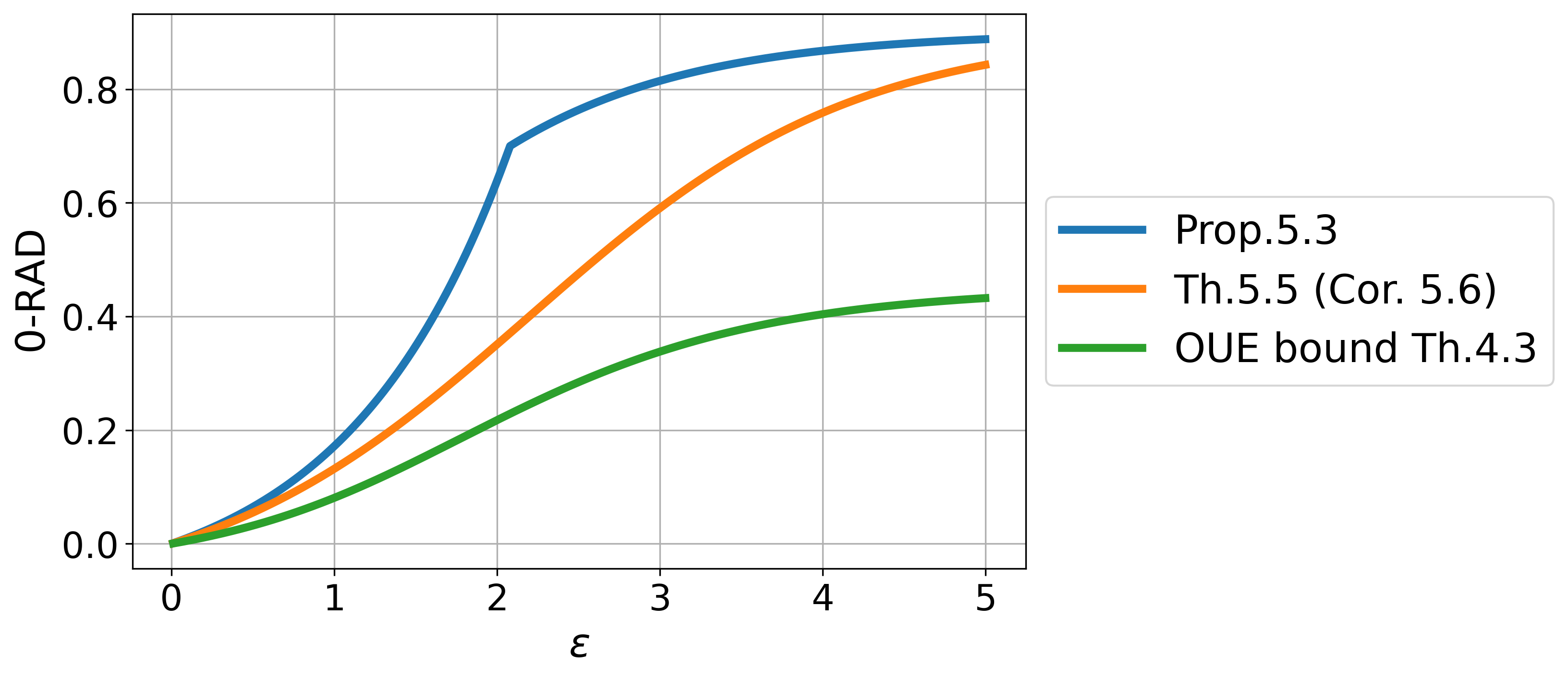}
    \caption[RAD Black-Box Estimation]{Comparison of black-box bounds for $0$-RAD without auxiliary knowledge, $\pi = U[10]$ and $\delta = 10^{-5}$. The bound given in \Cref{prop:dp_implies_urero} is more general and applies in any setting. In contrast, \Cref{th:perfect_reco_bb} is specific to categorical data but provides a tighter risk estimate when applicable. Finally, if the mechanism is known---here, OUE---it is always preferable to use the tighter bound provided by \Cref{th:optimal_bound}.}
    \label{fig:black_box bounds}
\end{figure}

This formulation enables the assessment of intermediate configurations of $\pi$. Notably, when $\pi=U[m]$ yields a marked improvement:
\begin{corollary}[Black-box Uniform Prior]\label{co:unif}
Let $\pi=\mathrm{U}[m]$ the uniform distribution over $\Z$. If a mechanism $\mathcal{M}$ satisfies $(\varepsilon, \delta)$-DP, for any attack $A\colon \Theta\to\D(\Z)$ it guarantees
\[
0\text{-}\mathrm{RAD}\leq \frac{e^{\varepsilon}-1+\delta m}{e^{\varepsilon}+m-1}\frac{m-1}{m}.
\]
\end{corollary}
\begin{proof}
    For every $K\in[m]$, $K_{\pi}=\sum_{i=1}^{K}(1-\pi_i)\pi_i=K\frac{m-1}{m^2}$ and $(K-\sum_{i=1}^{K}\pi_i)=K\frac{m-1}{m}$, therefore, denoting $A=\frac{\e^{\varepsilon}-1+2\delta}{\e^{\varepsilon}+1}$ and applying~\Cref{th:perfect_reco_bb} we get:
    \begin{gather}
        0\text{-RAD}\leq A K\frac{m-1}{m^2}+\frac{1}{m}(\Gamma-K\frac{m-1}{m}A)=\frac{1}{m}\Gamma= \frac{e^{\varepsilon}-1+\delta m}{e^{\varepsilon}+m-1}\frac{m-1}{m}.
    \end{gather}
\end{proof}

\textbf{Remark on composition.}
Since our $\eta$-RAD bounds depend explicitly on the privacy parameters---namely $\varepsilon$, $\delta$, and/or $f$---they can be directly recomputed under composition by first applying the corresponding composition results to obtain the composed privacy parameters (see \Cref{sec:background}), and then evaluating the bounds on these composed values.  In the following example, we illustrate how to derive RAD composition bounds for the particular case of DP-SGD.
\begin{example}\label{ex:dp-sgd_composition}
Given a risk threshold, $\mathrm{RAD}\leq\gamma$, we aim to calibrate the noise scale $\sigma$ (i.e., the standard deviation of the Gaussian noise added to the gradients during training~\cite{Abadi2016Deep}) on a full-batch DP-SGD, for $T$ steps to protect against the threat model considered by \citeauthor{Hayes2023Bounding}\cite{Hayes2023Bounding}, i.e., white-box access to private gradients, uniform prior over $|\Z|=m$ and $\eta=0$, hence $\kappa_{-}=\kappa_{+}=1/m$. 

Each iteration of a full-batch DP-SGD performs a Gaussian mechanism on the gradient computation, hence, we discussed in~\Cref{sec:background}, it verifies $\mu$-GDP~\cite{dong2019Gaussiana}, with $\mu=1/\sigma$. 

The adaptive composition of $T$ iterations of a $\mu$-GDP mechanism is $(\mu\sqrt{T})$-GDP, as discussed in~\Cref{sec:background}. Hence, a complete training of DP-SGD with $T$ iterations, is $(\sqrt{T}\sigma^{-1})$-GDP. Moreover, any $\mu$-GDP mechanism has total variation $\mathrm{TV}\leq 2\Phi(\frac{\mu}{2})-1$~\cite{ghazi2024total}, hence DP-SGD after $T$ iterations satisfies
\[
\gamma\leq\frac{m-1}{m}(2\Phi(\frac{\sqrt{T}}{2\sigma})-1)
.\]
Combining this composition result with our theorems we obtain direct calibration rules:

Without information about $aux$, we use~\Cref{th:dp_implies_aux-urero}. Obtaining,
\[
\eta\text{-}\mathrm{RAD}\leq\mathrm{TV}(\M)\left(1-\frac{1}{m}\right)=\frac{m-1}{m}\left(2\Phi(\frac{\sqrt{T}}{2\sigma})-1\right).
\]
We plot this bound for $T=100$ in~\Cref{fig:dp-sgd_mia_mnist,fig:dp-sgd_aia_mnist}. We can then solve $\sigma$ for any desired risk $\gamma$.

If we consider the whole records sensitive, $aux=\{\varnothing\}$, then we apply~\Cref{th:f-DP}:
\begin{align*}0\text{-}\mathrm{RAD}&\leq \frac{m-1}{m}\max\limits_{\alpha\in[0,\frac{1}{m-1}]}\left(1-\Phi\left(\Phi^{-1}(1-\alpha)-\frac{\sqrt{T}}{\sigma}\right)-\alpha\right)
\end{align*}

Hence, given $\alpha^*
=
\min\!\left\{
\frac{1}{m-1},
\;
1-\Phi\!\left(\frac{\sqrt{T}}{2\sigma}\right)
\right\}$ (see~\Cref{ex:maz_alpha}),
 the minimum $\sigma$ to guarantee $0\text{-RAD}\leq\gamma$ is:
\[
\sigma
\;\ge\;
\frac{\sqrt{T}}{
\Phi^{-1}\!\left(1-\alpha^{*}\right)
-
\Phi^{-1}\!\left(1-\frac{m}{m-1}\,\gamma-\alpha^{*}\right)
}.
\]

We plot this bound for the case of $T=100$ in~\Cref{fig:dp-sgd_dra_mnist}. A practitioner can then choose the minimum noise scale $\sigma$ for any given risk threshold $\gamma$. For instance, given that a set of $m=10$ individuals do not tolerate a risk bigger than $0.1$, for a training of $T=100$ iterations, one must add noise calibrated to $\sigma=22$.
\end{example}

In summary, this section provides reasonable closed-form upper bounds (as we show in~\Cref{sec:results}) for estimating RAD when \Cref{th:optimal_bound} cannot be computed explicitly or $\mathcal{M}$ is unknown and $aux = \{\varnothing\}$, hence \Cref{th:dp_implies_aux-urero} would  overestimate the risk. Importantly, these bounds offer composition results as we summarize in \Cref{tab:bounds_comparison}.

\begin{table}[t]
    \centering
    \renewcommand{\arraystretch}{1.8} % Mayor espacio vertical entre filas
    \setlength{\tabcolsep}{3pt}      % Mayor espacio horizontal entre columnas
    \resizebox{!}{0.2\textwidth}{
    \begin{tabular}{@{}c lcccc@{}}
    \toprule
    & \textbf{Notion} & \textbf{Assumptions} & \textbf{RAD bound} & \textbf{Composition} & \textbf{ReRo bound}\\
    \midrule
    & Total variation 
    & --- 
    & \Cref{th:dp_implies_aux-urero} 
    & $\checkmark$ & $\nexists$\\

    & $f$-DP 
    & $aux=\{\varnothing\}$ 
    & \Cref{th:f-DP}
    & $\checkmark$ 
    & \Cref{eq:hayes_bound}~\cite{Hayes2023Bounding}\\

    & $\M$ 
    & $aux$ known 
    & \Cref{th:optimal_bound}
    & $\times$ 
    & $\nexists$\\
    \midrule
    & $(\varepsilon,\delta)$-DP 
    & ---  
    &\Cref{th:dp_implies_aux-urero}
    & $\checkmark$ 
    & $\nexists$\\
    
    & $(\varepsilon,\delta)$-DP 
    & $aux=\{\varnothing\}$ 
    & \Cref{prop:dp_implies_urero}
    & $\checkmark$   
    &  \Cref{eq:balle_bound}~\cite{Balle2022Reconstructing}\\
    
    & $(\varepsilon,\delta)$-DP 
    & $aux=\{\varnothing\}, \eta=0$ 
    & \Cref{th:perfect_reco_bb} 
    & $\checkmark$ 
    &  \Cref{eq:balle_bound}~\cite{Balle2022Reconstructing}\\
    \bottomrule
\end{tabular}
}
          \caption[Summary of RAD Bounds Applicability]{Summary of  RAD bounds applicability. }
    \label{tab:bounds_comparison} 
\end{table}

\section{RAD for DP Auditing}\label{sec:auditing}
DP auditing is a crucial tool for assessing the tightness of DP mechanisms, establishing the practical impact of the mechanism parameters, and detecting implementation flaws in deployed DP mechanisms~\cite{Annamalai2024Nearly,Jagielski2020Auditing, Bichsel2021DPSniper}. While previous DP auditing tools focus on solving specifically one of the aforementioned aspects, we propose a general-purpose DP auditing framework: RAD-based DP auditing. 

RAD provides a unifying framework for analyzing adversarial risk under arbitrary threat models. Moreover, our bounds establish a tight and explicit connection between RAD and the standard DP privacy parameters. Taken together, these results yield a simple and principled approach to general-purpose DP auditing. Precision and tightness are especially critical in this context, since loose estimates may underestimate privacy risks or fail to detect bugs and implementation flaws.

The core idea of RAD-based auditing is straightforward: given a measured RAD value $\tilde{\gamma}$, we invert our theoretical bounds to estimate an empirical privacy budget.
This empirical $\Tilde{\varepsilon}$ reflects the observed privacy loss in practice, complementing theoretical worst-case values and providing a more realistic perspective on real-world risk. Formally, in previous sections, we provide bounding functions $B$ such that $\mathrm{RAD
}(\M)\leq B(\varepsilon,\delta)$ for any $(\varepsilon,\delta)$-DP mechanism. Given a bound $\eta$-RAD $\leq B(\varepsilon,\delta)$, we compute RAD empirically obtaining $\gamma$, and estimate $\Tilde{\varepsilon}\geq B^{-1}(\gamma,\delta)$. 

The bound we employ depends on the specific setting. For instance, in a completely black-box scenario---where not even the mechanism used is known---for categorical data, in which we assume $\pi=U[m]$, the best bound is~\Cref{co:unif}. 
Therefore, the DP auditing framework consists of running an attack, measuring its empirical RAD $\widetilde{\gamma}$, and deriving  $\widetilde{\varepsilon}$ as follows: 
 \begin{gather}\label{eq:ldp_audit}
    \widetilde{\varepsilon} = 
    \begin{cases}
    \ln \left( \frac{\widetilde{\gamma}\,m+1}{1 - \widetilde{\gamma} \frac{m}{m-1}}\right) & \text{if the term can be evaluated,} \\[1em]
    \text{undefined} & \text{otherwise.}
    \end{cases}
    \end{gather}

However, if the mechanism $\M$ is known, we can use our improved bound from~\Cref{th:optimal_bound} (see~\Cref{ex:grr_short,ex:oue_short,ex:ss_short}).

Our auditing framework overcomes the fundamental scalability limitations of prior learning-based approaches such as DP-Sniper and Eureka~\cite{Bichsel2021DPSniper,Lu2024Eureka}, enabling auditing in high-dimensional categorical LDP settings. Unlike these methods, our approach avoids costly hyperparameter tuning and the search for worst-case neighboring databases, and remains computationally feasible even when the input domain contains thousands of categories (see~\Cref{sec:experiments}).

Despite the importance of LDP mechanisms~\cite{Erlingsson2014RAPPOR,Lu2024Eureka}, only one major work has so far focused on LDP auditing: \textsc{\textsc{LDP Auditor}}~\cite{Arcolezi2024Revealing}. Applying our RAD-based DP auditing to LDP, we address key limitations of prior work. 
In contrast to \textsc{\textsc{LDP Auditor}}, which focuses exclusively on perfect reconstruction without target-specific auxiliary knowledge---excluding important use-cases such as AIAs---we allow auditing under broader threat models by leveraging optimal attacks (see Algorithm~\ref{alg:OptimalAIA1}). 
Moreover, \textsc{LDP Auditor} uses the Clopper–Pearson method to compute confidence intervals for the attacker’s success probability. Since the upper bound of the interval must conservatively cover the true probability with high confidence, it systematically produces estimates that are higher than the actual value~\cite{Arcolezi2024Revealing}. This intrinsic limitation is avoided in our approach, which does not rely on confidence intervals.

We investigate and empirically show the improvement in accuracy of our auditing approach in~\Cref{sec:experiments} (cf.~\Cref{fig:audit_porto,fig:audit_beijing} for results), where we audit three main LDP mechanisms---GRR, SS and OUE---showing improved accuracy for all of them.

\section{Experiments}\label{sec:experiments}
In this section, we empirically examine the limitations of ReRo described in~\Cref{sec:related_work}, focusing on how existing bounds fail to account for realistic attackers with target-specific auxiliary information. Moreover, we validate our theoretical bounds and our RAD-based DP auditing framework in real-world databases and DP mechanisms. Our experiments show that RAD accurately distinguishes privacy leakage from imputation, with tight bounds in practice, making it a reliable tool for interpretable noise calibration. RAD also enables auditing of LDP mechanisms, improving both scope and accuracy over the state-of-the-art~\cite{Arcolezi2024Revealing}.

To ensure a fair comparison between ReRo and RAD in risk assessment, we emulate their experimental designs and dataset choices whenever possible. These design choices are particularly suitable for evaluating the tightness of our bounds, as they were originally used to assess the tightness of the ReRo bounds and yielded nearly tight results~\cite{Hayes2023Bounding}. This suggests that these settings already serve as strong testbeds for tightness evaluation.
Similarly, when comparing our RAD-based auditing framework to \textsc{LDP Auditor}, we adhere to their experimental design choices to ensure a coherent and consistent evaluation. We provide detailed descriptions of the datasets and experimental parameters in the following sections.

\subsection{Database Description}
We evaluate private learning, aggregation and LDP scenarios, using tailored datasets for each setting. The database selection is guided by their relevance in prior work and availability.

For DP-SGD, we use the same dataset as in ReRo~\cite{Hayes2023Bounding} for consistency: MNIST~\cite{Lecun1998Gradient}, with 70\,000 grayscale images of handwritten digits. We also replicate results on Fashion-MNIST~\cite{Xiao2017Fashion} (Fashion), which similarly contains 70\,000 grayscale images of clothes.

To evaluate the imputation attack~\cite{Jayaraman2022Are}, we use the Census and Texas-100X datasets in consistency with the original paper. 
The Census dataset~\cite{Jayaraman2022Are} contains 1\,676, records with 14 attributes, where race is treated as the sensitive attribute with eight categories. The Texas-100X dataset~\cite{Jayaraman2022Are} comprises 925\,128 patient records from 441 hospitals, including demographic and medical attributes, with a binary ethnicity attribute designated to be sensitive.

We evaluate aggregation in the Adult dataset~\cite{adult_2}, a census dataset commonly used in privacy-preserving aggregation~\cite{soria2014enhancing}. It consists of
$32\,561$, records with two numerical attributes, from which we select (working) hours-per-week following previous work~\cite{soria2014enhancing}, leading to the domain $\Z =\{0, \dots, 100\}$.

Finally, we evaluate our LDP auditing framework on location-reconstruction attacks using two real-world mobility datasets: the Porto dataset~\cite{OConnell2015Taxi} and the Geolife dataset~\cite{Zheng2011Geolife}. Both datasets are widely used in privacy and mobility research (e.g.,~\cite{Pyrgelis2017What,Lestyan2022In, Xiao2015Protecting}) and are publicly available. Each dataset consists of GPS coordinates, which we map to the OpenStreetMap (OSM) graph format~\cite{OpenStreetMap} like prior work. 
The Porto dataset contains a total of $83,409,386$ location reports that we map to the OSM roadgraph at Porto’s city center (41.1475\textdegree~N, 8.5870\textdegree~W) with a 2.7\,km radius, capturing the urban core of Porto. This radius leads to a universe size $|\Z|=3\,052$. The Geolife dataset contains a total of $24\,876\,978$ locations that we mapped to an OSM graph centered near Tiananmen Square (39.9130\textdegree~N, 116.3703\textdegree~E) with a 5\,km radius covering major central districts, leading to a universe of size $|\Z|=5\,356$.

\begin{figure*}[t] 
  \centering
  \begin{subfigure}[b]{0.33\textwidth}
    \centering
    \includegraphics[width=0.9\linewidth]{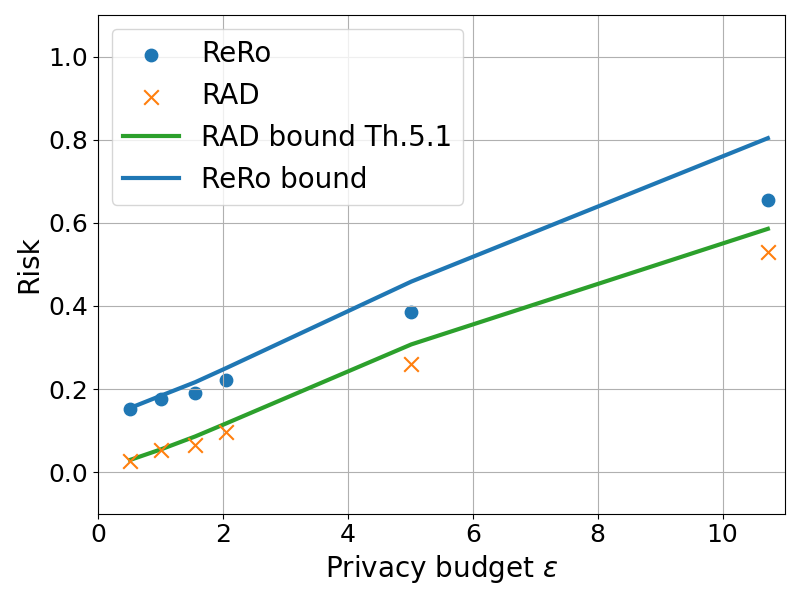}
    \caption{DRA, $aux=\{\varnothing\}$.}
    \label{fig:dp-sgd_dra_mnist}
  \end{subfigure}\hfill
  \begin{subfigure}[b]{0.33\textwidth}
    \centering
    \includegraphics[width=0.9\linewidth]{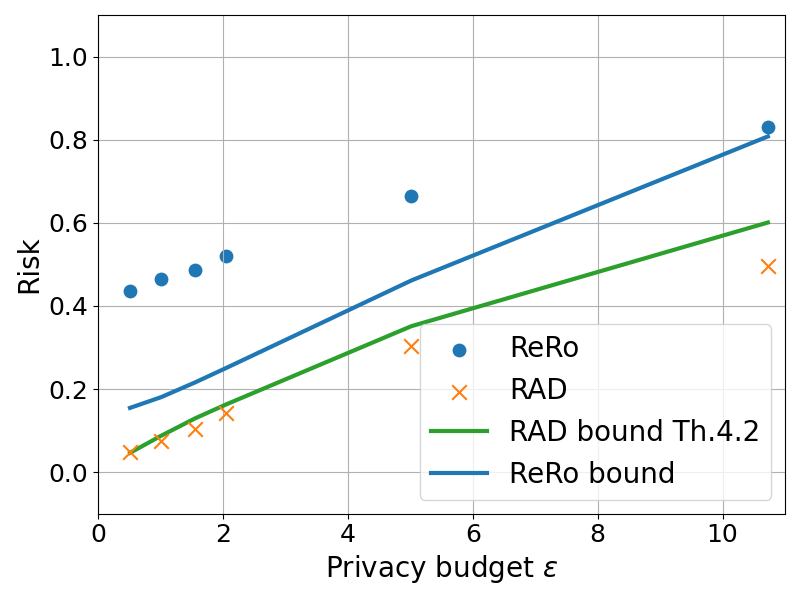}
    \caption{DRA, $a(z)=\text{image label}$.}
    \label{fig:dp-sgd_aia_mnist}
  \end{subfigure}\hfill
  \begin{subfigure}[b]{0.33\textwidth}
    \centering
    \includegraphics[width=0.9\linewidth]{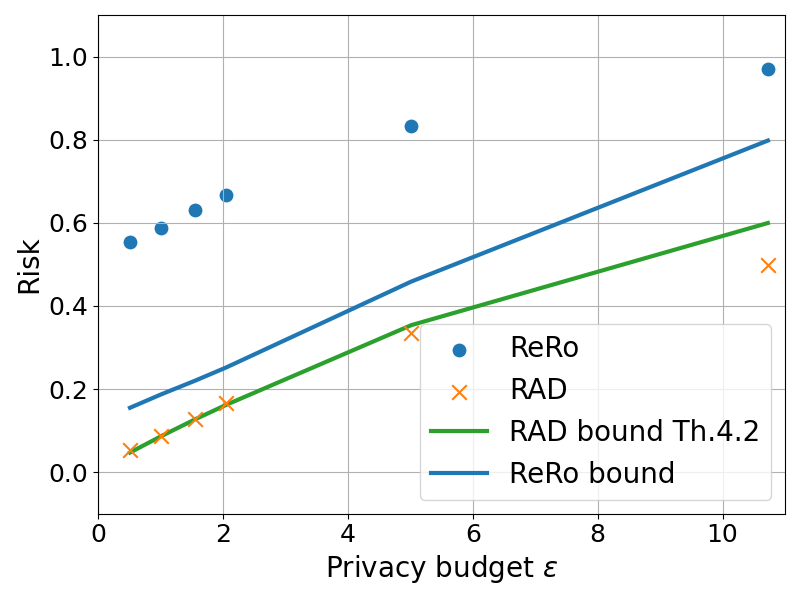}
    \caption{MIA, $a(z)=x$.}
    \label{fig:dp-sgd_mia_mnist}
  \end{subfigure}
  \caption[RAD vs.\ ReRo Results for Optimal Attacks against DP-SGD on MNIST]{RAD vs.\ ReRo results for optimal attacks against DP-SGD on MNIST. Lines show theoretical bounds and markers of empirical risk as estimated by RAD/ReRo. Empirical results exceed the bounds estimated by ReRo, whereas our RAD bounds remain close to the true risk. Moreover, while ReRo sharply increases when auxiliary knowledge is available, RAD effectively discounts imputation.}
  \label{fig:dp-sgd_mnist}
  
\end{figure*}

\begin{figure*}[t] 
  \centering
  \begin{subfigure}[b]{0.33\textwidth}
    \centering
    \includegraphics[width=0.9\linewidth]{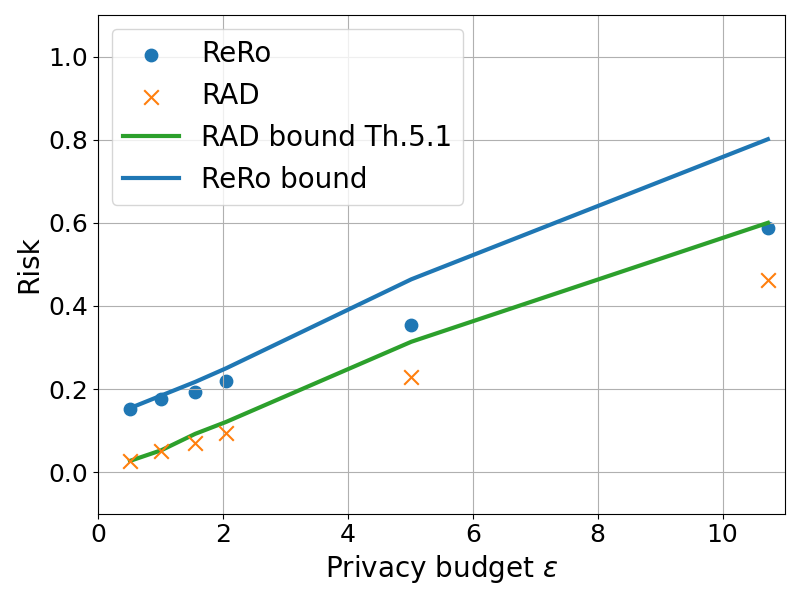}
    \caption{DRA, $aux=\{\varnothing\}$.}
    \label{fig:dp-sgd_dra_fashion}
  \end{subfigure}\hfill
  \begin{subfigure}[b]{0.33\textwidth}
    \centering
    \includegraphics[width=0.9\linewidth]{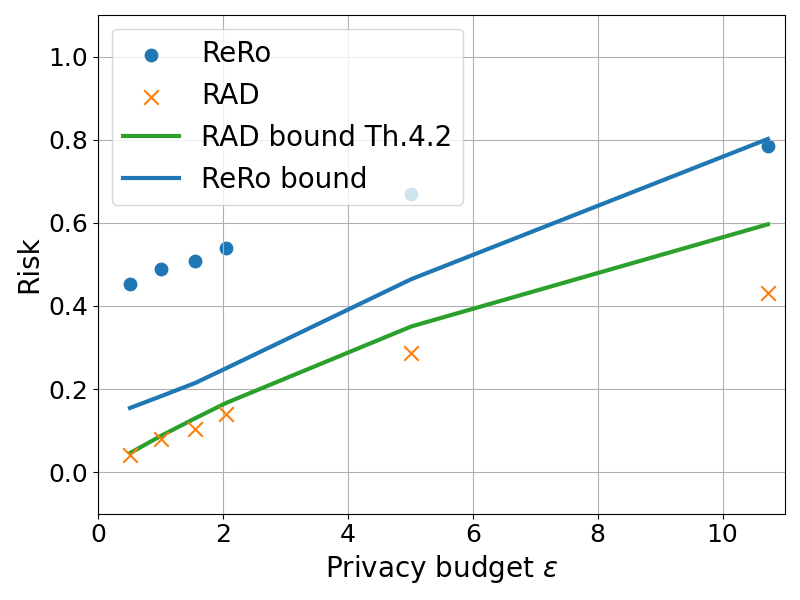}
    \caption{DRA, $a(z)=\text{image label}$.}
    \label{fig:dp-sgd_aia_fashion}
  \end{subfigure}\hfill
  \begin{subfigure}[b]{0.33\textwidth}
    \centering
    \includegraphics[width=0.9\linewidth]{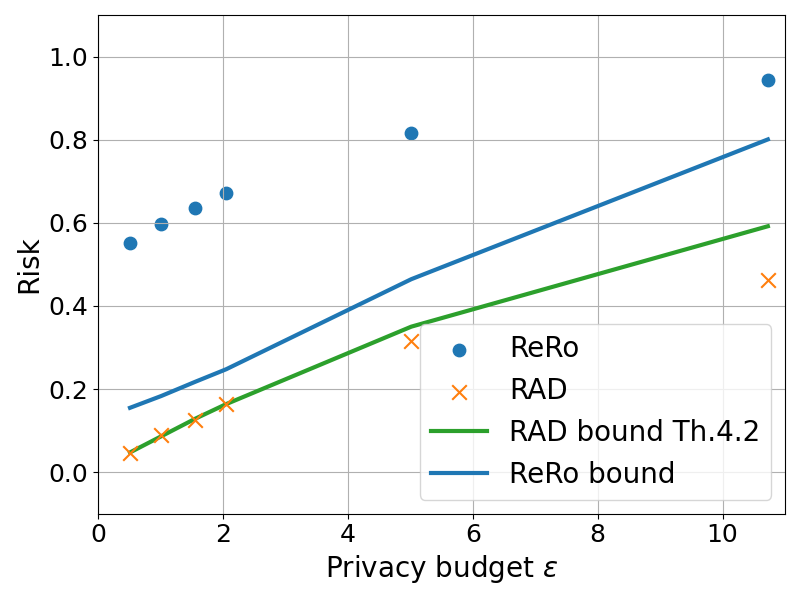}
    \caption{MIA, $a(z)=x$.}
    \label{fig:dp-sgd_mia_fashion}
  \end{subfigure}
  \caption[RAD vs.\ ReRo Results for Optimal Attacks against DP-SGD on Fashion]{RAD vs.\ ReRo results for optimal attacks against DP-SGD on Fashion. Lines show theoretical bounds and markers of empirical risk as estimated by RAD/ReRo. Both ReRo and RAD show a consistent behavior with respect to the MINST dataset.}
  \label{fig:dp-sgd_fashion}
  
\end{figure*}

\subsection{Experiment Design}\label{sec:exp_design_attacks}
We investigate attacks on private learning (DP-SGD), aggregation queries (Laplace mechanism), and LDP protocols (GRR, OUE, SS) under varying auxiliary information settings to validate our bounds, compare RAD and ReRo, and evaluate our auditing framework.

To demonstrate that \textit{ReRo overestimates risk}---and how RAD overcomes this limitation---we consider an attack that completely ignores the output of the private mechanism and relies solely on public information. This allows us to assess how ReRo behaves in a scenario where no private information is disclosed due to participation. To this end, we select the pure imputation attack~\cite{Jayaraman2022Are}. This attack uses a public dataset $D_{-}$ to train a separate attack classifier $A_{I}$ that, given the public attributes of a target, returns as label a prediction for the sensitive one. The adversary is given  only the target public attribute $a(z)$ and outputs the prediction
$
\widetilde{s_z} = \arg\max_{s_i \in \Theta} \Pr_{\mathcal{I}}[s_i \mid a(z)],
$
where the conditional distribution $\Pr[s_i \mid a(z)]$ is estimated by $A_{I}$, once the imputation model has been trained on $D_{-}$. This attack does not use any information from the target model $\M(D)$; therefore, adversarial success cannot be privacy leakage resulting from a user's participation in the training dataset of $\M(D)$. Following the original paper~\cite{Jayaraman2022Are}, we tested in both the Census and Texas datasets. We set $|D_{-}|=49\,000$ and a universe $\Z$ of $m=1\,000$, randomly selected from the remaining data records consistent with~\cite{Jayaraman2022Are}. 
We define the attack to be successful, $\ell(z,z')=0$, if $a(z)=a(z')$, as a classical AIA.

We demonstrate how \textit{RAD improves over ReRo} and establish the optimality of our bounds in both private learning and DP aggregation settings. In both cases, we evaluate tightness by testing our corresponding optimal attacks.
To ensure a fair comparison, we emulate the original ReRo experimental setup for private learning, where the authors report their bounds to be nearly tight for 
$aux=\{\varnothing\}$. For DP aggregation, although no experimental results are reported in the original work, we adhere as closely as possible to the same parameter choices.

For private learning we run the attacks against DP-SDG on the MNIST and Fashion image datasets in three settings: $aux = z$ (a MIA), $aux = \{\varnothing\}$ (a DRA, replicating the setting in~\cite{Hayes2023Bounding}), and $aux = a(z)$ (a DRA, where the adversary also knows the target image's label, i.e., which object is contained). To ensure a fair comparison with ReRo bounds, we select the parameters and thresholds exactly as specified in the original paper~\cite{Hayes2023Bounding}. Namely, we declare an attack successful when $A(\theta,a(z)) = z$, that is, $\eta = 0$. We set $|D_{-}| = 999 $ (and so the training set size is $|D_{-}\cup\{z\}| = 1\, 000$) and train with full-batch DP-SGD for $T = 100$ steps.  We set
 the clipping rate, i.e., the maximum norm we clip the real gradients to while training, $C = 0.1$ and $\delta=10^{-5}$ and adjust the noise scale $\sigma$ (see~\Cref{ex:dp-sgd_composition}) for a given target $\varepsilon$.  We set the  uniform prior with size $|\Z|=8$ (disjoint from $D_{-}$), meaning that  $\kappa_{\pi,0}^{+}=\kappa_{\pi}=0.125$. Hence, we exactly replicate the original ReRo study~\cite{Hayes2023Bounding}  parameters.

For DP aggregation, we evaluate the optimal attack against the Laplace mechanism on sum queries using the ``working-hours'' attribute of Adult, employing truncation as a post-processing operation. Analogously to the private learning experiments, we set $|D| = 999$, $aux=\{\varnothing\}$ but in this case we evaluate the performance for $\eta \in \{0, 40, 80, 100\}$ to assess the impact or the error threshold on the risk estimation.
Moreover, to understand the impact of the prior distribution on risk assessment, we compare three different distributions. As a baseline, we consider a uniform distribution. To simulate a realistic setting, we empirically estimate the distribution 
$\pi$ from the original data, reflecting real-world frequencies (e.g., working 40 hours per week is \emph{a priori} more likely than working 100 hours per week). Finally, we evaluate a fully skewed distribution with $\pi(100)=\pi(0)=0.5$, representing a worst-case scenario in which the attacker’s prior is concentrated on the two records that are easiest to distinguish in the dataset—analogous to the worst-case perspective in the original DP definition.

Finally, \textit{we evaluate our RAD framework in LDP}, and we compare our auditing framework with the state-of-the-art tool \textsc{\textsc{LDP Auditor}}~\cite{Arcolezi2024Revealing} for three relevant LDP mechanisms: GRR, OUE and SS~\cite{Gursoy2022Adversarial, Arcolezi2023On} .The results for \textsc{LDP Auditor} were obtained in collaboration with Héber H.\ Arcolezi, based on the implementation provided in \citeauthor{Arcolezi2024Revealing}'s public GitHub repository \cite{Arcolezi2024repo}. \textsc{LDP Auditor} estimates the empirical privacy budget in $10^6$ runs. 

We evaluate RAD based on our optimal attack (see Alg.~\ref{alg:OptimalAIA1}) under a uniform prior and without auxiliary knowledge, allowing comparison with \textsc{LDP Auditor}. We then test our own LDP auditing framework: based on the obtained RAD value $\gamma$, we evaluate $B^{-1}(\gamma)$ for $B$ following~\Cref{th:optimal_bound} and obtain an estimate of the empirical privacy budget. The precise $B(\varepsilon)$ for GRR, OUE and SS are shown in~\Cref{ex:grr_short,ex:oue_short,ex:ss_short} respectively. Since $B^{-1}$ is not explicit for OUE, we approximate it numerically using the bisection method, which converges in $\mathcal{O}(\log(\tau^{-1}))$
iterations, where $\tau$ denotes the tolerance level~\cite{sauer2018numerical}. We set $\tau=10^{-6}$. 
Consistent with~\cite{Arcolezi2024Revealing}, we repeat the $\varepsilon$ estimation five times and report the mean and standard deviation.

All experiments rely on empirical estimates of ReRo and RAD, i.e., estimates of a probability and a difference of probabilities, respectively. To obtain these estimates, we use Monte Carlo methods, approximating expected values by repeatedly sampling from the random process and computing the average.
Following~\cite{Hayes2023Bounding}, ReRo is estimated by repeating $J$ times the attack $A(\mathcal{M}(D_x), a(z))$ for each $z \in \Z$ and computing the $\pi$-weighted average. The RAD correction term is estimated analogously by evaluating  $J$ times the attacks $A(\mathcal{M}(D_{z_0}), a(z_1))$ for each target–challenger pair $z_1, z_0 \in \Z$ and averaging the results.

For MNIST, Fashion and Adult, we set $J=1\,000$ (as in~\cite{Hayes2023Bounding}). Note that in the LDP cases $D_{-}=\varnothing$, and we  set $J=10^{6}/m$ ensuring
the total number of runs matches those $10^{6}$ repetitions of \textsc{LDP Auditor}. Finally, for the imputation attack, we do not require a target model as it is target model-independent and set $J=1$. We repeat the imputation attack with five different seeds and report the averaged ReRo and RAD scores.

We use Python and TensorFlow~\cite{TensorFlow} to evaluate the attacks. For DP-SGD, we rely on a minimal implementation provided by~\citeauthor{Hayes2023Bounding}~\cite{Hayes2023Bounding}, which we extend to incorporate RAD and target-specific auxiliary knowledge. For the imputation attack~\cite{Jayaraman2022Are}, we adapt the authors’ public implementation~\cite{Jayaraman2022Repo}.

\subsection{Results}\label{sec:results}
\begin{figure*}[t] 
  \centering
  \begin{subfigure}[b]{0.4\linewidth}
    \centering
    \includegraphics[width=1\linewidth]{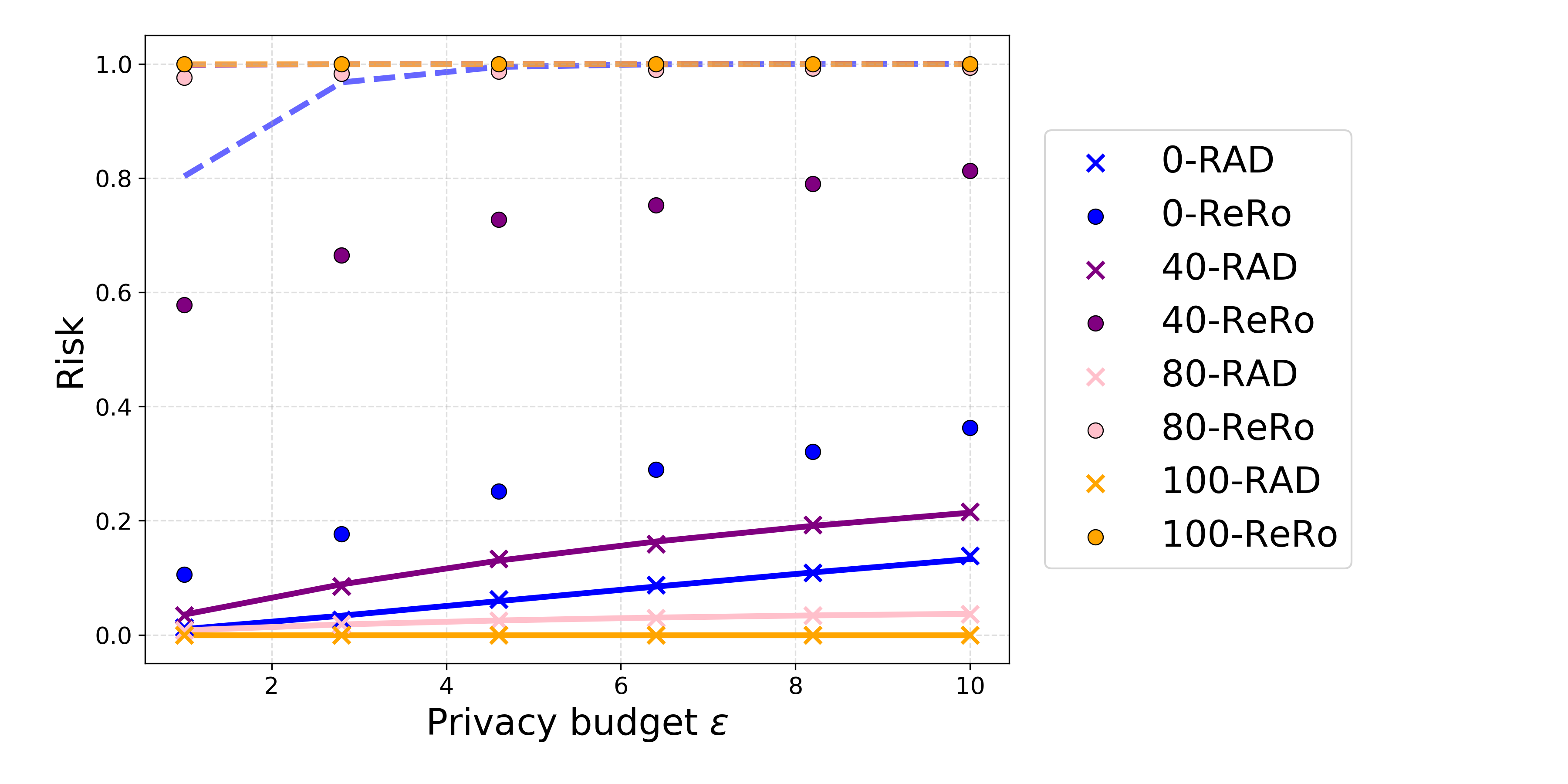}
    \caption{Adult distribution.}
    \label{fig:census_lapalce}
  \end{subfigure}\hfill
  \begin{subfigure}[b]{0.3\linewidth}
    \centering
    \includegraphics[width=1\linewidth]{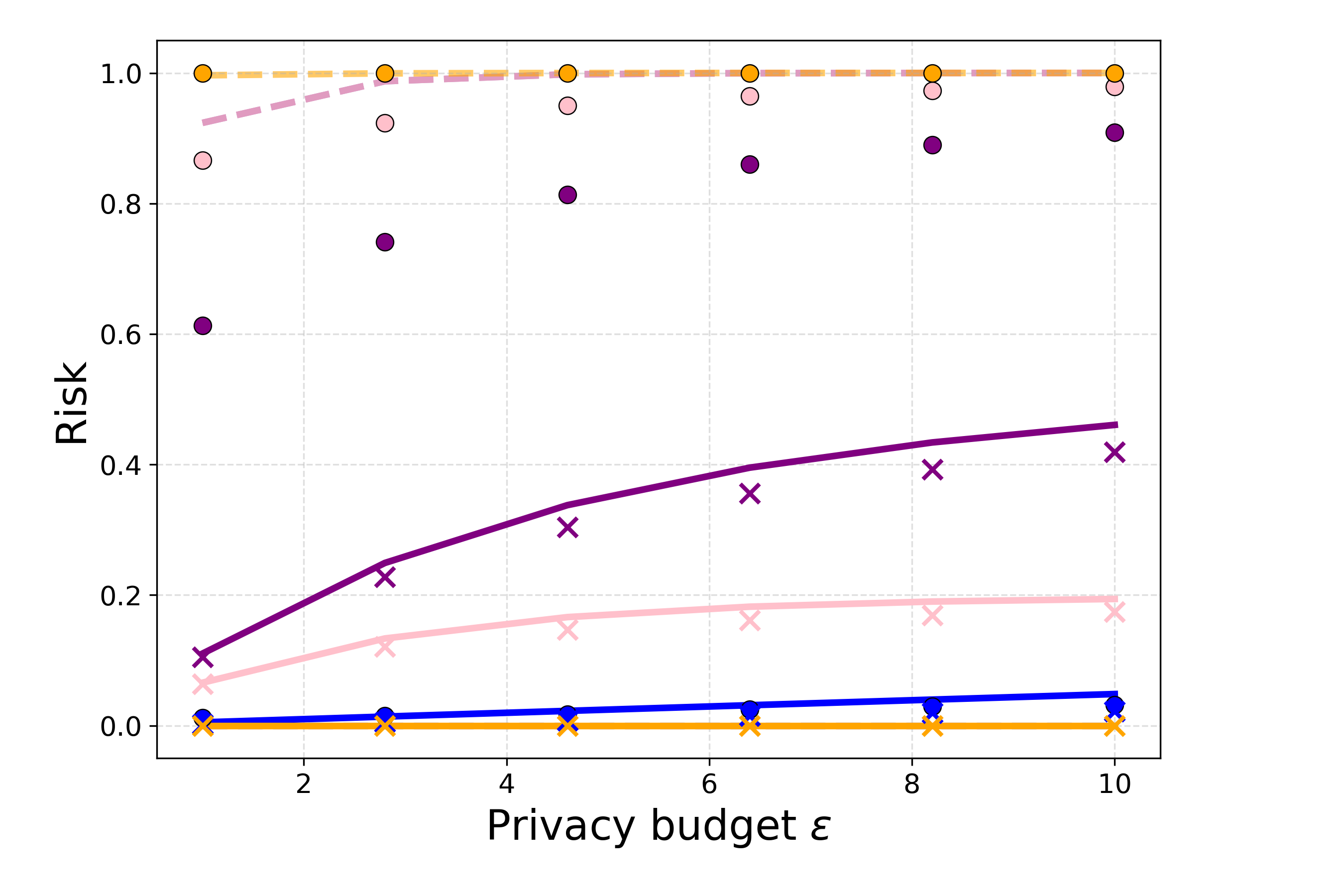}
    \caption{Uniform distribution.}
    \label{fig:uni_laplace}
  \end{subfigure}\hfill
  \begin{subfigure}[b]{0.3\linewidth}
    \centering
    \includegraphics[width=1\linewidth]{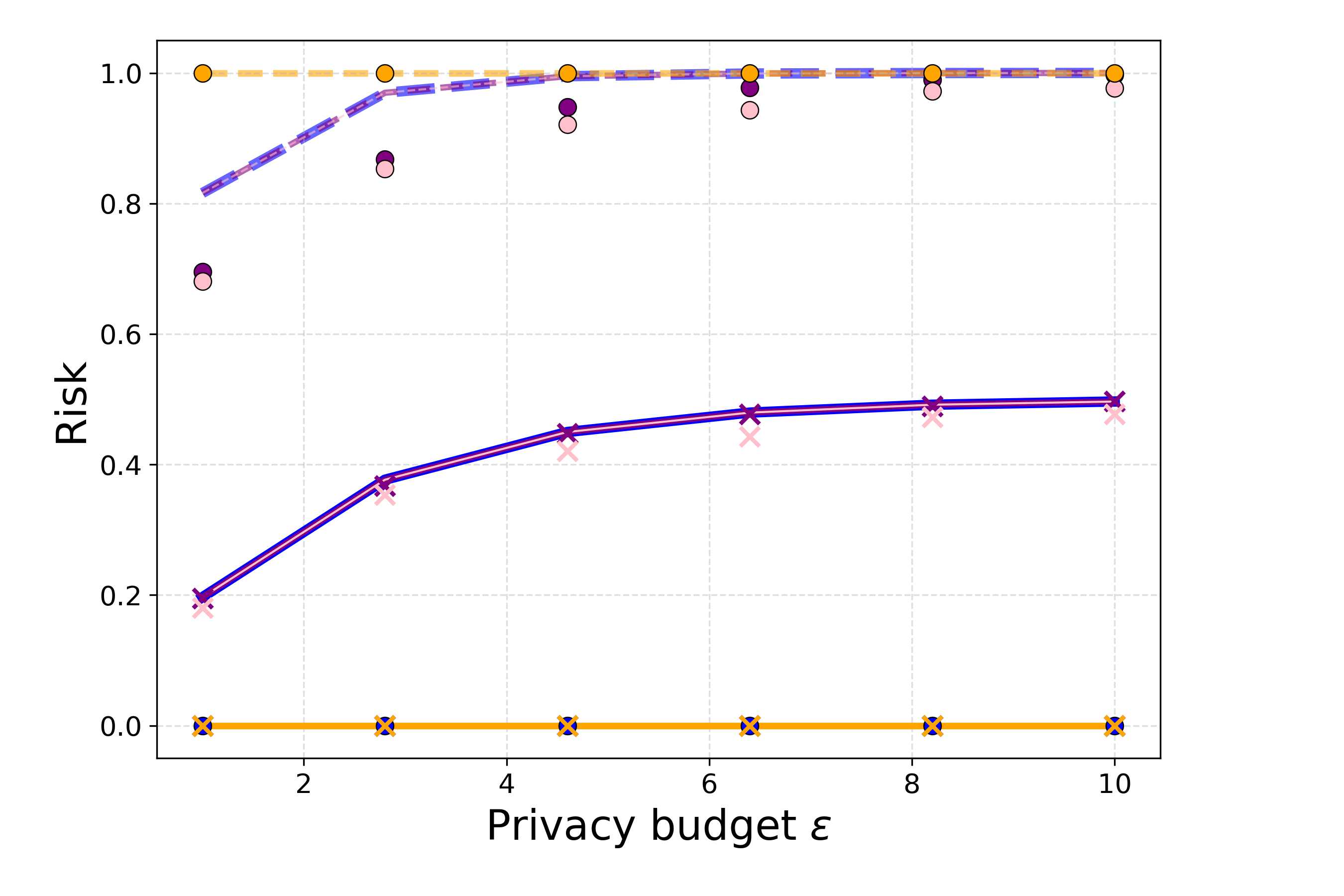}
    \caption{$\pi(0)=\pi(100)=\frac{1}{2}$.}
\label{fig:laplace_exteme}
  \end{subfigure}
  \caption[RAD vs.\ ReRo Results for Optimal Attack against the Laplace Mechanism on Adult]{
Empirical risk of RAD (crosses) and ReRo (dots) for different error tolerances 
$\eta\in[0,100]$ on the Adult dataset with truncated Laplace noise. Straight lines show theoretical RAD bounds, dashed lines ReRo bounds. While RAD bounds closely match the empirical risk, ReRo bounds consistently overestimate the risk. Moreover, ReRo increasingly overestimates the risk at larger error tolerances across all considered distributions.}
  \label{fig:census_attack}
\end{figure*}

\begin{table}[t]
\small
    \centering
        \begin{tabular}{ccc}
            \toprule
             \textbf{Dataset} & \textbf{ReRo}  & \textbf{RAD} \\
            \midrule
            Census & 0.81 & 0 \\
            Texas & 0.73 & 0 \\
            \bottomrule
        \end{tabular}
    \caption[ReRo vs.\ RAD Risk Estimation for Imputation Attack]{ReRo vs.\ RAD risk estimation for imputation attack. This type of attack does not access the dataset directly and therefore cannot induce any participation risk, which RAD correctly captures while ReRo significantly overestimates the risk.}
    \label{tab:bb_ml}
    \normalsize
\end{table}

In this section, we present the RAD and ReRo empirical risk results on real attacks, along with their corresponding theoretical bounds. For both measures, the y-axis represents the risk, with values close to one indicating high risk and values near zero indicating low risk.

\subsubsection{RAD covers, but ReRo breaks for auxiliary knowledge}
\begin{figure}
\begin{subfigure}[b]{0.45\textwidth}
    \centering    \includegraphics[width=0.8\linewidth]{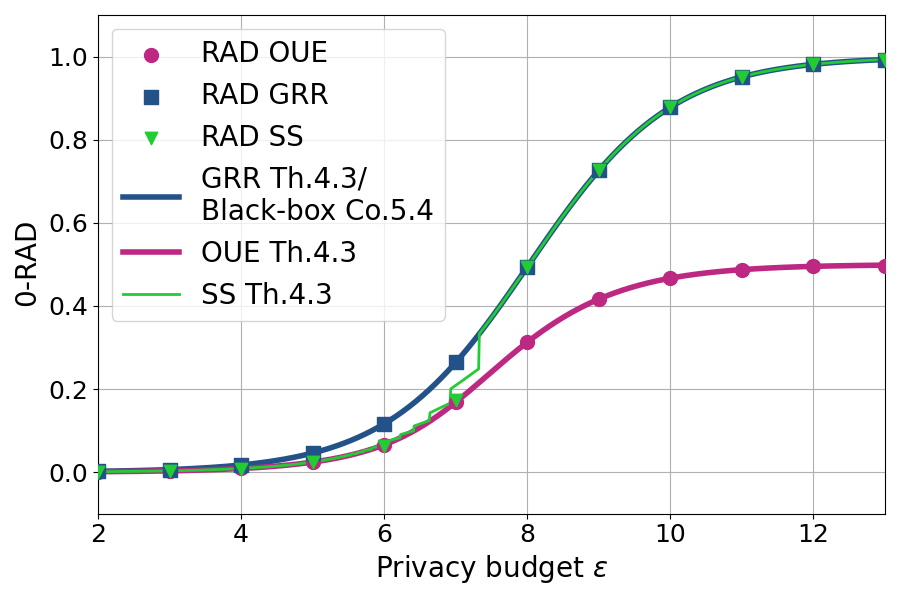}
    \caption{Porto }
    \label{fig:attacks_porto}
\end{subfigure}
\begin{subfigure}[b]{0.45\textwidth}
    \centering    \includegraphics[width=0.8\linewidth]{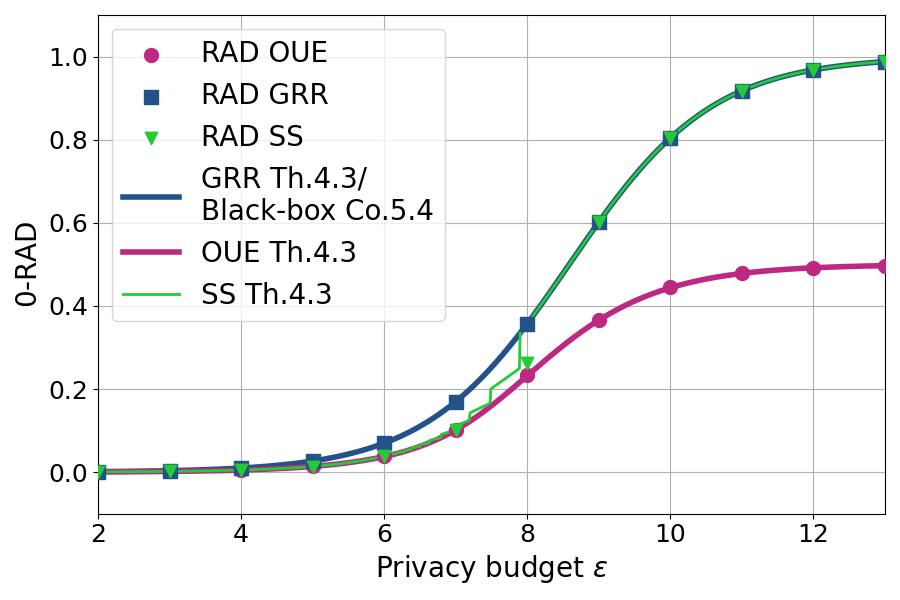}
    \caption{Geolife }
    \label{fig:attacks_beijing}
\end{subfigure}
  \caption[RAD Results for LDP Mechanisms]{RAD results for LDP mechanisms.  Lines show theoretical bounds and markers empirical RAD. First, we note that our bounds are perfectly tight for all tested mechanisms and datasets. Adittionally, we see that OUE offers the higher protection among the LDP mechanisms even for the same $\varepsilon$ choices. }\label{fig:LDP_attacks}
\end{figure}

\Cref{fig:dp-sgd_mnist} shows the results of ReRo and RAD risk estimation for our optimal attacks against DP-SGD on the MNIST dataset. Analogous results for the Fashion dataset are provided in \Cref{fig:dp-sgd_fashion}. We also include the corresponding theoretical bounds for ReRo and RAD for comparison.
As expected, the existing ReRo bounds~\cite{Hayes2023Bounding} correctly provide an upper limit on the empirically observed ReRo risk when the adversary has no prior knowledge of the target record ($aux=\{\varnothing\}$). \Cref{fig:dp-sgd_dra_mnist}). However, when the adversary has prior knowledge of the victim record (\Cref{fig:dp-sgd_aia_mnist,fig:dp-sgd_mia_mnist}), ReRo estimates exceed the values predicted by their theoretical bounds---which are meant to be upper bounds and, therefore, should never be surpassed by the true risk. In contrast, our RAD bounds consistently upper-limit the empirically estimated RAD risks across all tested attacks.

This supports our expectation that the ReRo bounds only hold under the assumption that the adversary has no auxiliary knowledge about the victim ($aux=\{\varnothing\}$), but fail to correctly estimate privacy risks when target-specific auxiliary knowledge exists.

We can also observe that our bounds for RAD overcome this estimation error: they hold for any auxiliary knowledge and are nearly tight. 
In particular,~\Cref{fig:dp-sgd_mia_mnist,fig:dp-sgd_aia_mnist} show that the tightness of our worst-case bound~\Cref{th:dp_implies_aux-urero} is not an isolated feature of GRR, but a reliable property that also applies to other widely used mechanisms, such as DP-SGD.  Finally, \Cref{fig:dp-sgd_dra_mnist} shows that our closed-form bound~\Cref{th:f-DP} offers a reasonable upper-bound  when \Cref{th:optimal_bound} needs to be numerically approximated (as is the case, for instance, with DP-SGD).

\subsubsection{Leakage vs.\ Imputation}
\Cref{tab:bb_ml} compares the risk estimates of RAD and ReRo for the imputation attack. 
This attack is not based on any information leakage from the mechanism and ignores any output in the process.
RAD in this case does estimate the privacy risk to be $0$, whereas ReRo reports notably higher values ($0.81$ for Census and $0.73$ for Texas).
This underlines how RAD is the more reliable measure of actual privacy risks: 
RAD shows the absence of leakage when the attack's success relies solely on imputation, whereas ReRo suggests serious disclosures (or: attack potential), effectively overestimating the privacy risk. This result suggests that RAD is a safer choice for risk estimation, as it allows practitioners to measure the true risk of data disclosure without being affected by data imputation.

This tendency of ReRo to overestimate risk is not confined to this setting. In our optimal attacks on DP-SGD (\Cref{fig:dp-sgd_mnist}), ReRo consistently overestimates leakage across all investigated cases, with the effect becoming more pronounced as more auxiliary information is incorporated. Membership inference ($a(z)=z$) provides the clearest example, where ReRo reports risk values exceeding $0.6$ even for privacy budgets $\varepsilon \leq 4$, which are commonly considered to offer strong privacy guarantees~\cite{Lee2011How}. This behavior aligns with expectations, as ReRo cannot discount auxiliary information; consequently, greater attacker knowledge leads to larger overestimation.

Similarly, \Cref{fig:census_attack} shows that ReRo fails to capture the effect of the success threshold~$\eta$. As $\eta$ increases, an oblivious attacker’s success probability rises, but ReRo cannot account for this since it depends only on success probability and thus converges to $1$ for all $\varepsilon$. This results in substantial overestimation: for $\eta=100$, a trivial setting where any guess is correct, ReRo reports maximal risk despite the mechanism providing no advantage. In contrast, RAD properly discounts this effect, showing that increasing $\eta$ boosts advantage only up to a point (here, $\eta=40$), after which the advantage decreases as success becomes nearly granted.

\subsubsection{Bound tightness} \Cref{fig:census_attack} shows the results of RAD and ReRo for our optimal attack against Laplace mechanism on Adult including their corresponding theoretical bounds. \Cref{fig:attacks_porto,fig:attacks_beijing} shows the analogous for LDP mechanisms, GRR, OUE and SS, on the Porto and Geolife datasets.  On the z-axis, we see $\varepsilon$, and on the y-axis, the exact  estimated risk for such $\varepsilon$ selection. Note that for LDP,  RAD and ReRo results coincide, since the attack relies solely on the released output (with no auxiliary information or imputation effects). Moreover, the prior-based chance level under the uniform prior  is negligible for 
$|\Z|=3,052$. We therefore report only RAD to avoid redundancy.

We observe that our bounds (cf.~\Cref{th:optimal_bound}) are tight for every prior $\pi$ and capture even subtle differences between mechanisms. In particular, the RAD estimates for GRR perfectly match our perfect-reconstruction black-box bound (\Cref{th:perfect_reco_bb}), confirming its tightness.

Moreover, \Cref{fig:census_attack} clearly illustrates the impact of the data distribution: the skewed distribution (\Cref{fig:laplace_exteme}) constitutes the worst case, while the empirical distribution represents the best case. This highlights that knowledge of the data distribution can substantially improve utility; in the absence of such knowledge,  we must fall back to the worst-case scenario.

Finally, these results provide concrete evidence for the importance of attack-based noise calibration. For identical values of $\varepsilon$, OUE offers significantly stronger protection against DRAs than GRR and SS. Hence, $\varepsilon$ alone does not capture the full privacy picture, and RAD is essential for understanding the actual privacy implications of a mechanism for users.

\subsubsection{Auditing Local DP with RAD}
\begin{figure*}[t] 
  \centering
  \begin{subfigure}[b]{0.33\textwidth}
    \centering
    \includegraphics[width=0.9\linewidth]{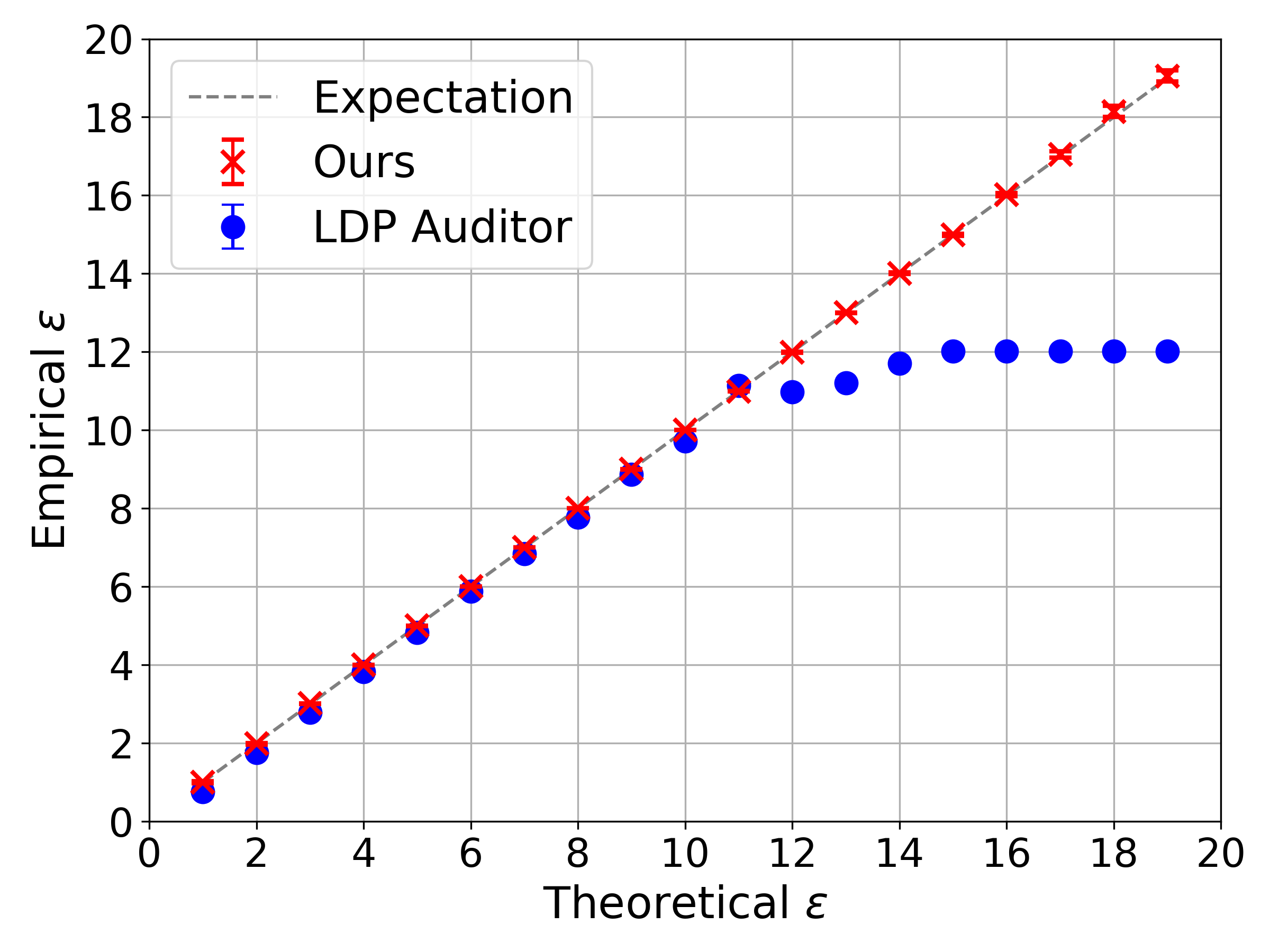}
    \caption{(GRR)}
    \label{fig:audit_grr_porto}
  \end{subfigure}\hfill
  \begin{subfigure}[b]{0.33\textwidth}
    \centering
    \includegraphics[width=0.9\linewidth]{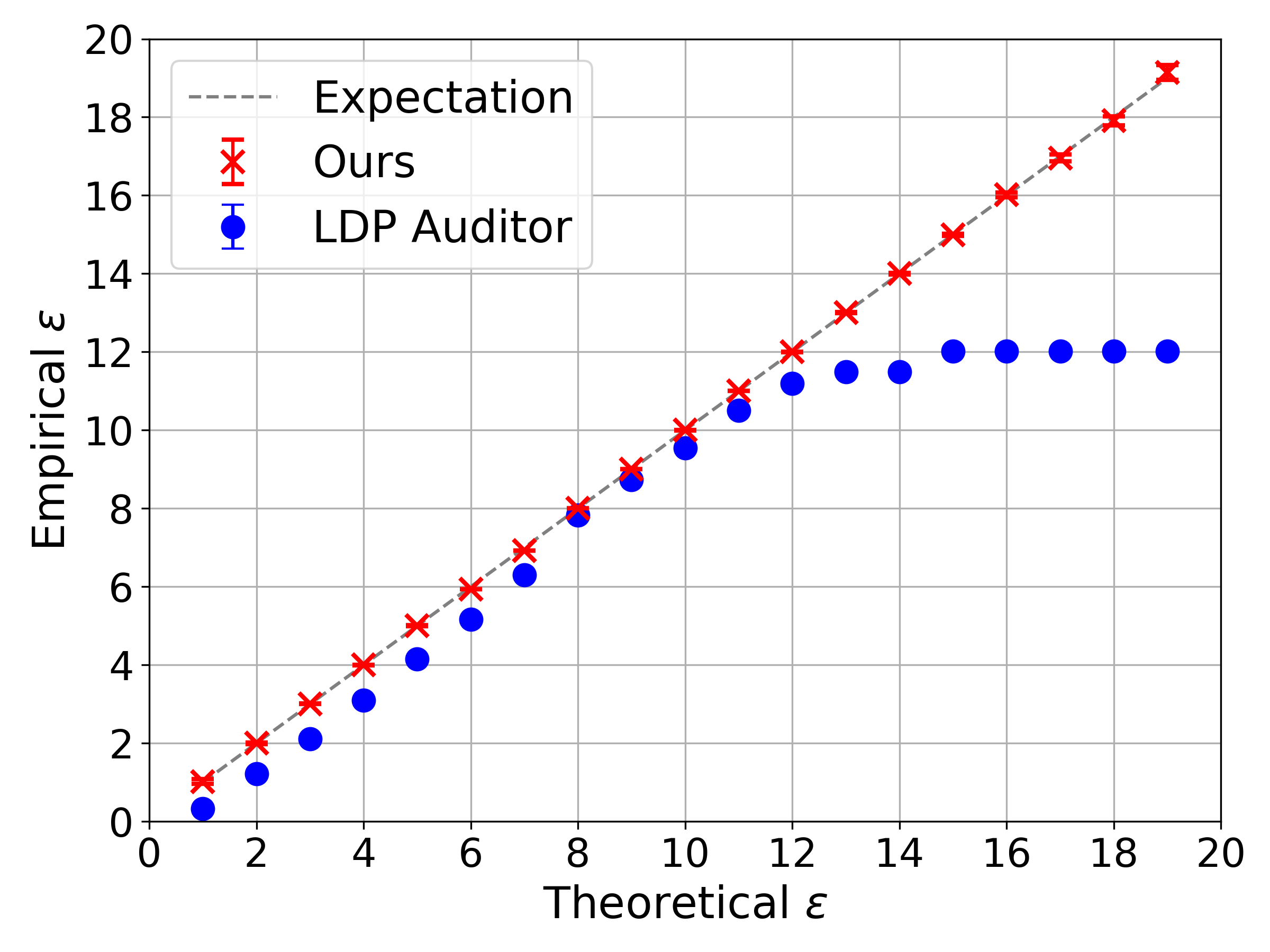}
    \caption{(SS)}
    \label{fig:audit_ss_porto}
  \end{subfigure}\hfill
  \begin{subfigure}[b]{0.33\textwidth}
    \centering
    \includegraphics[width=0.9\linewidth]{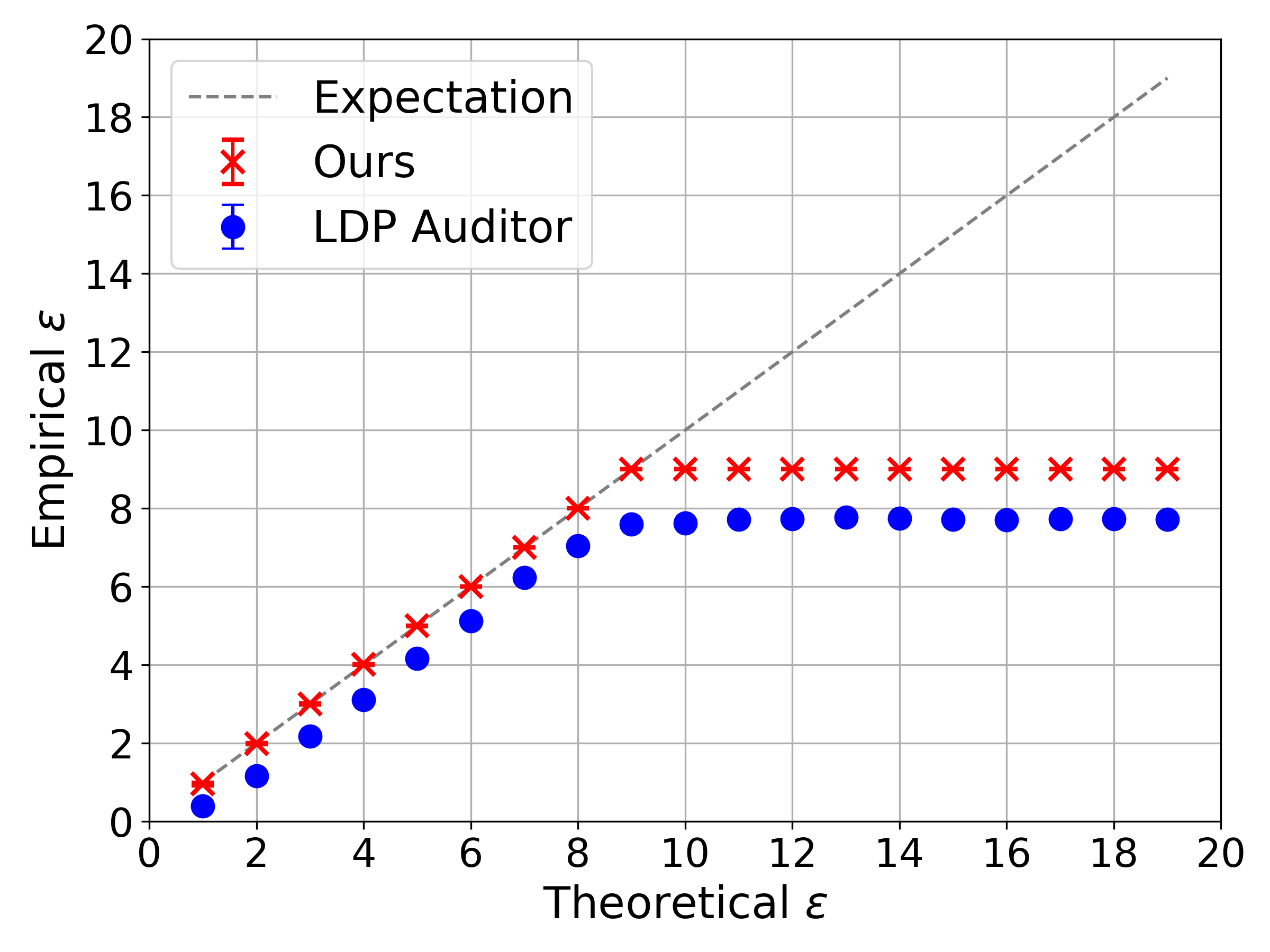}
    \caption{(OUE)}
    \label{fig:audit_oue_porto}
  \end{subfigure}
  \caption[LDP Audit Results from RAD-Based Auditing and \textsc{LDP Auditor} on the Porto Dataset]{LDP Audit results from RAD-based auditing and \textsc{LDP Auditor}~\cite{Arcolezi2024Revealing} on Porto dataset. Values along the diagonal indicate perfect accuracy; below it, privacy is overestimated; above it, underestimated.}
  \label{fig:audit_porto}
  
\end{figure*}
\begin{figure*}[t] 
  \centering
  \begin{subfigure}[b]{0.33\textwidth}
    \centering
    \includegraphics[width=0.9\linewidth]{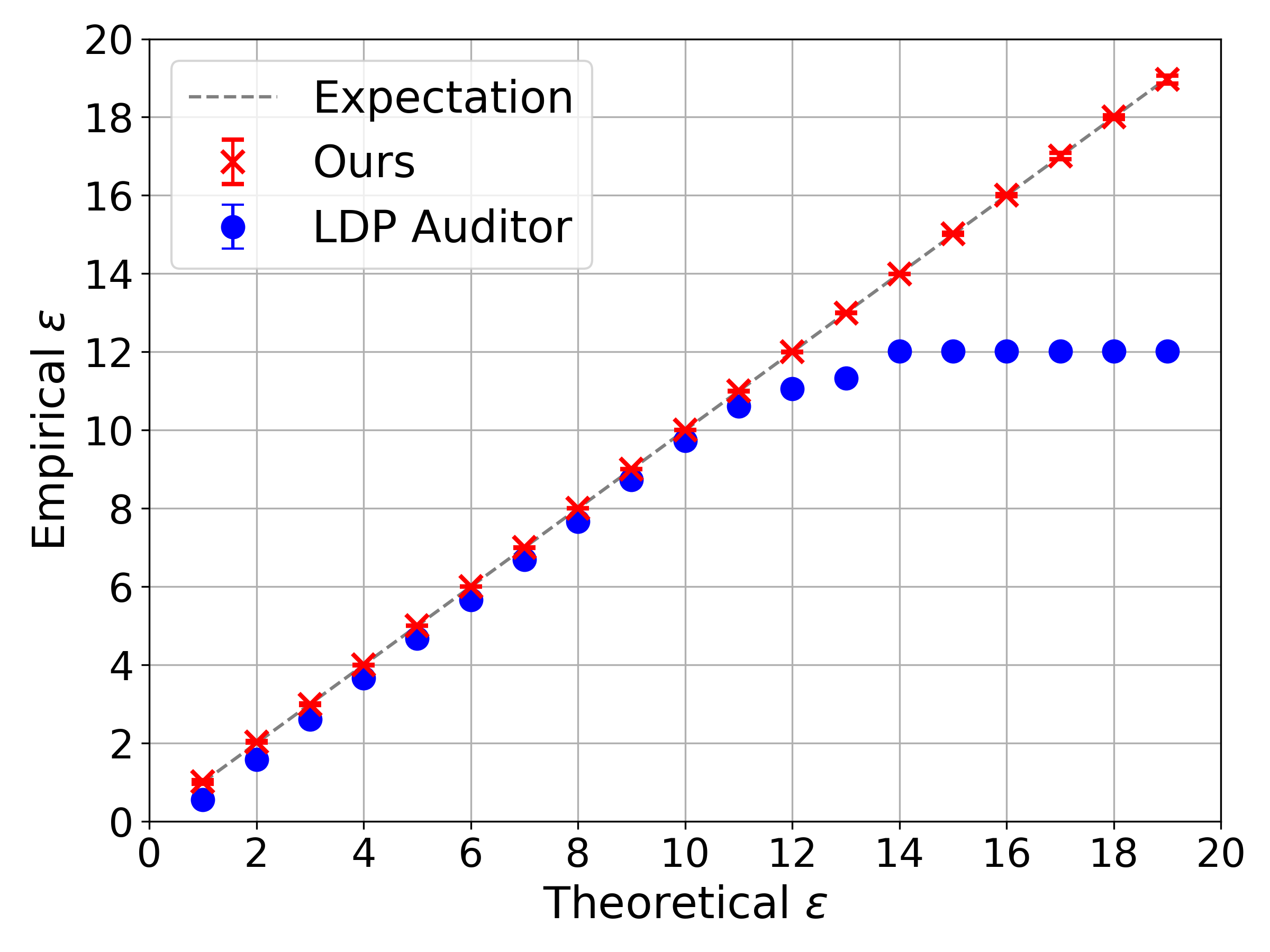}
    \caption{(GRR)}
    \label{fig:audit_grr_beijing}
  \end{subfigure}\hfill
  \begin{subfigure}[b]{0.33\textwidth}
    \centering
    \includegraphics[width=0.9\linewidth]{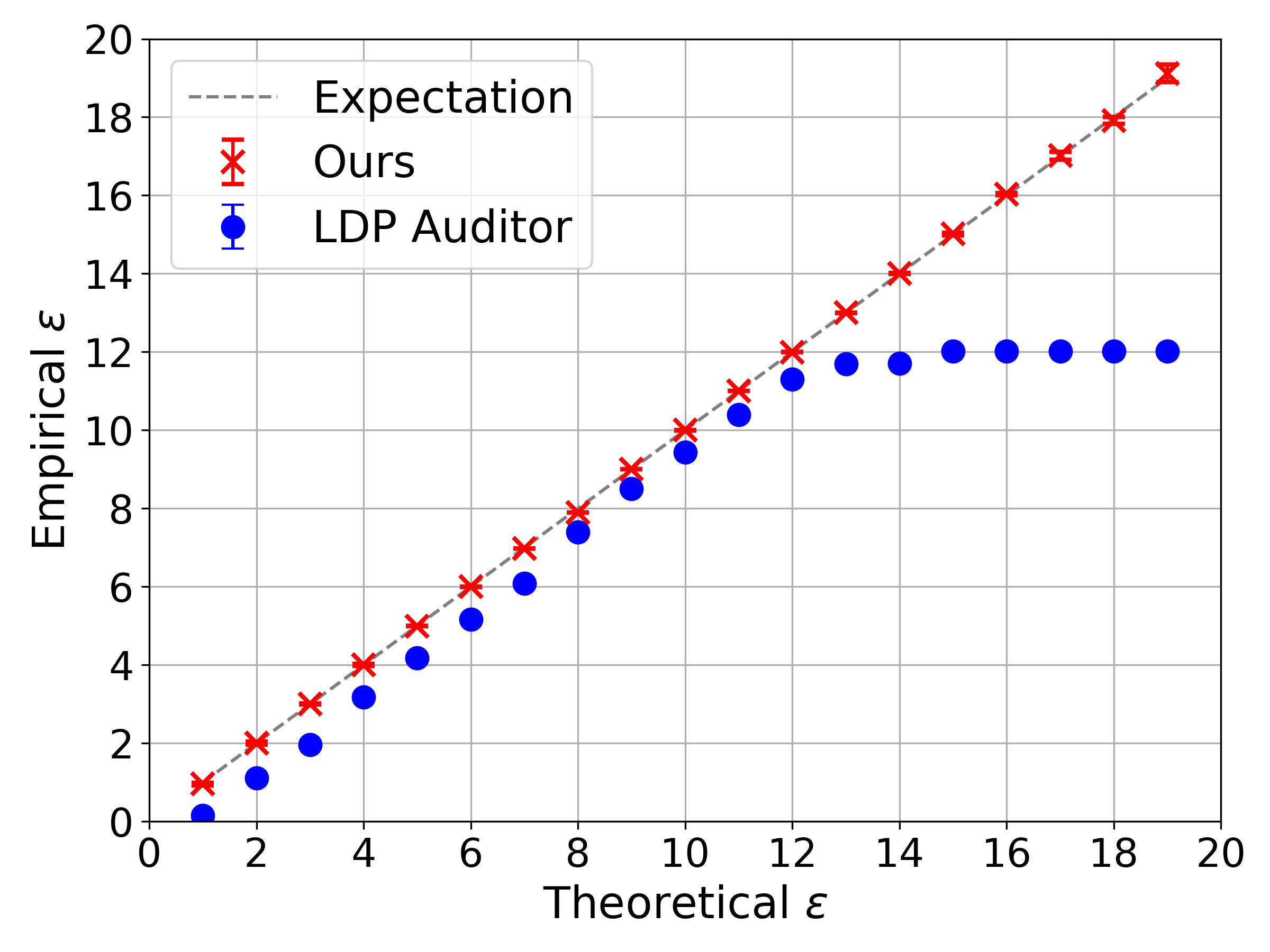}
    \caption{(SS)}
    \label{fig:audit_ss_beijing}
  \end{subfigure}\hfill
  \begin{subfigure}[b]{0.33\textwidth}
    \centering
    \includegraphics[width=0.9\linewidth]{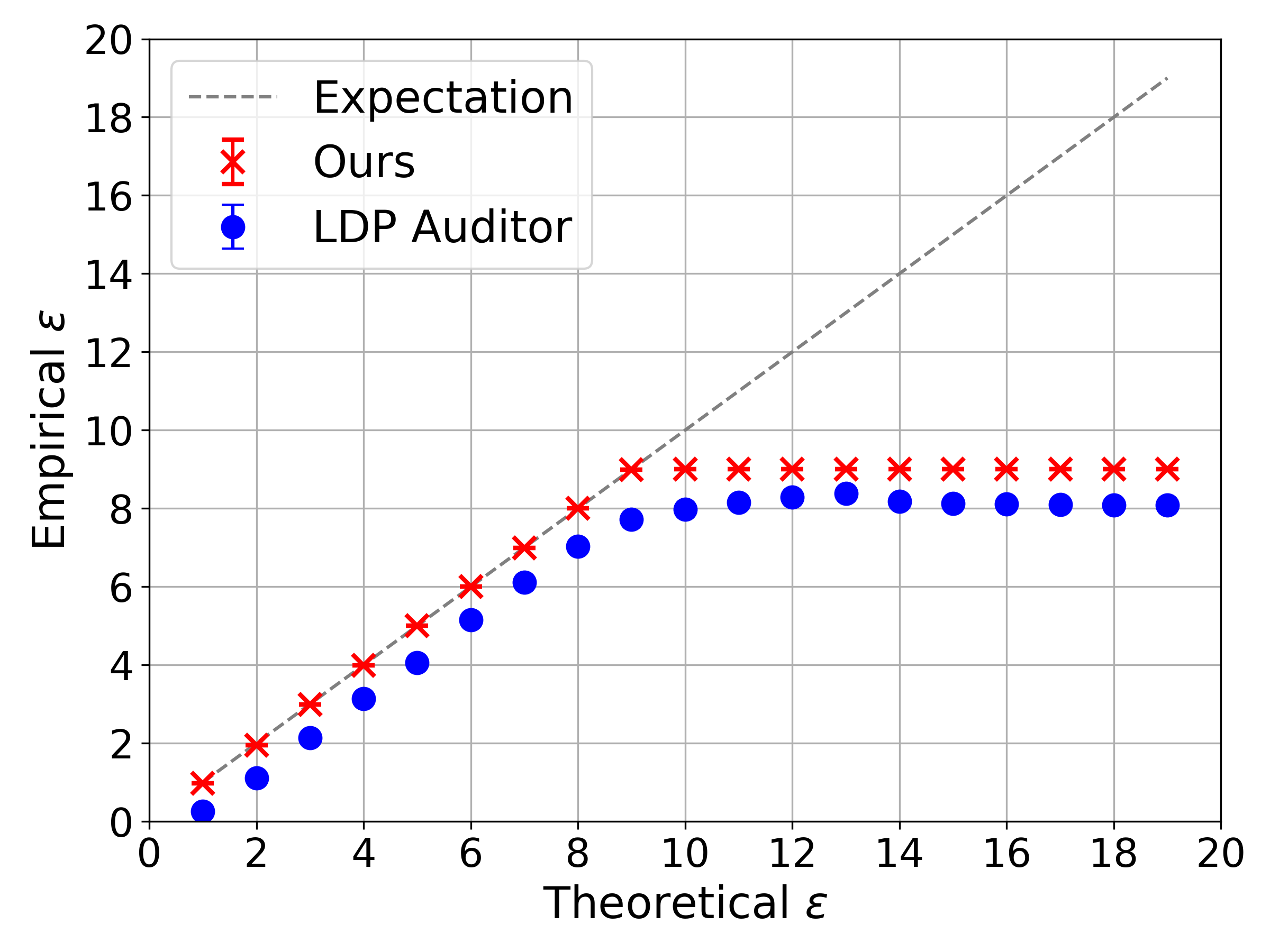}
    \caption{ (OUE)}
    \label{fig:audit_oue_beijing}
  \end{subfigure}
  \caption[LDP Audit Results from RAD-Based Auditing and \textsc{LDP Auditor} on the Geolife Dataset]{LDP Audit results from RAD-based auditing and \textsc{LDP Auditor}~\cite{Arcolezi2024Revealing} on Geolife dataset. Values along the diagonal indicate perfect accuracy; below it, privacy is overestimated; above it, underestimated.}
  \label{fig:audit_beijing}
  
\end{figure*}

\Cref{fig:audit_porto,fig:audit_beijing} show the results from our LDP auditing experiments using the Porto and Geolife datasets.
They compare the accuracy of predicting the actual $\varepsilon$ using our RAD-based auditing versus~\textsc{LDP auditor}. The closer the empirical $\varepsilon$ is to the theoretical value (diagonal line), the more accurate the auditing tool. Additionally, smaller standard deviations indicate greater stability of the method. 

For all tested mechanisms, our auditing approach improves over \textsc{LDP Auditor} for all $\varepsilon$ values. 
In particular, we see that the highest $\varepsilon$ \textsc{LDP Auditor} manages to estimate for both GRR and SS are capped around $\widetilde{\varepsilon} \approx 12.25$, hence preventing auditing of deployments with higher values. This limitation was already acknowledged by the authors of~\textsc{LDP Auditor}, as it stems from the intrinsic shortcomings of the Clopper-Pearson method underlying their approach~\cite{Arcolezi2024Revealing}. 
In contrast, the tightness of our RAD bound enables our auditing approach to accurately estimate empirical privacy budgets for the whole range, without such a limitation. 
Notably, for GRR and SS, our DP auditing yields near-perfect estimates for all epsilon values. 
For the OUE mechanism, our approach also outperforms \textsc{LDP Auditor}, however, the estimation accuracy declines at $\varepsilon\leq 9$. Note that this is an inherent limitation of OUE auditing as already mentioned in~\cite{Arcolezi2024Revealing}: as we prove in~\Cref{ex:oue}, $0$-RAD converges to $\frac{m-1}{2m}$ when $\varepsilon$ tends to infinity.  
Overall, these results support that the universal tightness of our theoretical bound \Cref{th:optimal_bound} enables precise and reliable auditing based on DRAs.

\section{Conclusion}
In this paper, we investigate the reconstruction risk that users incur when their data are processed by DP mechanisms.
Our results reveal that the current state-of-the-art risk metric, ReRo \cite{Balle2022Reconstructing}, drastically overestimates the actual leakage of DP mechanisms when target-specific public knowledge exists---leading to excessive utility loss if used as noise calibration methods. Crucially, we show that under real attacks, existing ReRo bounds are violated.

To address these limitations, we first introduce $\eta$-RAD, a novel metric consistent with attribute and membership advantage, that accurately captures the privacy risk imposed by any specific mechanism. 
More importantly, we advance the understanding and practical interpretation of DP guarantees by proving tight bounds that connect DP mechanisms with their risk, using RAD. Offering new insights and clarity beyond existing analyses, we establish (i) universally tight bounds when the attacker’s knowledge is specified, along with optimal strategies achieving them, (ii) closed-form bounds that remain valid regardless of auxiliary knowledge, and (iii) black-box upper bounds for settings with completely secret records.  Our theoretical and empirical evaluation---across private learning, DP aggregation and LDP settings---demonstrates not only the robustness of RAD as a risk measure, but also the significant impact of our bounds on improving DP noise calibration (proving better utility) and auditing in DP (broadening the scope and improving accuracy).

Overall, our work demonstrates that privacy risk depends on the mechanism's structure, not just its nominal privacy parameters, and provides both fundamental insight and practical tools for privacy risk assessment and calibration---enabling notable utility gains without increasing the effective privacy risk.

\section*{Acknowledgments}
  This work was funded by the Topic Engineering Secure Systems of the Helmholtz Association (HGF) and supported by KASTEL Security Research Labs, Karlsruhe, and Germany’s Excellence Strategy (EXC 2050/2 ‘CeTI’; ID 390696704).  
  H.H. Arcolezi has been partially supported by the French National Research Agency (ANR), under contracts: ``ANR-24-CE23-6239'' and ``ANR-23-IACL-0006''  

The authors thank Ana-María Crețu for her helpful comments and feedback on an earlier draft of this work, and Daniel Schadt for his support with the code implementation.

%\clearpage
\printbibliography

%% ----------------
%% |   Appendix   |
%% ----------------

\appendix
\section{Additional Proofs and Remarks}\label{ap:rad}
In this section, we provide the omitted details needed to complete the results presented in this work.

\subsection{Detailed Examples of RAD Computation}
Here, we provide the detailed computation of RAD for~\Cref{ex:oue_short,ex:ss_short,ex:laplace_short,ex:gaussian_short} using~\Cref{th:optimal_bound}.
\begin{example}\label{ex:oue}
    In  the optimal unary encoding (OUE) mechanism~\cite{Wang2017Locally} each user encodes its input $z\in\Z$ as a one-hot $m$-dimensional binary vector and perturbs each bit independently. For each position $i\in[m]$, the obfuscated vector $\theta$ is sampled such that $\Pr[\theta_i=1]=1/2$ if $i=z$, and $q=\frac{1}{\e^{\varepsilon}+1}$ otherwise. Denoting $p=1-q$ and $k_{\theta}=\#\{\theta_i=1\}$, for every $\theta$ such that $k_{\theta}\geq 1$, we have that 
\begin{gather}
    \Pr(\theta\mid z)=\begin{cases}
        P\equiv \frac{1}{2}q^{k_{\theta}-1}p^{m-k_{\theta}} &\text{ if } \theta_x=1\\
        Q\equiv \frac{1}{2}q^{k}p^{m-k_{\theta}-1} &\text{ if } \theta_x\neq1
    \end{cases}
\end{gather}
and $\Pr(\Vec{0}\mid z)=\frac{1}{2}p^{m-1}$.
Hence, $\Pr(\Vec{0})=\frac{1}{2}p^{m-1}$ and $w(\Vec{0},z)=0$ for all $z$. For all $\theta\neq\Vec{0}$ we obtain,
\begin{align*}
p(\theta)&=\frac{1}{2}q^{k_{\theta}-1}p^{m-k_{\theta}}\underbrace{(\sum_{z\colon \theta_x=1}\pi_x)}_{S_{\theta}} + \frac{1}{2}q^{k}p^{m-k_{\theta}-1}(1-\sum_{z\colon \theta_x=1}\pi_x)\\
\end{align*}
Note that $P-Q=\frac{p-q}{2}(q^{k_{\theta}-1}p^{m-k_{\theta}-1})\geq 0$.
Consequently,
\[
w(\theta,z)=\begin{cases}
    (P-Q)\,(1-S_{\theta})\geq 0& \text{ if } \theta_x=1\\
    (Q-P)S_{\theta}\,\leq 0& \text{otherwise.} 
\end{cases}
\]
Applying~\Cref{th:optimal_bound} for $a(z)=z$ we obtain,
\begin{align*}
    \eta\text{-}\mathrm{RAD} 
&\le \sum_{\theta} \sum_{z:\theta_x=1} (P-Q) \, (1-\sum_{\theta_x=1}\pi_x)\pi_x \\
&=\frac{p-q}{2}\sum_{\theta} \sum_{z:\theta_x=1} (q^{k_{\theta}-1}p^{m-k_{\theta}-1}) \, (1-\sum_{\theta_x=1}\pi_x)\pi_x\\
&= \frac{p-q}{2p} \sum_{k=1}^{m} q^{k-1} p^{m-k} \sum_{\theta : k_\theta = k} 
\sum_{\theta_x=1} \pi_x \Big(1 - \sum_{\theta_x=1} \pi_x \Big) \\
&= \frac{p-q}{2p} \sum_{k=1}^{m} q^{k-1} p^{m-k} \binom{m-2}{k-1} (1 - \kappa_\pi) \\
&= (1 - \kappa_\pi) \frac{p-q}{2p} \sum_{k=1}^{m} \binom{m-2}{k-1} q^{k-1} p^{m-k} \\
&= (1 - \kappa_\pi) \frac{p-q}{2p} \sum_{r=0}^{m-2} \binom{m-2}{r} q^r p^{m-1-r} 
\quad (\text{index change } r=k-1) \\
&= (1 - \kappa_\pi) \frac{p-q}{2p} \, p^{m-1} \sum_{r=0}^{m-2} \binom{m-2}{r} (q/p)^r \\
&= (1 - \kappa_\pi) \frac{p-q}{2p} \, p^{m-1} (1 + q/p)^{m-2} 
\quad (\text{binomial identity}) \\
&= (1 - \kappa_\pi) \frac{p-q}{2p} \, p^{m-1} (1/p)^{m-2} 
\quad (\text{since } p+q=1) \\
&= (1 - \kappa_\pi) \frac{p-q}{2} \\
&= \frac{1}{2} \frac{e^\varepsilon - 1}{e^\varepsilon + 1} (1 - \kappa_\pi) 
= \mathrm{TV}(\mathcal{M}) \, (1 - \kappa_\pi).
\end{align*}

When $aux=\{\varnothing\}$, we have
\[
\max_{z\in\Z}w(\theta,z)\pi_x=(P-Q)(1-\sum_{\theta_x=1}\pi_x)\max_{\theta_x=1}\pi_x.
\]
Hence,
\begin{align*}
    \eta\text{-}\mathrm{RAD} 
&\le \sum_{\theta} (P-Q) \, (1-\sum_{\theta_x=1}\pi_x)\max_{\theta_x=1}\pi_x \\
&=\frac{p-q}{2}\sum_{\theta}  (q^{k_{\theta}-1}p^{m-k_{\theta}-1}) \, (1-\sum_{\theta_x=1}\pi_x)\max_{\theta_x=1}\pi_x\\
\end{align*}
We order $\pi_1\leq\dots\leq\pi_m$. Then,
\begin{align*}
    \Theta_i=\{\theta\colon \max_x w(\theta,z)\pi_x=w(\theta,z_i)\}=\{\theta\colon \theta_i=1\wedge\theta_j=0\text{ for all }j>i\},
\end{align*}
and we can rewrite
\begin{align*}
    \eta\text{-}\mathrm{RAD}&
\leq \frac{p-q}{2p}\sum_{i=1}^m \pi_i \underbrace{\sum_{\theta\in\Theta_i}q^{k_{\theta}-1}p^{m-k_{\theta}}(1-\sum_{\theta_x=1}\pi_x)}_{A_i}
\end{align*}

For every $k = 1,\dots,i$, there are $\binom{i-1}{k-1}$ vectors $\theta \in \Theta_i$ such that $k_\theta=k$. Moreover, the sum $\sum_{z \in S} \pi_x$ over all sets $S$ of size $k$ containing $i$:

\[
\sum_{\substack{S \subseteq \{1,\dots,i\} \\ i \in S, |S|=k}} \sum_{z \in S} \pi_x
= \pi_i \binom{i-1}{k-1} + \sum_{z=1}^{i-1} \pi_x \binom{i-2}{k-2}.
\]

Hence,

\begin{align*}
A_i =& \sum_{k=1}^{i} q^{\,k-1} p^{\,m-k} \Bigg[ \binom{i-1}{k-1} (1-\pi_i) - \binom{i-2}{k-2} \sum_{z=1}^{i-1} \pi_x \Bigg]\\
=& (1-\pi_i)\sum_{k=1}^{i}\binom{i-1}{k-1} q^{\,k-1} p^{\,m-k} -\left(\sum_{z=1}^{i-1}\pi_x\right)\sum_{k=1}^{i}\binom{i-2}{k-2} q^{\,k-1} p^{\,m-k} \\
=&(1-\pi_i)p^{m-i}\sum_{r=0}^{i-1}\binom{i-1}{r} q^{\,r} p^{\,i-1-r} -\left(\sum_{z=1}^{i-1}\pi_x\right)qp^{m-i}\sum_{j=0}^{i-2}\binom{i-2}{j} q^{\,j} p^{(i-2-j)}\\
=&(1-\pi_i)p^{m-i}+qp^{m-i}\left(\sum_{z=1}^{i-1}\pi_x\right).
\end{align*}
since for $k=1$, $\binom{i-2}{k-1}=0$, so we can start in $2$, i.e., $k=j+2$. Hence,
\[
A_i = p^{\,m-i} \Big[ (1-\pi_i) - q \sum_{z=1}^{i-1} \pi_x \Big], \qquad i=1,\dots,m.
\]
so
\[
\eta\text{-}\mathrm{RAD}\leq \frac{p-q}{2p}\left(\sum_{i=1}^m p^{m-i}\pi_i(1-\pi_i) -q\sum_{i=1}^m p^{m-i}\pi_i\sum_{z=1}^{i-1}\pi_x\right)
\]
For instance, $\pi_i=\frac{1}{m}$
\begin{align*}
    \eta\text{-}\mathrm{RAD}&\leq \frac{p-q}{2p}\sum_{i=1}^m p^{m-i}\pi_i(1-\pi_i) -q\sum_{i=1}^m p^{m-i}\pi_i\sum_{z=1}^{i-1}\pi_x\\
    &=\frac{p-q}{2mp}\left(\frac{m-1}{m}\sum_{i=1}^m p^{m-i}-q\frac{1}{m}\sum_{i=1}^m p^{m-i}(1-i)\right)\\
&= \frac{p-q}{2p}\left[\frac{1}{m}\Big(1-\frac{1}{m}\Big)\frac{p^m-1}{p-1}
- q\cdot\frac{1}{m^2}\cdot\frac{p^m-1-m(p-1)}{(p-1)^2}\right] \\[2mm]
&\overset{p=1-q}{=} \frac{1-2q}{2(1-q)}\left[-\frac{m-1}{m^2}\cdot\frac{(1-q)^m-1}{q}
- \frac{1}{m^2}\cdot\frac{(1-q)^m-1+mq}{q}\right] \\[1mm]
&= \frac{1-2q}{2(1-q)}\left(-\frac{1}{m q}\big((1-q)^m-1+q\big)\right) \\[1mm]
&= \frac{2q-1}{2(1-q)}\cdot \frac{(1-q)^m-1+q}{m q} \\[1mm]
&\overset{q=1-p}{=} \frac{1-2p}{2(1-p)p}\cdot \frac{p^m-p}{m} 
= \frac{1-2p}{2(1-p)}\cdot \frac{p^{m-1}-1}{m} \\[1mm]
&= \frac{(2p-1)\big(1-p^{\,m-1}\big)}{2m(1-p)}.
\end{align*}
Note that
\[
\lim_{\varepsilon\to \infty}\frac{\e^{\varepsilon}-1}{2m}\left(1-\left(\frac{\e^{\varepsilon}}{1+\e^{\varepsilon}}\right)^{\left(m-1\right)}\right) \;=\; \frac{m-1}{2m},
\]
hence even if we keep reducing the noise (increasing $\varepsilon$), the attacker's advantage is limited.
\end{example}
\begin{example}[Subset Selection, $aux=\{\varnothing\}$]\label{ex:ss}
In the subset selection mechanism (SS)~\cite{Min2018Optimal} users report a subset $\theta \subseteq \Z=\{z_1,\dots,z_m\}$ containing their true value $z$ with probability $p = \frac{\omega \e^\varepsilon}{\omega \e^\varepsilon + m - \omega}$, where $\omega = |\theta| = \max\left(1, \left\lfloor \frac{m}{\e^\varepsilon + 1} \right\rfloor \right)$. The subset is completed by sampling uniformly from $\Z \setminus \{z\}$. 

Note that, given $A=\binom{m-1}{\omega-1}$ and $B=\binom{m-1}{\omega}$, 
\begin{gather}
    \Pr_{\M}(\theta\mid z)=\begin{cases}
        \frac{p}{A} &\text{ if }z\in\theta\\
        \frac{1-p}{B} &\text{ if }z\notin\theta\\
    \end{cases}
\end{gather}
Since $|\Theta|=\binom{m}{\omega}$ we have that, according to~\Cref{eq:no_auz_uni}, for $\pi=U[m]$,
\begin{gather}
    \text{0-RAD}\leq\frac{1}{m}(\sum_{\theta\in\Theta}\max_{z}p_{\M}(\theta\mid z)-1)\\
    =\frac{1}{m}\binom{m}{\omega}\frac{p}{A}=\frac{m}{m \omega}p-\frac{1}{m}=\frac{pm-\omega}{m\omega}.
\end{gather}

\end{example}
\begin{example}[Gaussian mechanism and $aux=\{\varnothing\}$]\label{ex:gauss}
The Gaussian mechanism adds Gaussian noise $\mathcal{N}(0,\sigma)$ the query value $q(D)\in\R$~\cite{balle2018gaussian}.
If $\Z=\{z_1,\dots,z_m\}$ is uniformly distributed and $\Delta q=1$, 
\begin{equation*}
     z\in \argmax_{j}w(\theta,z_j)\pi_j \Leftrightarrow z\in \argmax_{j}w(\theta,z_j).
\end{equation*}
Hence, applying \Cref{eq:no_auz_uni} we obtain 
\begin{gather}
    0\text{-RAD}\leq \frac{1
    }{m}\sum_{i=1}^{m} \left(\Pr_{\M}(\Theta_{i}\mid z_i)-\Pr_{\M}(\Theta_{i})\right)=\frac{1
    }{m}\sum_{i=1}^{m} \left(\Pr_{\M}(\Theta_{i}\mid z_i)-1\right)
\end{gather}
Note that for each $z$, $\Pr_{\M}(\theta\mid z)=\Pr_{\M}(\theta\mid q(D_{z}))$. Since $D_{-}$ is fixed, $q(D_{z})$ is completely determined by $z$, hence we use the abuse of notation $q(D_{z})\equiv z$.
We want to compute $\Pr_{\M}(\Theta_{i}\mid z_i)$ for $i\in[m]$. Without loss of generality we re-order $
z_1 < z_2 < \cdots < z_n,
$
and define the gaps
$
\Delta_i := z_{i+1}-z_i
$.
For fixed $\theta$, the maximizing density corresponds to the $z_i$ closest to $\theta$.  
Thus $\mathbb{R}$ is partitioned into Voronoi intervals:
\begin{align}
\Theta_1 &= \big(-\infty,\tfrac{z_1+z_2}{2}\big],\\
\Theta_i &= \big[\tfrac{z_{i-1}+z_i}{2}, \tfrac{z_i+z_{i+1}}{2}\big], \quad 2\le i\le n-1,\\
\Theta_n &= \big[\tfrac{z_{n-1}+z_n}{2},\infty).
\end{align}
On $\Theta_i$, the maximizer is $z_i$. 
Let $\Phi$ denote the standard normal CDF and $\varphi$ its density function. Then, for $i=1$
\begin{align*}
\Pr_{\M}(\Theta_1\mid z_1)=\int_{\Theta_1} \varphi_\sigma(\theta-z_1)\,\diff\theta
&= \Phi\!\Big(\tfrac{(z_1+z_2)/2 - z_1}{\sigma}\Big)=  \Phi\!\Big(\tfrac{\Delta_1}{2\sigma}\Big).
\end{align*}
For $i=m$
\begin{align*}
\Pr_{\M}(\Theta_m\mid z_m)=\int_{\Theta_m} \varphi_\sigma(\theta-z_m)\,\diff\theta
&= 1 - \Phi\!\Big(\tfrac{(z_{m-1}+z_m)/2 - z_m}{\sigma}\Big) =\Phi\!\Big(\tfrac{\Delta_{n-1}}{2\sigma}\Big).
\end{align*}
Finally, for $2\le i\le m-1$,
\begin{align*}
\Pr_{\M}(\Theta_i\mid z_i)=\int_{\Theta_i} \varphi_\sigma(\theta-z_i)\,\diff\theta
&= \Phi\!\Big(\tfrac{(z_i+z_{i+1})/2 - z_i}{\sigma}\Big)
   - \Phi\!\Big(\tfrac{(z_{i-1}+z_i)/2 - z_i}{\sigma}\Big)\\
&= \Phi\!\Big(\tfrac{\Delta_i}{2\sigma}\Big)-\Phi\!\Big(-\tfrac{\Delta_{i-1}}{2\sigma}\Big)\\[6pt]
&= \Phi\!\Big(\tfrac{\Delta_i}{2\sigma}\Big)+\Phi\!\Big(\tfrac{\Delta_{i-1}}{2\sigma}\Big)-1,
\end{align*}
using $\Phi(-z)=1-\Phi(z)$. Therefore
\begin{align*}
\sum_{i=1}^{m}\Pr_{\M}(\Theta_{i}\mid z_i)&= \Phi\!\Big(\tfrac{\Delta_1}{2\sigma}\Big)
  + \sum_{i=2}^{m-1}\Big(\Phi\!\Big(\tfrac{\Delta_i}{2\sigma}\Big)+\Phi\!\Big(\tfrac{\Delta_{i-1}}{2\sigma}\Big)-1\Big)
  + \Phi\!\Big(\tfrac{\Delta_{m-1}}{2\sigma}\Big)\\
  &=2\sum_{j=1}^{m-1} \Phi\!\Big(\frac{\Delta_j}{2\sigma}\Big)\;-\;(m-2),
\end{align*}
since each $\Delta_j$ appears exactly twice in the sum (once from its left neighbor, once from its right). Hence,
\begin{align}
    0\text{-RAD}&\leq \frac{2}{m}\sum_{j=1}^{m-1} \Phi\!\Big(\frac{\Delta_j}{2\sigma}\Big)\;-\frac{m-1}{m}\\
    &\leq \frac{2(m-1)}{m}\Phi\!\Big(\frac{1}{(m-1)}\sum_{j=1}^{m-1}\frac{\Delta_j}{2\sigma}\Big)\;-\frac{m-1}{m}\label{eq:jensen_normal}\\
    &\leq \frac{m-1}{m} \left(2\Phi\!\Big(\frac{1}{2\sigma(m-1)}\Big)\;-1\right) \label{eq:sensitivity}. 
\end{align}
Where, \Cref{eq:jensen_normal} follows since $\Delta_j\geq 0$, hence $\Phi$ concave, and we can apply Jensen's inequality, and~\Cref{eq:sensitivity} since $\Delta q=1$ therefore, $\sum_{j=1}^{m-1}\Delta_j=\Delta q=1$.
\end{example}
\begin{example}[Laplace Mechanism and $aux=\{\varnothing\}$] The Laplace mechanism adds Laplace noise with scale $b=\Delta q/\varepsilon$ to the query value $q(D)\in\R$~\cite{Dwork2014Algorithmic}.
If $\Z=\{z_1,\dots,z_m\}$ if uniformly distributed and $\Delta q=1$, analogously to~\Cref{ex:gauss},
\begin{equation*}
     z\in \argmax_{j}w(\theta,z_j)\pi_j \Leftrightarrow z\in \argmax_{j}w(\theta,z_j).
\end{equation*}
Hence, applying \Cref{eq:no_auz_uni} we obtain 
\begin{gather}
    0\text{-RAD}\leq \frac{1
    }{m}\sum_{i=1}^{m} \left(\Pr_{\M}(\Theta_{i}\mid z_i)-p_{\M}(\Theta_{i})\right)=\frac{1
    }{m}\sum_{i=1}^{m} \left(\Pr_{\M}(\Theta_{i}\mid z_i)-1\right)
\end{gather}
Analogously to the Gaussian case, we use the abuse of notation $z\equiv q(D_{z})$ . We want to compute $\Pr_{\M}(\Theta_{i}\mid z_i)$ for $i\in[m]$. Without loss of generality we re-order $
z_1 < z_2 < \cdots < z_n,
$
and define the gaps
$
\Delta_i := z_{i+1}-z_i
$.
For fixed $\theta$, the maximizing density corresponds to the $z_i$ closest to $\theta$.  
Thus $\mathbb{R}$ is again partitioned into Voronoi intervals from~\Cref{ex:gauss}.  Given the Laplace distribution CDF
\begin{align}F_i(z)&=\begin{cases}{\frac {1}{2}}\exp \left({\frac {z-z_i }{b}}\right)&{\mbox{if }}z<z_i \\1-{\frac {1}{2}}\exp \left(-{\frac {z-z_i }{b}}\right)&{\mbox{if }}z\geq z_i\end{cases},
\end{align}
for $i=1$,
\begin{align*}
\Pr_{\M}(\Theta_1\mid z_1)=F(\frac{z_1+z_2}{2})-F(-\infty)
= 1-{\frac {1}{2}}\exp \left(-{\varepsilon \frac{\Delta_1}{2} }\right),
\end{align*}
for $i=m$,
\begin{align*}
\Pr_{\M}(\Theta_m\mid z_m)=1-F(\frac{z_m+z_{m-1}}{2})
= 1-{\frac {1}{2}}\exp \left(-{\varepsilon \frac{\Delta_{m-1}}{2} }\right),
\end{align*}
and for the reminder $2\leq i<m$:
\begin{align*}
\Pr_{\M}(\Theta_m\mid z_m)=&F\left(\frac{z_i+z_{i+1}}{2}\right)-F\left(\frac{z_{i-1}+z_{i}}{2}\right)\\
=& 1-\frac {1}{2}\exp \left(-{\varepsilon \frac{\Delta_i}{2} }\right)+{\frac {1}{2}}\exp \left(-{\varepsilon \frac{\Delta_{i-1}}{2} }\right).
\end{align*}
Hence,
\begin{gather*}
    \sum_{i=1}^{m}\Pr_{\M}(\Theta_{i}\mid z_i)= m -\frac{1}{2}\left( \e^{-\frac{\varepsilon\Delta_1}{2}}+\e^{-\frac{\varepsilon\Delta_{m-1}}{2}}+\sum_{i=2}^{m-1}\e^{-\frac{\varepsilon\Delta_i}{2}}+\e^{-\frac{\varepsilon\Delta_{i-1}}{2}}
    \right)
    =m-\sum_{j=1}^{m-1}\e^{-\frac{\varepsilon\Delta_{j}}{2}}
\end{gather*}
since each $\Delta_j$ appears exactly twice in the sum (once from its left neighbor, once from its right). Hence,
\begin{align}
    0\text{-RAD}\leq& \frac{1}{m}\left(m-1-\sum_{j=1}^{m-1}\e^{-\frac{\varepsilon\Delta_{j}}{2}}\right)\\
    \leq& \frac{m-1}{m}-\frac{1}{m}\sum_{j=1}^{m-1}\e^{-\frac{\varepsilon\Delta_{j}}{2}}\\
    \leq&\frac{m-1}{m}-\frac{m-1}{m}\e^{-\frac{1}{m-1}\sum_{j}\frac{\varepsilon\Delta_{j}}{2}}\label{eq:jensen_laplace}\\
    &\leq\frac{m-1}{m}\left(1-\e^{-\frac{\varepsilon}{2(m-1)}}\right) \label{eq:sensitivity_lap}. 
\end{align}
Where, \Cref{eq:jensen_laplace} follows since $\Delta_j\geq 0$, hence we can apply Jensen's inequality, and~\Cref{eq:sensitivity_lap} since $\Delta q=1$ therefore, $\sum_{j=1}^{m-1}\Delta_j=\Delta q=1$.
\end{example}

\subsection{Numerical Approximation of Theorem~\ref{th:optimal_bound}}\label{ap:montecarlo}

We focus on approximating~\Cref{th:optimal_bound} in the continuous data domain $\Z = [a,b] \subseteq \R$. For simplicity, we consider the $\ell_1$ metric, i.e., $\ell(z,z') = |z-z'|$, so that
\[
S_{\eta}(z) = [\max(z-\eta, a), \min(z+\eta, b)] = [l(z), u(z)],
\]
and we set $aux = \{\varnothing\}$ (otherwise one could directly use the closed-form in~\Cref{th:dp_implies_aux-urero}).

We assume $\M \colon \Z \to \D(\Z)$ is a local differential privacy (LDP) mechanism. Given $D_{-}$, any global mechanism can be reduced to its local version $\M(z)=\left(\M(q(D_z)) - q(D_{-}\right)$, where $q$ is the query. Hence, our goal is to approximate
\begin{align*}
\eta\text{-}\mathrm{RAD} 
&\le \int_{a}^{b} \max_{x \in \mathcal Z} 
\int_{l(x)}^{u(x)} 
\big[p_{\mathcal M}(\theta \mid z) - p_{\mathcal M}(\theta)\big] \pi_z \, dz \, d\theta,
\end{align*}
where $\pi_z$ denotes the density function evaluated at $z$.

\textit{\textbf{Outer integral: Monte Carlo approximation.}}  
The integral with respect to $\theta$ can be approximated using Monte Carlo integration, which is unbiased and has a controlled error via Hoeffding’s inequality~\cite{hoeffding1963probability}. Hence, we can estimate RAD sampling $\theta_1, \dots, \theta_{N_\theta} \sim U[a,b]$ and computing
\[
\tilde{\gamma}=\frac{b-a}{N_\theta} \sum_{i=1}^{N_\theta} \underbrace{\max_{z_\theta \in \mathcal Z} \int_{l(x)}^{u(x)} 
\big[p_{\mathcal M}(\theta_i \mid z) - p_{\mathcal M}(\theta_i)\big] \pi_z \, dz}_{f(\theta_i)}.
\]
Given $\gamma=\eta$-RAD and $f(\theta)\in[-M,M]$ over $[a,b]$, applying Hoeffding’s inequality~\cite{hoeffding1963probability}, we have
\[
\Pr[|\tilde{\gamma}-\gamma| > t ] \le 2 \exp\Big(-\frac{ N_{\theta} t^2}{M^2(b-a)^2}\Big).
\]
Note that $f(\theta)\in[-M,M]$ trivially with $M=\max_{z,\theta}p_{\M}(\theta\mid z)$.
For instance, for the exponential mechanism with $u=-|z-\theta|$, we have
\[
M = \max_{z,\theta} p_{\M}(\theta \mid z) = \frac{1}{s(1-\exp(-1/s))}, \quad \text{where } s = \frac{2\Delta}{\varepsilon}.
\]

\textit{\textbf{Inner integral: Nested Monte Carlo approximation.} } 
If the integral with respect to $z$ also requires a numerical approximation, we proceed with a nested Monte Carlo approach. To avoid bias, we first find $z_{\theta}$ analytically and then perform a nested Monte Carlo procedure.

Denote by $z_\theta$ any point achieving the maximum:
\begin{align*}
 \max_{x\in \mathcal Z} \int_{l(x)}^{u(x)} 
\big[p_{\mathcal M}(\theta \mid z) - p_{\mathcal M}(\theta)\big] \pi_z \, dz \, d\theta. \\
\end{align*}

To compute the $z_{\theta}$ for each $\theta$, we use the Leibniz rule. For $\eta \le (b-a)/2$:
\[
S_{\eta}(x) = [l(x),u(x)]=
\begin{cases}
[a, x+\eta] & x \in [a, a+\eta],\\
[x-\eta, x+\eta] & x \in [a+\eta, b-\eta],\\
[x-\eta, b] & x \in [b-\eta, b].
\end{cases}
\]

Let $g_\theta(x) = \int_{l(x)}^{u(x)} \left(p_{\mathcal M}(\theta \mid z)-p_{\M}(\theta)\right)\pi_z \, dz$. Then
\begin{itemize}
    \item $x \in [a, a+\eta]$: $g_\theta'(x) =( p_{\mathcal M}(\theta \mid x+\eta)-p_{\M}(\theta))\, \pi(x+\eta)$,
    \item $x \in [a+\eta, b-\eta]$: $g_\theta'(x) =( p_{\mathcal M}(\theta \mid x+\eta)-p_{\M}(\theta))\pi(x+\eta) - (p_{\mathcal M}(\theta \mid x-\eta)-p_{\M}(\theta))\pi(x-\eta)$,
    \item $x \in [b-\eta, b]$: $g_\theta'(x) = -(p_{\mathcal M}(\theta \mid x-\eta)-p_{\M}(\theta))\pi(x-\eta)$.
\end{itemize}
The analogous follows for $(b-a)/2\leq \eta\leq b$, but in such case 
\[
S_{\eta}(x) = [l(x),u(x)]=
\begin{cases}
[a, x+\eta] & x \in [a, b-\eta],\\
[a, b] & x \in [b-\eta, a+\eta],\\
[x-\eta, b] & x \in [a+\eta, b].
\end{cases}
\]

Note that if $p_{\mathcal M}(\theta)=y$ is known, then $g'_\theta(x)$ can be evaluated explicitly. Consequently, $z_\theta$ can either be computed exactly or reduced to a finite set of candidate points. Since $g$ is continuous and $g'_\theta$ is explicit, this set contains at most $k \le 4$ candidates: the endpoints of the domain intervals and the critical points where $g'_\theta(x)=0$, with $g'_\theta(x)\le 0$ immediately before and $g'_\theta(x)\ge 0$ immediately after.

However, $p_{\M}(\theta)$ is itself an integral, $\E_{Z\sim\pi}[p_{\M}(\theta\mid Z)]$. In many cases, this integral may not have a close form, consequently, we also need to numerically approximate it. However, thanks to the Lipchitz properties of $g$ respect to $p_{\M}(\theta)$ we can upper-bound the error as we will see in~\Cref{eq:MC_g}.

First, we write 
\[
g(\theta,x,y)=\int_{a}^b \boldsymbol{1}_{\{|z-x|\leq\eta\}}(p_{\M}(\theta\mid z)-y)\pi_z\diff z.
\]
Note that $g$ is Lipchitz with respect to $y$:
\begin{gather}
    |g(\theta,x,\Tilde{y})-g(\theta,x,y)|=\int_{a}^b \boldsymbol{1}_{\{|z-x|\leq\eta\}}|\tilde{y}-y|\pi_z \diff z \\
    =|\tilde{y}-y|\int_{a}^b \boldsymbol{1}_{\{|z-x|\leq\eta\}}\pi_z \diff z \leq \kappa_{\pi,\eta}^{+}|\tilde{y}-y|.\label{eq:lipschitz}
\end{gather}
Importantly, it works uniformly for every $x\in[a,b]$. Hence, we can bound the error when estimating the maximum:
\begin{gather}
    |\max_{x}g(\theta,x,\tilde{y})-\max_{x}g(\theta,x,y)|\leq |\max_{x}(g(\theta,x,\tilde{y})-g(\theta,x,y))|\leq \kappa^{+}|\Tilde{y}-y|.
\end{gather}
Given $\theta$ and $y$ fixed, we denote $\mathcal{K}_{y,\theta}$ the candidates to maximum deduced using Leibniz. As mentioned $|\mathcal{K}|\leq 4$.

Now, we have three Monte Carlo procedures nested. From each of them we control de error, hence using the triangular inequality we can give a concentration upper-bound for the whole process.

We use the following notation (\Cref{tab:nested-mc}):

\begin{table}[H]
    \centering
    \resizebox{!}{0.12\linewidth}{
    \renewcommand{\arraystretch}{1.5}
    \begin{tabular}{p{0.5\linewidth} p{0.7\linewidth}}
        \toprule
        \textbf{Target} & \textbf{Estimator} \\
        \midrule
        $
        y = p(\theta)
        = \int_a^b p_{\mathcal{M}}(\theta \mid z)\,\pi_z\,\diff z
        $
        &
        $
        \hat{y} = \hat{p}(\theta)
        = \frac{1}{N_p} \sum_{i=1}^{N_p} p_{\mathcal{M}}(\theta \mid z_i),
        \quad
        z_1,\dots,z_{N_p} \sim \pi
        $
        \\[0.5em]
        $
        g(\theta,x,y)
        = \int_a^b
        \boldsymbol{1}_{\{|z-x|\le \eta\}}
        \bigl(p_{\mathcal{M}}(\theta \mid z) - y\bigr)\,
        \pi_z\,\diff z
        $
        &
        $
        \hat{g}(\theta,x,y)
        = \frac{1}{N_z} \sum_{i=1}^{N_z}
        \boldsymbol{1}_{\{|z_i-x|\le \eta\}}
        \bigl(p_{\mathcal{M}}(\theta \mid z_i) - y\bigr),
        $  $
        z_1,\dots,z_{N_z} \sim \pi
        $
        \\[0.5em]
        $
        f(\theta)
        = \max_{x \in \mathcal{K}_{\theta,p(\theta)}}
        g\bigl(\theta,x,p(\theta)\bigr)
        $
        &
        $
        \hat{f}(\theta)
        = \max_{x \in \mathcal{K}_{\theta,\hat p(\theta)}}
        \hat{g}\bigl(\theta,x,\hat{p}(\theta)\bigr)
        $
        \\[0.5em]
        $
        \gamma
        = \int_a^b f(\theta)\,\diff\theta
        $
        &
        $
        \tilde{\gamma}
        = \frac{b-a}{N_\theta}
        \sum_{i=1}^{N_\theta} \hat{f}(\theta_i),
        \quad
        \theta_1,\dots,\theta_{N_\theta} \sim \mathrm{U}(a,b)
        $
        \\
        \bottomrule
    \end{tabular}
    }
    \caption{Nested Monte Carlo estimation notation.}
    \label{tab:nested-mc}
\end{table}

For simplicity we consider $[a,b]=[0,1]$, and we use the union bound to simplify the concentration bound of our problem:

\begin{align}
\Pr\big(|\hat{\gamma}-\gamma|\ge t\big)
&=
\Pr\Bigg(
\Big|
\frac{1}{N_{\theta}}\sum_{i=1}^{N_{\theta}}\hat f(\theta_i)
-
\E_{\theta}[f(\theta)]
\Big|
\ge t
\Bigg) \nonumber\\
&=
\Pr\Bigg(
\Big|
\frac{1}{N_{\theta}}\sum_{i=1}^{N_{\theta}}
\max_{x\in\K_{i,\hat y}}\hat g(\theta_i,x,\hat y)
-
\E_{\theta}\big[\max_{x} g(\theta,x,y)\big]
\Big|
\ge t
\Bigg) \nonumber\\
&=
\Pr\Bigg(
\Big|
\frac{1}{N_{\theta}}\sum_{i=1}^{N_{\theta}}
\Big(
\max_{x\in\K_{i,\hat y}}\hat g(\theta_i,x,\hat y)
-
\max_{x\in\K_{i,y}} g(\theta_i,x,y)
\Big) \nonumber\\
&\hspace{1.5cm}
+
\frac{1}{N_{\theta}}\sum_{i=1}^{N_{\theta}} f(\theta_i)
-
\E_{\theta}[f(\theta)]
\Big|
\ge t
\Bigg) \nonumber\\
&\le
\Pr\Bigg(
\Big|
\frac{1}{N_{\theta}}\sum_{i=1}^{N_{\theta}} f(\theta_i)
-
\E_{\theta}[f(\theta)]
 \nonumber\\
&\hspace{1.5cm}
+
\frac{1}{N_{\theta}}\sum_{i=1}^{N_{\theta}}
\big(
\max_{x\in\K_{i,\hat y}}\hat g(\theta_i,x,\hat y)
-
\max_{x\in\K_{i, y}} g(\theta_i,x,y)
\big)
\Big|
\ge t
\Bigg)\notag\\
&\le
\Pr\Bigg(
\Big|
\frac{1}{N_{\theta}}\sum_{i=1}^{N_{\theta}} f(\theta_i)
-
\E_{\theta}[f(\theta)]
 \nonumber\\
&\hspace{1.5cm}
+
\frac{1}{N_{\theta}}\sum_{i=1}^{N_{\theta}}
\max_{x\in\K_{i,\hat y}}\hat g(\theta_i,x,\hat y)
-
\max_{x\in\K_{\theta,\hat y}} g(\theta,x,\hat y)
\notag\\
&\hspace{1.5cm}
+ \frac{1}{N_{\theta}}\sum_{i=1}^{N_{\theta}} \max_{x\in\K_{\theta,\hat y}} g(\theta,x,\hat y)
-
\max_{x\in\K_{\theta, y}} g(\theta,x,y)
\big)
\Big|
\ge t
\Bigg)\notag\\
\leq&
\Pr\Bigg(\Big|
\frac{1}{N_{\theta}}\sum_{i=1}^{N_{\theta}}f(\theta_i)
-
\E_{\theta}[f(\theta)]
\Big|\geq \frac{t}{3}\Bigg)\label{eq:MC_thetas}\\
&\hspace{1.5cm}
+
\Pr\Bigg(\Big|
\frac{1}{N_{\theta}}\sum_{i=1}^{N_{\theta}}
\max_{x\in\K_{i,\hat y}}\hat g(\theta_i,x,\hat y)
-
\max_{x\in\K_{\theta,\hat y}} g(\theta,x,\hat y)
\Big|\geq \frac{t}{3}\Bigg)\label{eq:MC_g}\\
&\hspace{1.5cm}
+ 
\Pr\Bigg(\Big|\frac{1}{N_{\theta}}\sum_{i=1}^{N_{\theta}} \max_{x\in\K_{\theta,\hat y}} g(\theta,x,\hat y)
-
\max_{x\in\K_{\theta, y}} g(\theta,x,y)
\big)
\Big|\geq \frac{t}{3}\Bigg)\label{eq:MC_y}
\end{align}

So we proceed to bound~\Cref{eq:MC_thetas,eq:MC_g,eq:MC_y} individually.

\textbf{\Cref{eq:MC_y}:} Assuming $\max_{z,\theta}p_{\M}(\theta\mid z)=M$, hence $|p(\theta)|\in[0,M]$ and using that $g$ is Lipchitz with respect to $y$, we obtain

\begin{align*}
\Pr\Big(\Big|\frac{1}{N_{\theta}}\sum_{i=1}^{N_{\theta}} \max_{x} g(\theta_i,x,\hat y)
-
\max_{x} g(\theta_i,x,y)
\Big|\geq \frac{t}{3}\Big)
&\leq 
 \Pr\Big(\frac{1}{N_{\theta}}\sum_{i=1}^{N_{\theta}} \kappa^{+}|\hat p(\theta_i)-p(\theta_i)|
\geq \frac{t}{3}\Big)\\
&\leq 
 \Pr\Big(\frac{1}{N_{\theta}}\sum_{i=1}^{N_{\theta}}|\hat p(\theta_i)-p(\theta_i)|
\geq \frac{t}{3\kappa^{+}}\Big)\\
\end{align*}
Given $Z_i=|\hat p(\theta_i)-p(\theta_i)|$, since $\hat p(\theta_i)$ is the Monte Carlo approximation of $p(\theta)$, we have
\begin{gather*}
    \E[Z_i]\leq \sqrt{\E(Z_i^2)}=\sqrt{\mathrm{Var}(\hat y)}\leq\frac{M}{2\sqrt{N}_p}
\end{gather*}
where last inequality follows from Popoviciu's inequality applied to $p(\theta_i)\in[0,M]$ and Monte Carlo estimation variance formula. Finally, applying Hoeffding's inequality for $Z_i$ we obtain:
\begin{gather}
\Pr\Big( \sum_{i=1}^{N_\theta} Z_i - \mathbb{E}[Z_i] \ge t \Big) 
\le \exp\Big(-\frac{2 t^2}{N_\theta}\Big) \\
\Pr\Big( \frac{1}{N_\theta} \sum_{i=1}^{N_\theta} Z_i - \frac{1}{N_\theta}\sum_{i=1}^{N_\theta} \mathbb{E}[Z_i] \ge t \Big) 
\le \exp\Big(-2 N_\theta t^2\Big) \\
\Pr\Big( \frac{1}{N_\theta} \sum_{i=1}^{N_\theta} Z_i \ge t + \mathbb{E}[ Z] \Big) 
\le \exp\Big(-2 N_\theta t^2\Big) \\
\Pr\Big( \frac{1}{N_\theta} \sum_{i=1}^{N_\theta} Z_i \ge t \Big) 
\le \exp\Big(-2 N_\theta \big(t - \mathbb{E}[Z]\big)^2\Big) 
\le \exp\Big(-2 N_\theta \big(t - \frac{M}{2 \sqrt{N_p}}\big)^2\Big).
\end{gather}
Summarizing,
\[
  \Pr\Big(\Big|\frac{1}{N_{\theta}}\sum_{i=1}^{N_{\theta}} \max_{x} g(\theta_i,x,\hat y)
-
\max_{x} g(\theta_i,x,y)
\Big|\geq \frac{t}{3}\Big)\le \e^{-2 N_\theta \big(\frac{t}{3\kappa^+} - \frac{M}{2 \sqrt{N_p}}\big)^2}
\]

\textbf{\Cref{eq:MC_g}:} Note that,
\begin{gather}
    \Pr\Bigg(
\Big|
\frac{1}{N_{\theta}}\sum_{i=1}^{N_{\theta}}
\Big(\max_{x\in\K_{i,\hat y}}\hat g(\theta_i,x,\hat y)
-
\max_{x\in\K_{i,\hat y}}g(\theta_i,x,\hat y)\Big)
\Big|
\ge \frac{t}{3}\Bigg)\\
\leq 
\Pr\Bigg(
\Big|
\frac{1}{N_{\theta}}\sum_{i=1}^{N_{\theta}}
\max_{x\in\K_{i,\hat y}}\Big(\hat g(\theta_i,x,\hat y)
-
g(\theta_i,x,\hat y)\Big)
\Big|
\ge \frac{t}{3}\Bigg)\\
\leq 
\Pr\Bigg(
\frac{1}{N_{\theta}}\sum_{i=1}^{N_{\theta}}
\Big|\max_{x\in\K_{i,\hat y}}\Big(\hat g(\theta_i,x,\hat y)
-
g(\theta_i,x,\hat y)\Big)
\Big|
\ge \frac{t}{3}\Bigg)
\end{gather}
If $k=1$, then the maximum is completely determined, hence we proceed analogously that for previous case, since $g\in [-M,M]$, $\E[|\hat g-g|]\leq \frac{M}{ \sqrt{N_{z}}}$:
\begin{align}
   \Pr\Bigg(
\Big|
\frac{1}{N_{\theta}}\sum_{i=1}^{N_{\theta}}
\Big(\max_{x\in\K_{i,\hat y}}\hat g(\theta_i,x,\hat y)
-
\max_{x\in\K_{i,\hat y}}g(\theta_i,x,\hat y)\Big)
\Big|
\ge \frac{t}{3}\Bigg)& \leq \e^{- \frac{N_{\theta}}{2M^2}\big(\frac{t}{3} - \frac{M}{ \sqrt{N_z}}\big)^2}.
\end{align}
However, if $k\geq 2$, we need to consider the maximum.
Note that $X_i=\hat g- g$ is a sub-Gaussian variable with $\sigma^2=\frac{M^2}{N_z}$. Hence, $\max_{i\in\K} |X_i|$ satisfies:
\[
\mathbb{E}[\max_{1\leq i\leq k} |X_i|] \leq \sqrt{2\sigma^2\ln(2k)} \leq \frac{M}{\sqrt{N_{z}}}\sqrt{2\ln(2k)}
\]
so applying Hoeffding,
\begin{gather}
    \Pr\left(
    \frac{1}{N_{\theta}}\bigg(\sum_i \max_{x}|X_x^i|-\sqrt{2\sigma^2\ln 2k}\bigg)\geq t\right)\leq \e^{-\frac{N_{\theta} t^2}{2 M^2}} \\
    \Leftrightarrow
    \Pr\left(\frac{1}{N_{\theta}}\sum_i \max_{x}|X_x^i|\geq t\right) \leq \e^{-\frac{N_{\theta}}{2M^2} \left(t-\frac{M}{\sqrt{N_z}}\sqrt{2\ln(2k)}\right)^2}. 
\end{gather}
Hence, 
\[
   \Pr\Bigg(
\Big|
\frac{1}{N_{\theta}}\sum_{i=1}^{N_{\theta}}
\Big(\max_{x\in\K_{i,\hat y}}\hat g(\theta_i,x,\hat y)
-
\max_{x\in \K_{i,\hat y}}g(\theta_i,x,\hat y)\Big)
\Big|
\ge \frac{t}{3}\Bigg)\leq  \e^{-\frac{N_{\theta}}{2M^2} \left(\frac{t}{3}-\frac{M}{\sqrt{N_z}}\sqrt{2\ln(2k)}\right)^2}
\]
\textbf{\Cref{eq:MC_thetas} } Follows directly from Hoeffding's inequality with  $f\in[-M,M
]$:
\begin{align}
\Pr\Bigg(
\Big|
\frac{1}{N_{\theta}}\sum_{i=1}^{N_{\theta}}
f(\theta_i)
-
\E_{\theta}\big[f(\theta)\big]
\Big| \geq& \frac{t}{3}\Bigg)\le  2\e^{-\frac{N_{\theta}}{2 M^2}\frac{t^2}{9}}
\end{align}

Putting all together we have the final error control bound:
\[
\Pr(|\hat \gamma-\gamma|\geq t)\leq 2\e^{-\frac{N_{\theta}}{2M^2}\frac{t^2}{9}}+\e^{-\frac{N_{\theta}}{2M^2}(\frac{t}{3}-\frac{M}{\sqrt{N_z}}\sqrt{2\ln(2k) })^2}+\e^{-\frac{N_{\theta}}{2M^2}(\frac{t}{3\kappa^+}-\frac{M}{2\sqrt{N_p}})^2}.
\]
with $k\leq 4$. We empirically evaluate this method for the exponential mechanism~\cite{Dwork2014Algorithmic} with $u=-|z-\theta|$ in $\Z=[0,1]$, hence,
\[
M\leq \frac{1}{s(1-\e^{-\frac{1}{s}})}\,\quad \text{where }s=\frac{2\Delta}{\varepsilon},
\]
ensuring the convergence of the numerical estimation. 

In particular, we test uniform and beta distributions for $\eta\in\{0.1,0.25,0.5,0.75\}$. We report the estimate and empirical confidence intervals with $J=500$ repetitions, since this ensures a convergence of the intervals for a tolerance of $\tau=1\time10^{-3}$ for all tested $\varepsilon$. We present the results in~\Cref{fig:MC_results,fig:MC_results_beta}, showing the estimated risk along with empirical $95\%$ confidence intervals. The figures illustrate a consistent pattern: the size of the confidence intervals decreases as the sample size increases. As expected, the variance grows with increasing 
$\varepsilon$, since the corresponding constant 
$\M$ also increases with $\varepsilon$.

\begin{figure*}[t] 
  \centering
  \begin{subfigure}[b]{0.2\textwidth}
    \centering
    \includegraphics[width=1.1\linewidth]{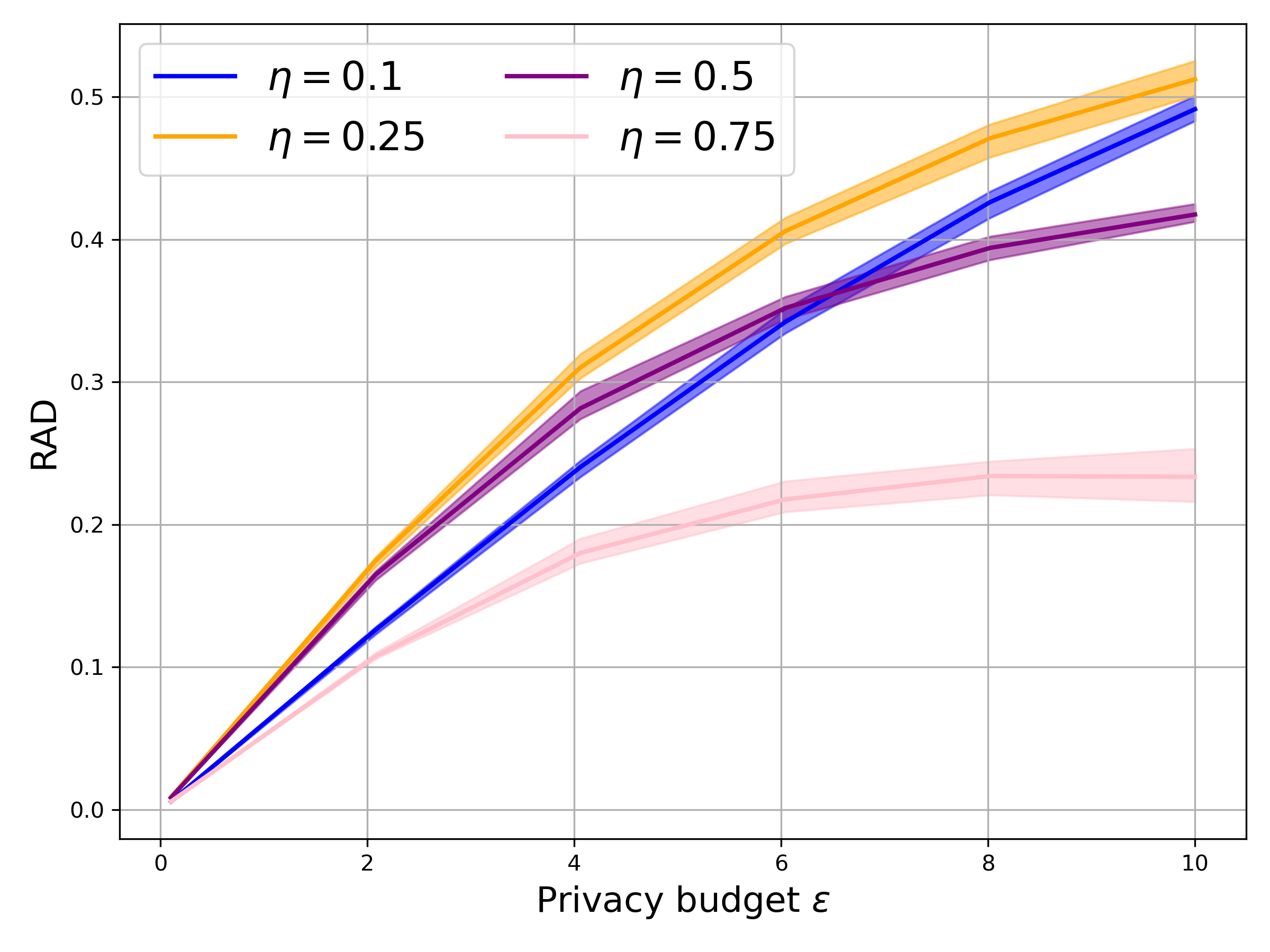}
    \caption{$N=500$.}
  \end{subfigure}\hfill
  \begin{subfigure}[b]{0.2\textwidth}
    \centering
    \includegraphics[width=1.1\linewidth]{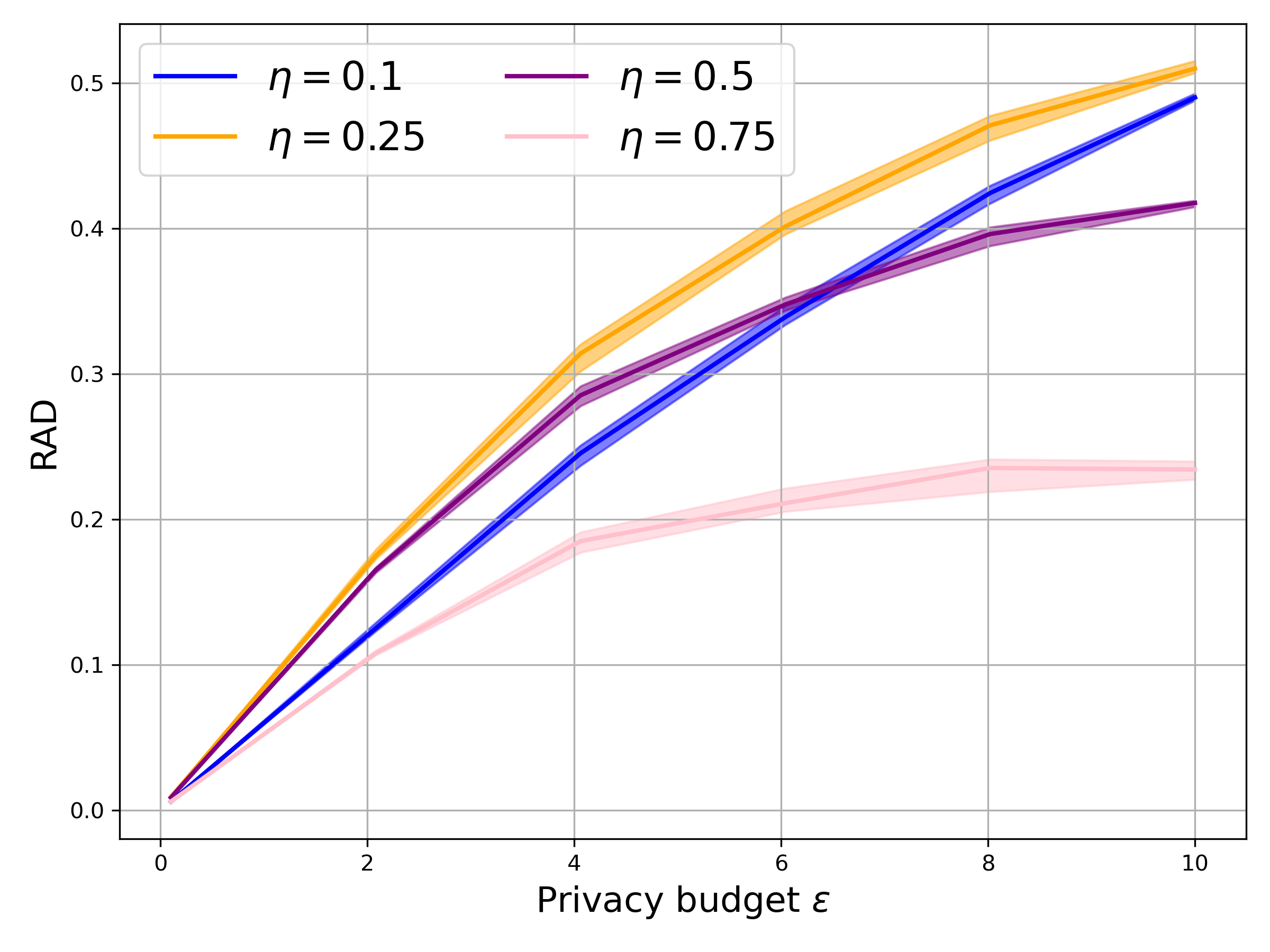}
    \caption{$N=1000$.}
  \end{subfigure}\hfill
  \begin{subfigure}[b]{0.2\textwidth}
    \centering
    \includegraphics[width=1.1\linewidth]{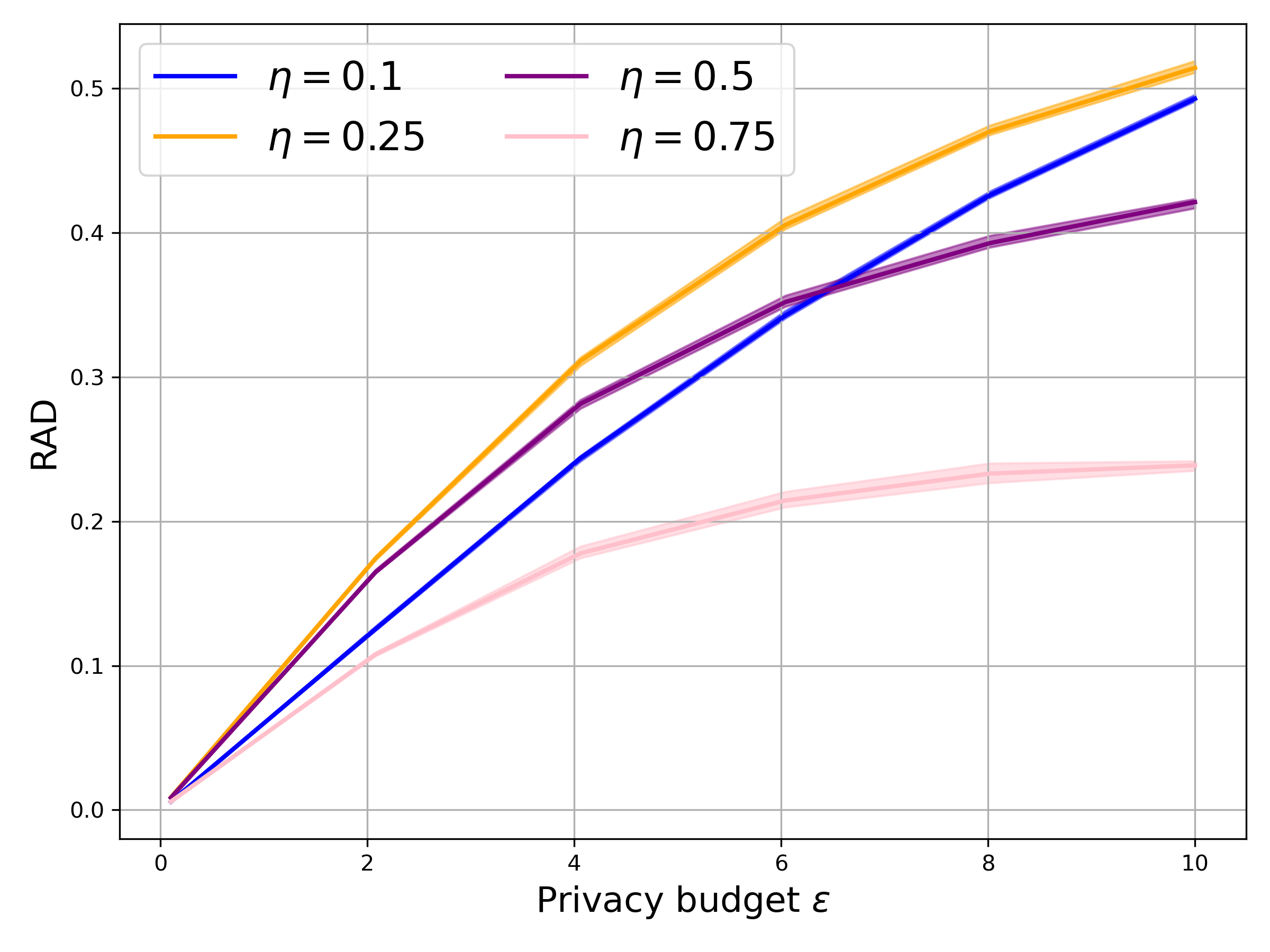}
    \caption{$N=2000$.}
  \end{subfigure}\hfill
    \begin{subfigure}[b]{0.2\textwidth}
    \centering
    \includegraphics[width=1.1\linewidth]{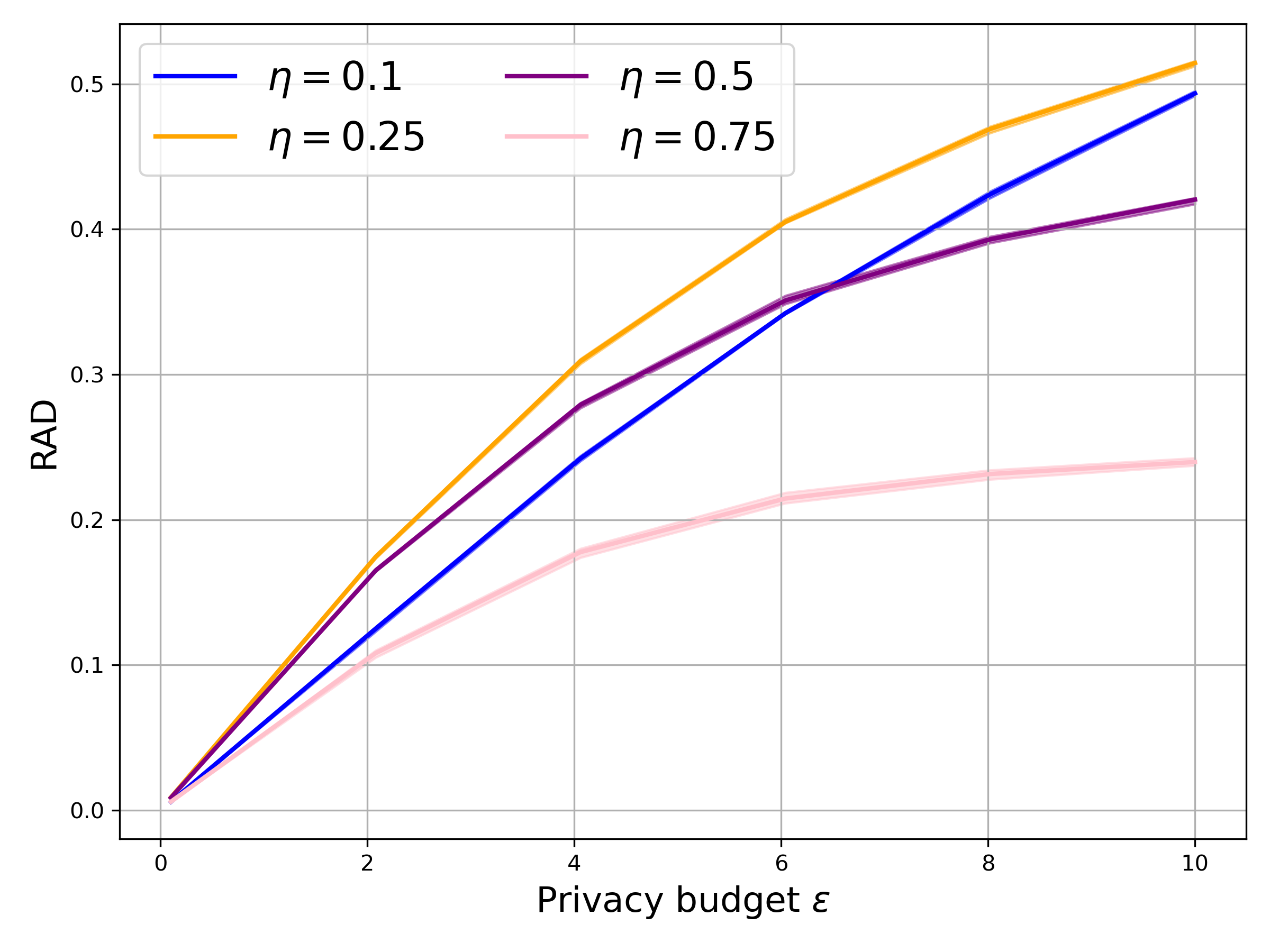}
    \caption{$N=5000$.}
  \end{subfigure}
  \caption{Numerical approximation of~\Cref{th:optimal_bound} with empirical $95\%$ confidence intervals for the exponential mechanism and $\pi=U(0,1)$ continuous and $N_p=N_{\theta}=N_z\equiv N$.}
  \label{fig:MC_results}
\end{figure*}
\begin{figure*}[t] 
  \centering
  \begin{subfigure}[b]{0.2\textwidth}
    \centering
    \includegraphics[width=1.1\linewidth]{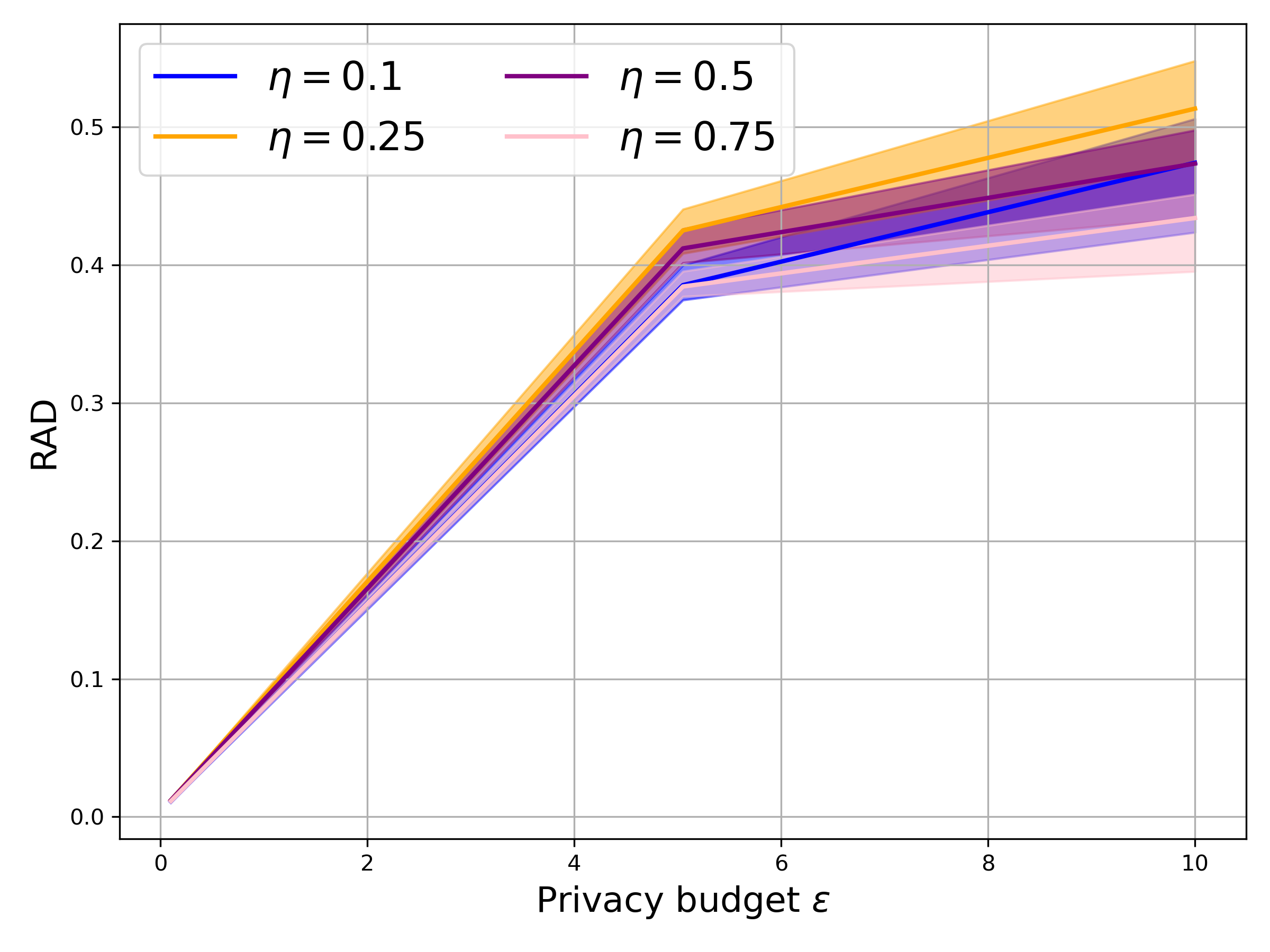}
    \caption{$N=500$.}
  \end{subfigure}\hfill
  \begin{subfigure}[b]{0.2\textwidth}
    \centering
    \includegraphics[width=1.1\linewidth]{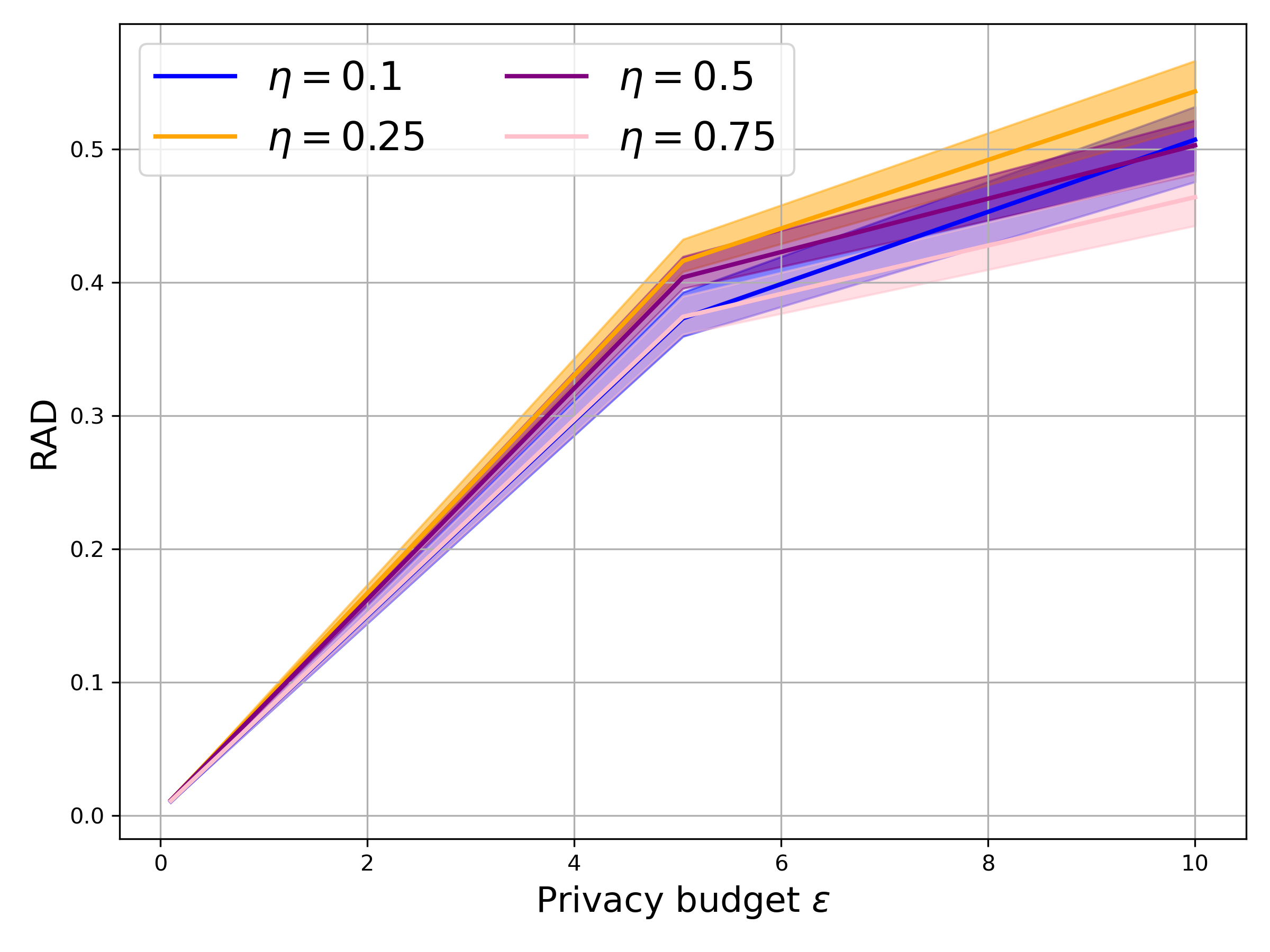}
    \caption{$N=1000$.}
\end{subfigure}\hfill
 \begin{subfigure}[b]{0.2\textwidth}
    \centering
    \includegraphics[width=1.1\linewidth]{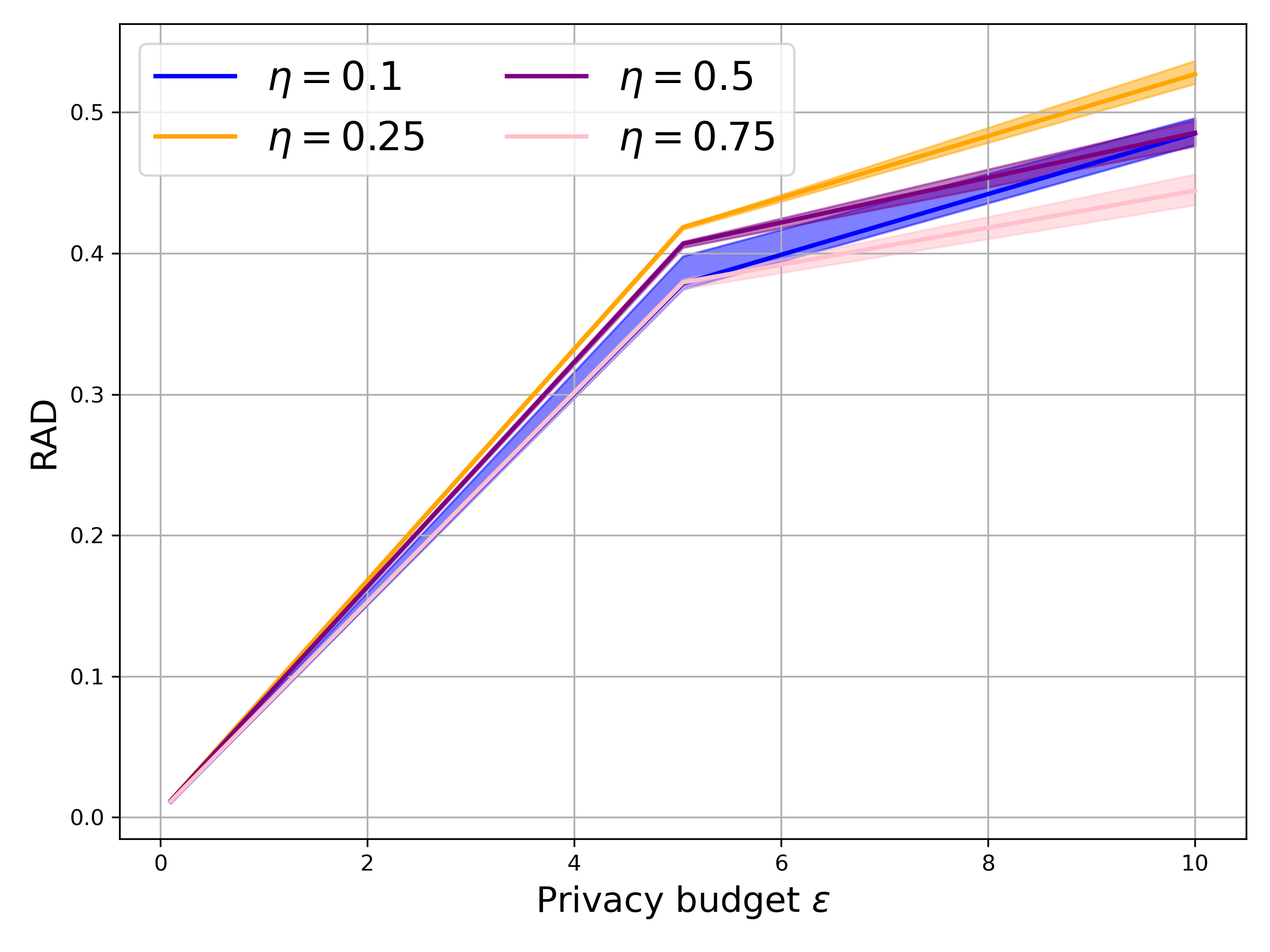}
    \caption{ $N=2000$.}
\end{subfigure}
\hfill
 \begin{subfigure}[b]{0.2\textwidth}
    \centering
    \includegraphics[width=1.1\linewidth]{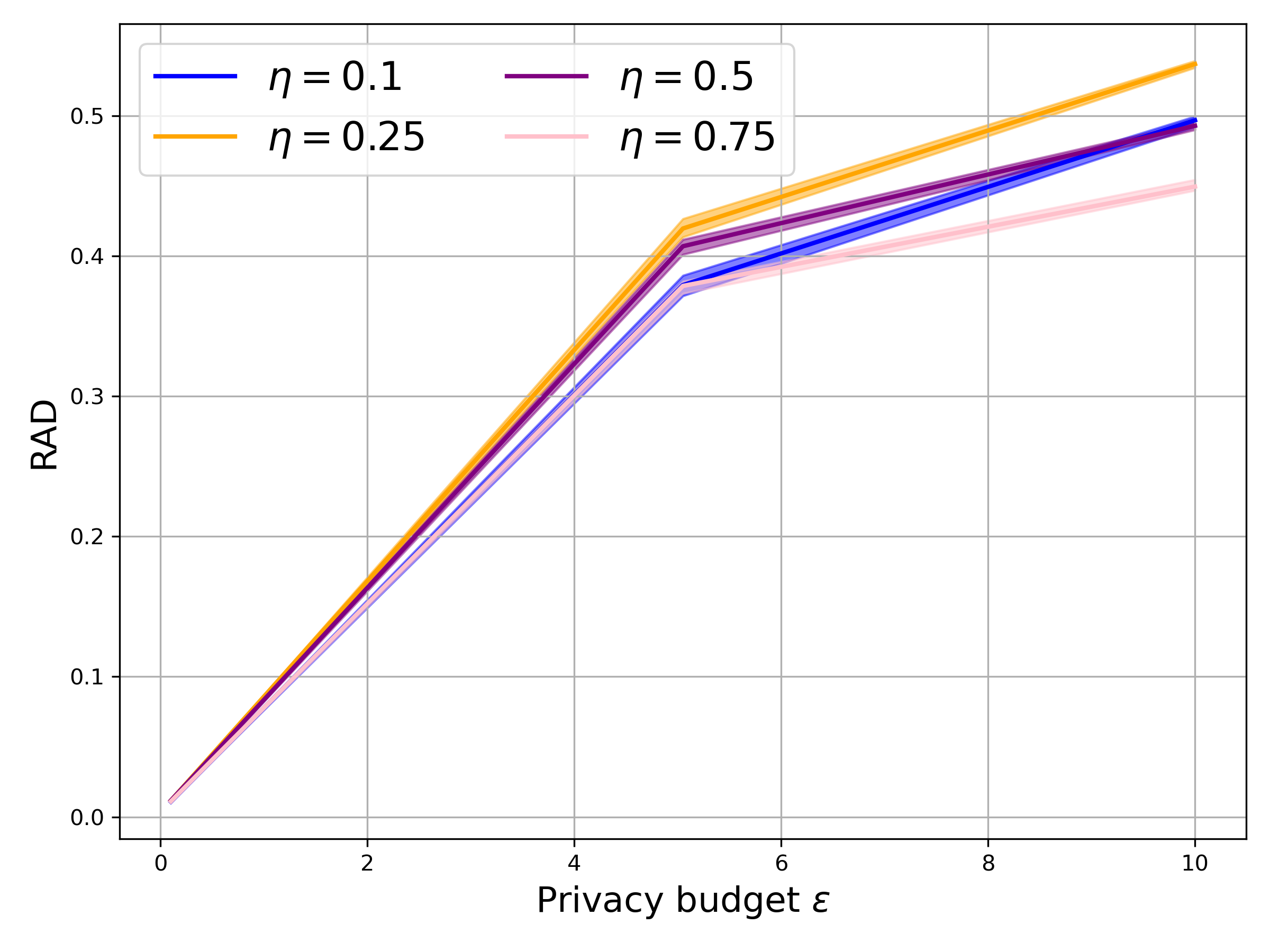}
    \caption{ $N=5000$.}
  \end{subfigure}
  \caption{Numerical approximation of~\Cref{th:optimal_bound} with empirical $95\%$ confidence for the exponential mechanism and $\pi=\mathrm{Beta}(0.1,0.1)$ continuous and $N_p=N_{\theta}=N_z\equiv N$.}
  \label{fig:MC_results_beta}
\end{figure*}

\end{document}